\documentclass[11pt,a4paper]{article}

\usepackage[utf8]{inputenc}
\usepackage[T1]{fontenc}
\usepackage{lmodern}
\usepackage{amsmath,amssymb,amsthm}
\usepackage[margin=1.1in]{geometry}
\usepackage{booktabs}
\usepackage{enumitem}
\usepackage{microtype}
\usepackage{xcolor}
\usepackage{tikz}\usetikzlibrary{arrows.meta,calc}
\usepackage[colorlinks=true,linkcolor=blue!50!black,citecolor=blue!50!black,urlcolor=blue!50!black]{hyperref}
\hypersetup{pdftitle={Amenable groups with nearly exponential sofic profile, and quantum channels that need nearly linear memory},
pdfauthor={Seth Douglas and Nidhal Mghirbi},pdfkeywords={sofic profile, amenable groups, Houghton groups, Thompson's groups,
wreath products, quantum channels, quantum memory}}

\theoremstyle{plain}
\newtheorem{theorem}{Theorem}[section]
\newtheorem{proposition}[theorem]{Proposition}
\newtheorem{corollary}[theorem]{Corollary}
\newtheorem{lemma}[theorem]{Lemma}
\theoremstyle{definition}

\theoremstyle{remark}
\newtheorem{remark}[theorem]{Remark}
\newtheorem*{theoremA}{Theorem A}
\newtheorem*{theoremB}{Theorem B}
\newtheorem*{theoremC}{Theorem C}

\DeclareMathOperator{\Tr}{Tr}
\DeclareMathOperator{\Sym}{Sym}
\DeclareMathOperator{\FSym}{FSym}
\DeclareMathOperator{\GL}{GL}
\DeclareMathOperator{\AGL}{AGL}
\DeclareMathOperator{\Area}{Area}
\DeclareMathOperator{\rank}{rank}
\newcommand{\ddia}{d_\diamond}
\newcommand{\cK}{\mathcal K}
\newcommand{\F}{\mathbb F}
\newcommand{\Z}{\mathbb Z}
\newcommand{\N}{\mathbb N}

\title{Amenable groups with nearly exponential sofic profile,\\ and quantum channels that need nearly linear memory}
\author{Seth Douglas \and Nidhal Mghirbi}
\date{October 2026}

\begin{document}
\maketitle

\begin{abstract}
How far from finite can an amenable group be? Every finite piece of an amenable group can be imitated by permutations of a finite set, and
Cornulier's sofic profile counts how many points such an imitation needs at accuracy $1/r$. We build a finitely presented elementary amenable
group with a finite piece whose profile is $\exp(r^{1+o(1)})$, close to the most our counting method can ever detect.

The construction is a lamplighter with its lamps on configurations rather than on points. Houghton's group moves the points of three rays,
finitely supported affine maps act on their binary configurations, and a copy of $S_3$ sits on every configuration. Then $N$ points carry
$2^N$ lamps, yet any two lamps can be brought together by moves that each involve at most three points. The group embeds in Brin's group
$3V$, so $3V$ too has a finite piece of nearly exponential profile.

The group also gives an explicit quantum channel on $\mathbb C^{873}$ that is cheap to use once and expensive to use many times. A device that
applies it $n$ times, releasing each output before the next input arrives, needs about $n/B$ qubits of memory when it may spend $B$ bits of
purity, and exact devices achieve this up to polylogarithmic factors. With exchange at the optimal rate, the least memory is about
$\sqrt n$, attained with purity of the same order. The channel is factorizable and lies in the closure of channels with finite maximally mixed baths, yet any such bath that imitates
it to accuracy $u$ needs dimension $\exp(u^{-1+o(1)})$, although its minimal Stinespring dilation implements it exactly with a pure environment of dimension $130$. The engine is a one-round theorem: after a single use, the memory of any such device carries an approximate representation of the
group.
\end{abstract}

{\small\tableofcontents}

\section{Introduction}\label{sec:intro}

Amenable groups are the infinite groups closest to finite ones. Every finite piece of an amenable group can be imitated by permutations of a
finite set, with products correct on all but a small fraction of the points. How many points does such an imitation need? For a finite piece
$E$ of a group, a \emph{chunk}, Cornulier's \emph{sofic profile} $\sigma_E(r)$ is the least number of points on which $E$ can be modelled with
multiplication correct, and distinct elements separated, on all but a fraction $1/r$ of the points~\cite{Cornulier2013}. He asked whether a
sofic group can have a profile that is unbounded and not linear in $r$~\cite[Problem~3.18]{Cornulier2013}.

Groups whose chunks all have finite profile are the sofic groups of Gromov and Weiss~\cite{Gromov1999,Weiss2000}; amenable and residually
finite groups are sofic. Whether every group is sofic was open until 2026, when OpenAI announced a proof that the unit group of the binary
Leavitt algebra is not sofic~\cite{OpenAI2026}. Fournier-Facio then gave a finitely presented torsion-free example~\cite{FournierFacio2026},
and Kun and Thom showed that generalized wreath products, with lamps $\Z/2$ over the cosets of a subgroup, are not sofic for suitable pairs of
residually finite groups with property~(T)~\cite{KunThom2026}. In a non-sofic group some chunk has infinite profile, so the question becomes how large a profile can be
when it is finite. Our group $\Gamma$ has the shape of those of Kun and Thom, lamps over a permutation action, but it is amenable, so all its
profiles are finite. Slofstra and Vidick tied the hyperlinear profile, its analogue for unitary matrices, to the entanglement needed to play
linear-system non-local games near-optimally~\cite{SlofstraVidick2018}. Slofstra then gave a finitely presented group with a chunk whose
profile is at least $\exp(r^{c})$ for some $c>0$, but his group is not known to be sofic, so this profile may be infinite~\cite{Slofstra2018}.

In earlier work we found a profile of this size in an amenable group, where every profile is finite. Houghton's group $H_3$ is elementary amenable and finitely
presented~\cite{Brown1987}, with an explicit presentation due to Johnson~\cite{Johnson1999}, and one of its chunks has profile at least
$\exp(\Omega(r^{0.129}))$~\cite{DouglasMghirbiA}. To our knowledge it was the first amenable group shown to have superpolynomial profile. This contradicts
Cornulier's extension theorem~\cite[Theorem~3.15]{Cornulier2013}, which would bound that profile by $O(r^2)$, and the conclusion drawn from
it in~\cite[Example~3.16]{Cornulier2013}, that every elementary amenable group has polynomial profile; Section~7 of~\cite{DouglasMghirbiA}
shows where the argument for that theorem breaks. The
proof is a counting criterion (Theorem~\ref{thm:criterion}). If a few relations force many displaced copies of $S_3$ to commute cheaply, every
model must be large, because nearly commuting copies of $S_3$ carry entropy. The criterion can never give more than $\exp(O(r))$. This paper
shows that an amenable group can come close to that ceiling.

\begin{theoremA}
There is a finitely presented elementary amenable group $\Gamma$, generated by five elements, with a chunk $E$ for which the following hold.
Put $\kappa=6+\log107<12.75$.
\begin{enumerate}[label=(\alph*),itemsep=1pt]
\item There are absolute constants $c,C>0$ such that for all $r\ge2$
\[
\exp\Bigl(\frac{c\,r}{(\log r)^{\kappa}}\Bigr)\ \le\ \sigma_E(r)\ \le\ \exp(C\,r\log r).
\]
So $\sigma_E(r)=\exp(r^{1+o(1)})$. Every chunk $E'$ of $\Gamma$ satisfies $\sigma_{E'}(r)\le\exp(C_{E'}\,r\log r)$ for all $r\ge2$. The lower
bound is certified by an explicit set of $25$ relations of $\Gamma$.
\item The same relations control unitary models. Two of the generators, $a$ and $c$, give $b=ac^4$, and $a,b$ generate a copy of $S_3$ in
$\Gamma$. Let $0<\Delta\le1$. There are $c_\Delta,e_\Delta>0$ with the following property. Let $\phi$ assign unitaries in $U(D)$ to the five
generators so that each of the $25$ relations is mapped within $e$ of $I$ in the normalized Hilbert--Schmidt norm, where $0<e\le e_\Delta$,
and put $\phi(b)=\phi(a)\phi(c)^4$. If $\|\phi(a)\phi(b)-I\|_2\ge\Delta$, then
\[
\log D\ \ge\ c_\Delta\,\frac{e^{-2}}{(\log(1/e))^{2\kappa}} .
\]
\end{enumerate}
\end{theoremA}

Theorem~\ref{thm:main-group} gives the precise statement, with the chunk $E$ and the $25$ relations of Table~\ref{tab:certificate}. The lower
bound is within a polylogarithmic factor in the exponent of the most the criterion can give. The exponent $\kappa$ records the price of bringing
two lamps together: about $N^2$ local moves, each paid for by the Dehn function of $H_3$, which grows at most like $n^{4+\log107}$~\cite[Theorems~3.2 and~3.3]{DouglasMghirbiA}. That
polynomial bound is the one input of Theorem~A whose proof we do not reproduce; with Lee's exponential bound alone the method gives only
$\exp(c\,r^{1/(1+B)})$ for a constant $B$ that we cannot make explicit (Remark~\ref{rem:lee}). The $107$ is
the factor by which a doubling endomorphism of $H_3$ multiplies area in~\cite[Section~3]{DouglasMghirbiA}; it is the only constant in $\kappa$
that is not structural, and a quadratic Dehn function for $H_3$ would give $\kappa=4$. The upper bound comes from models that hash
configurations by random affine maps. These are far smaller than the models that Følner sets give~\cite[Proposition~3.14]{Cornulier2013},~\cite[Section~4.3]{ArzhantsevaCherix2024},
because the Følner function of $\Gamma$ is doubly exponential (Proposition~\ref{prop:folner-lower}). The finite presentation
(Proposition~\ref{prop:fp}) comes from Cornulier's criterion for permutational wreath products~\cite[Theorem~1.1]{Cornulier2006}. Finally,
$\Gamma$ embeds in Brin's higher-dimensional Thompson group $3V$ (Proposition~\ref{prop:3V}). So $3V$, a finitely presented simple group whose
soficity is open, has a chunk whose profile is finite and $\exp(r^{1+o(1)})$.

\paragraph{Configuration lamps.} In Houghton's group the copies of $S_3$ sit on the points of a ray. Bringing two of them together costs a
word as long as the distance $k$ between them, and commuting them costs at least of order $k^2$ relations~\cite[Proposition~3.4]{DouglasMghirbiA}.
A family of $N$ copies spreads over distance $N$, so its commutators are expensive. We move the lamps off the points and onto the finitely
supported binary configurations of the three rays of $H_3$. Houghton's group moves the points. Finitely supported affine maps act on the
configurations: a translation flips one point, and a transvection adds one point into another. One copy of $S_3$, a lamp, sits on each
configuration (Figure~\ref{fig:lamps}). The $2^N$ configurations supported on the first $N$ points of a ray are reached by words of length
$O(N^2)$, and two of them differ by a configuration that at most $N$ transvections clear down to a single point. Each step touches at most
three points, so its cost is governed by the Dehn function of $H_3$ at scale $N$, which is polynomial~\cite[Theorems~3.2 and~3.3]{DouglasMghirbiA}.
Commuting two of the $2^N$ lamps therefore costs a polynomial in $N$, which is polylogarithmic in the number of lamps. That is exactly what the
criterion needs.

\begin{figure}[t]
\centering
\begin{tikzpicture}[x=1cm,y=1cm,font=\small,>=Stealth,baseline=(current bounding box.center),
  pt/.style={circle,draw,fill=white,inner sep=0pt,minimum size=5pt},
  on/.style={pt,fill=black},
  mv/.style={accent,thick,->,shorten >=1.5pt,shorten <=1.5pt},
  lbl/.style={font=\footnotesize,inner sep=1pt}]
\colorlet{accent}{blue!50!black}
\def\u{0.74}
\foreach \i/\ang in {1/0,2/120,3/240}{
  \draw[black!50] (0,0) -- (\ang:4.4*\u);
  \node at (\ang:4.8*\u) {$\cdots$};
  \foreach \j in {1,...,4}{\node[pt] (p\i\j) at (\ang:\j*\u) {};}
}
\node[lbl,anchor=west] at ($(120:4.8*\u)+(0.3,0)$) {ray $2$};
\node[lbl,anchor=west] at ($(240:4.8*\u)+(0.3,0)$) {ray $3$};
\node[lbl,anchor=north] at ($(0:4.8*\u)+(0.1,-0.2)$) {ray $1$};
\node[lbl,anchor=north] at ($(p11)+(-0.02,-0.2)$) {$1$};
\foreach \j in {2,3,4}{\node[lbl,anchor=north] at ($(p1\j)+(0,-0.16)$) {$\j$};}
\node[lbl,anchor=north east] at ($(p21)+(-0.08,-0.06)$) {$(2,1)$};
\node[lbl,anchor=south east] at ($(p31)+(-0.08,0.06)$) {$(3,1)$};
\node[on] at (p11) {};
\node[on] at (p13) {};
\node[on] at (p22) {};
\node[lbl,anchor=west] at (1.55,1.45) {$v=e_1+e_3+e_{(2,2)}$};
\foreach \j/\k in {4/3,3/2,2/1}{\draw[mv] (p2\j) -- (p2\k);}
\draw[mv] (p21) to[out=-15,in=170] (p11);
\foreach \j/\k in {1/2,2/3,3/4}{\draw[mv] (p1\j) -- (p1\k);}
\draw[mv] (p14) -- ($(0:4.55*\u)$);
\node[lbl,accent] at ($(120:2.5*\u)+(0.32,0.04)$) {$p$};
\foreach \j/\k in {4/3,3/2,2/1}{\draw[mv,dashed] (p3\j) -- (p3\k);}
\draw[mv,dashed] (p31) to[out=15,in=190] (p11);
\node[lbl,accent] at ($(240:2.5*\u)+(0.32,-0.04)$) {$q$};
\node[lbl,accent,anchor=west] at ($(0:4.8*\u)+(0.28,0)$) {$p,q$};
\draw[mv] (1.55,-0.95) -- ($(p11)+(0.18,-0.08)$);
\node[lbl,accent,anchor=north west] at (1.5,-0.95) {$t=\tau_1$: flip $1$};
\draw[mv] (p11) to[bend left=60] (p12);
\node[lbl,accent,anchor=west] at (1.35,0.62) {$d=E_{2\leftarrow1}$: flip $2$ if $1\in v$};
\coordinate (L) at (3.4,-2.15);
\draw ($(L)+(90:0.5)$) -- ($(L)+(210:0.5)$) -- ($(L)+(330:0.5)$) -- cycle;
\draw[accent,thick,<->,shorten >=4pt,shorten <=4pt] ($(L)+(90:0.5)$) -- ($(L)+(210:0.5)$);
\draw[accent,thick,<->,shorten >=4pt,shorten <=4pt] ($(L)+(210:0.5)$) -- ($(L)+(330:0.5)$);
\node[lbl,anchor=south] at ($(L)+(90:0.5)+(0,0.03)$) {$1$};
\node[lbl,anchor=north east] at ($(L)+(210:0.5)+(0.02,-0.02)$) {$2$};
\node[lbl,anchor=north west] at ($(L)+(330:0.5)+(-0.02,-0.02)$) {$3$};
\node[lbl,accent,anchor=east] at ($(L)+(150:0.3)+(-0.08,0)$) {$a$};
\node[lbl,accent,anchor=north] at ($(L)+(270:0.3)+(0,-0.08)$) {$b$};
\node[lbl,anchor=north,align=center] at ($(L)+(0,-0.72)$) {the lamp at $0$:\\ $a=(1\,2)$, $b=(2\,3)$};
\end{tikzpicture}%
\hspace{0.8cm}%
\begin{tikzpicture}[x=1cm,y=1cm,font=\small,>=Stealth,baseline=(current bounding box.center),
  lamp/.pic={\draw[fill=white] (90:0.16) -- (210:0.16) -- (330:0.16) -- cycle;},
  lbl/.style={font=\footnotesize,inner sep=1pt}]
\colorlet{accent}{blue!50!black}
\coordinate (v0) at (0,0);
\coordinate (v1) at (2.4,0);
\coordinate (v2) at (0,2.4);
\coordinate (v12) at (2.4,2.4);
\coordinate (v3) at (1.3,1.1);
\coordinate (v13) at (3.7,1.1);
\coordinate (v23) at (1.3,3.5);
\coordinate (v123) at (3.7,3.5);
\draw (v0) -- (v1) -- (v12) -- (v2) -- cycle;
\draw (v3) -- (v13) -- (v123) -- (v23) -- cycle;
\draw (v0) -- (v3) (v1) -- (v13) (v2) -- (v23) (v12) -- (v123);
\node[lbl,anchor=north] at ($(v0)!0.5!(v1)+(0,-0.08)$) {$t=\tau_1$};
\node[lbl,anchor=east] at ($(v0)!0.5!(v2)+(-0.08,0)$) {$\tau_2$};
\node[lbl,anchor=south east] at ($(v0)!0.55!(v3)+(-0.04,0.04)$) {$\tau_3$};
\draw[accent,thick,dashed,<->,shorten >=5pt,shorten <=5pt] (v1) to[bend right=38] (v12);
\draw[accent,thick,dashed,<->,shorten >=5pt,shorten <=5pt] (v13) to[bend right=38] (v123);
\node[lbl,accent,anchor=west] at (4.2,2.3) {$d=E_{2\leftarrow1}$};
\foreach \v in {v0,v1,v2,v12,v3,v13,v23,v123}{\pic at (\v) {lamp};}
\node[lbl,anchor=north east] at ($(v0)+(-0.1,-0.1)$) {$0$};
\node[lbl,anchor=north] at ($(v1)+(0,-0.18)$) {$e_1$};
\node[lbl,anchor=east] at ($(v2)+(-0.18,0)$) {$e_2$};
\node[lbl,anchor=north east] at ($(v12)+(-0.14,-0.12)$) {$e_1{+}e_2$};
\node[lbl,anchor=north west] at ($(v3)+(0.12,-0.1)$) {$e_3$};
\node[lbl,anchor=north west] at ($(v13)+(0.1,-0.12)$) {$e_1{+}e_3$};
\node[lbl,anchor=south] at ($(v23)+(0,0.18)$) {$e_2{+}e_3$};
\node[lbl,anchor=south] at ($(v123)+(0,0.18)$) {$e_1{+}e_2{+}e_3$};
\end{tikzpicture}
\caption{Left: the points of the three rays, a configuration $v$ (filled points), the lamp at $0$, and the generators as moves: $p$ and $q$
move points, and $t$ and $d$ move configurations. Right: the $2^3$ configurations supported on the points $1,2,3$, each carrying a lamp, joined by the
translations $\tau_1,\tau_2,\tau_3$ and folded by the transvection $d=E_{2\leftarrow1}$, which exchanges $e_1$ with $e_1+e_2$ and $e_1+e_3$ with
$e_1+e_2+e_3$ and fixes the other four. $N$ points carry $2^N$ lamps, and any two are brought together by moves that each involve at most
three points.}
\label{fig:lamps}
\end{figure}
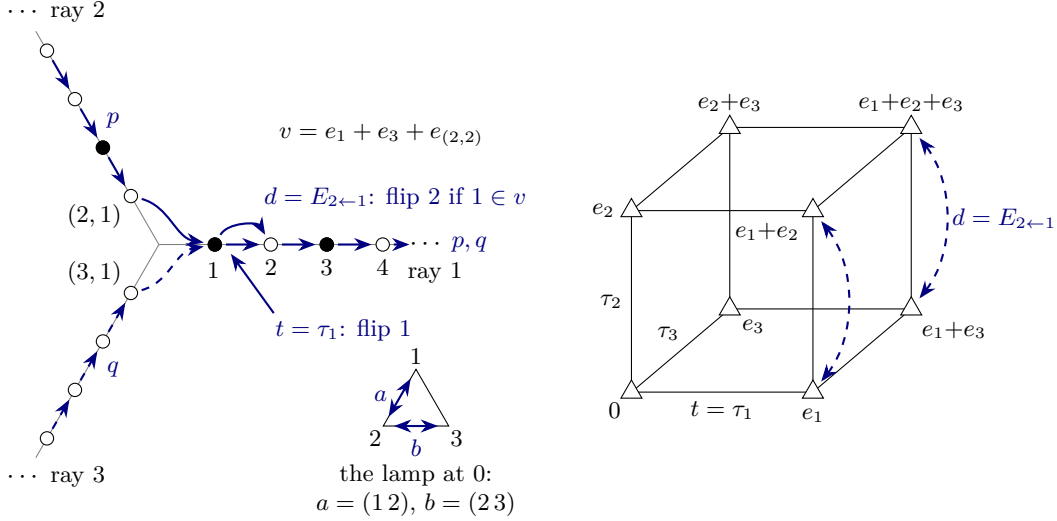

\paragraph{A ladder.} The configuration lamps are the last step of a ladder of transport mechanisms, each of which the criterion turns into a
profile bound.

\begin{center}
\begin{tabular}{llll}
\toprule
transport & group & $\log\sigma_E(r)\ \gtrsim$ & profile finite \\
\midrule
along a ray & $H_3$~\cite{DouglasMghirbiA} & $r^{0.129}$ & yes \\
by one element & dyadic $\mathcal G$ (\S\ref{sec:dyadic}) & $r^{1/3}$ & yes \\
exponential orbits & $\mathcal L\le V$ (\S\ref{sec:thompson}) & $r/\log^2r$ & if $F$ is amenable \\
configurations & $\Gamma$ (\S\S\ref{sec:group}--\ref{sec:profile}) & $r/(\log r)^{\kappa}$ & yes \\
\bottomrule
\end{tabular}
\end{center}

\noindent Transport by the powers of one element gives exponent $1/(p+1)$ when far commutators have area of order $k^p$. Along a ray of $H_3$
the area is at least quadratic~\cite[Proposition~3.4]{DouglasMghirbiA}, so this transport stops at $1/3$, which the dyadic group $\mathcal G$
attains with seventeen relations. Thompson's group $F$ transports cheaply along exponentially growing orbits, and a finitely presented
subgroup $\mathcal L$ of Thompson's group $V$ turns this into a dichotomy: either $V$ is not sofic, or its models of a fixed chunk need
$2^{c\,r/\log^2r}$ points. Models of $\mathcal L$ exist if $F$ is amenable, and $\mathcal L$ is hyperlinear exactly when $F$ is
(Proposition~\ref{prop:thompson-hyperlinear}). The configuration lamps keep the exponential transport and drop the condition on $F$.

\paragraph{Quantum memory.} The second half of the paper turns $\Gamma$ into a quantum channel. Picture a device that must apply the same
channel $n$ times. Each input arrives, the device acts on it together with its memory, and the output leaves before the next input comes in.
The device may swap memory qubits for fresh ones and may start with some purity, but it never resets, and every qubit it keeps counts. How much
memory does it need? This question goes back to Kretschmann and Werner, who showed that every causal quantum process is a concatenation of
memory channels~\cite{KretschmannWerner2005}, and to Rybár and Ziman, who asked when a device that is never reset applies the same channel to
each of its inputs. They showed that random unitary channels can be applied this way with a fixed finite memory, while nonunital channels
cannot~\cite{RybarZiman2008}. We let the memory grow with the number of uses and ask how fast it must grow. We set up the resource accounting in~\cite{DouglasMghirbiB}. The \emph{exchange} counts the qubits
the device swaps out, and the \emph{purity} counts the qubits it starts with or imports, minus their entropy (an entropy deficit, not $\Tr\rho^2$). The error is
measured against observers who choose their inputs adaptively and keep quantum references, as for quantum strategies, combs and
process tensors~\cite{GutoskiWatrous2007,Chiribella2009,Pollock2018}. A finite set of relations of a
group defines a channel, its \emph{gadget channel}, whose memory cost is governed by how well the group can be imitated: small memory would
give small approximate representations. For a gadget channel in the closure of channels with finite tracial factorizations, the optimal
exchange is $\frac12\log r_*$ qubits per use, where $r_*$ is the Choi rank~\cite[Propositions~3.1 and~9.1]{DouglasMghirbiB}. We call it the
rank rate.

The quantum reverse Shannon theorem simulates $n$ uses of a channel as one block, acting on all inputs at once with free shared entanglement,
and charges classical communication~\cite{BDHSW2014,BCR2011}; our device sees one input at a time and pays in memory. Storage and retrieval
keeps an unknown unitary in a quantum memory for later use~\cite{SBZ2019}; our channel is known, and memory is needed because it is not
unitary. Memory costs are known for classical simulations of sequences of quantum measurements~\cite{KGPLC2011,BFK2019}
and for quantum models of stochastic processes and of adaptive agents~\cite{GWRV2012,EGGT2022}; their targets have temporal correlations,
while $n$ independent uses of one channel have none. Our devices have the form of collision models~\cite{CLGP2022}: the memory plays the
system, and the inputs and imported qubits are units that each meet it once. Collision models usually fix the collisions and derive the
dynamics of the system; we fix the channel that the inputs see and ask how large the system must be. Purity, counted as above, is the asymptotic
measure of the resource theory of purity, or informational nonequilibrium~\cite{HHO2003,GMNSY2015}: maximally mixed qubits are free, and its free
operations from a system to itself, the noisy operations, are the channels with a finite tracial factorization that appear in Theorem~C.

One use reveals nothing: the channel below lives on $\mathbb C^{873}$, and a single round is easy. The dimension $873$ is what the gadget
construction gives for the $29$ words that define the channel; we made no attempt to minimize it. The cost shows up only over many rounds.
From $H_3$ we obtained a channel that needs memory $(n/\log n)^{0.069}$ at logarithmic purity~\cite{DouglasMghirbiB}. The configuration lamps
give nearly linear memory, and a trade-off between memory and purity that is sharp up to polylogarithmic factors. All these bounds are
asymptotic: the polylogarithmic factors have exponent up to $2\kappa+1<26.5$, and the constants, which come from the Dehn function of $H_3$,
are not evaluated, so the lower bounds take effect only for very large $n$.

\begin{theoremB}
The $25$ relations of Theorem~A, with two new letters for $b=ac^4$ and $j=ab$ and four short words that relate them ($29$ words in seven
letters), define an explicit quantum channel $\Phi$ on $\mathbb C^{873}$.
It lies in the closure of channels with finite tracial factorizations, and its Choi rank is $r_*=130$. With $\kappa$ as in
Theorem~A, there are constants $c,c_0,K>0$ and $n_0$ such that the following hold for causal services of $n$ uses of $\Phi$ with error at most
$1/16$.
\begin{enumerate}[label=(\alph*),itemsep=1pt]
\item \emph{Memory against purity.} If $n\ge n_0$ and $\log n\le B\le c_0n$, every service with purity at most $B$ has memory at least
$c\,n/(B(\log n)^{2\kappa})$, at any exchange. For every $n\ge1$ and $0<B\le n$ there is an exact service, with error $0$, purity less than
$B$ and exchange $(n-1)Q$, whose memory is $Q\le K(n/B)\log(2n/B)$. In particular, services with purity $O(\log n)$ need memory
$n^{1-o(1)}$, while memory $Kn\log n$ suffices at zero purity.
\item \emph{The rank rate.} If $n\ge n_0$, every service with exchange at most $\frac n2\log r_*$ has memory at least
$c\sqrt n/(\log n)^{\kappa}$, whatever its purity. For every $n\ge2$ there is a service with exchange at most $\frac n2\log r_*$ whose memory
and purity are both at most $K\sqrt{n\log n}$.
\end{enumerate}
So memory and purity trade at a fixed rate: in the range of~(a), the least memory at purity $B$, times $B$, is $n^{1+o(1)}$. With exchange at
most the rank rate the least memory is $n^{1/2+o(1)}$, and it is attained with purity of the same order.
\end{theoremB}

Theorem~\ref{thm:main-channel} gives the precise statement, and Figure~\ref{fig:tradeoff} draws it. The engine is a one-round theorem
(Theorem~\ref{thm:weighted}) for every gadget channel whose relations contain those of $S_3$ and a name for $ab$: a single round of any device
reveals an approximate representation of the group on the device's memory, with relator defect of order the square root of the leakage and
entropy production of that round. In~\cite[Theorem~5.5 and Appendix~B]{DouglasMghirbiB} this was proved for one gadget, built from Houghton's
group. The upper bounds come from sofic models, which give devices directly (Proposition~\ref{prop:sofic-upper}). At most points of a model
the pure bath state reproduces the channel exactly, so a memory that is maximally mixed on those points serves the channel with no error at
all. At the rank rate a device cannot buy memory with purity, because the purity it can spend is at most linear in its memory
(Proposition~\ref{prop:purity-budget}).

\paragraph{Finite environments.} The same bounds say how far from finite the channel itself is. A channel is factorizable if it factors through a finite von Neumann algebra with a faithful trace~\cite{AD2006,HaagerupMusat2011},
and it has a finite tracial factorization if it is $X\mapsto(\mathrm{id}\otimes\tau_M)[U(X\otimes I)U^*]$ for a unitary $U$ on $\mathbb C^d\otimes\mathbb C^M$, that is, if
it can be implemented with a finite bath in its maximally mixed state. By a theorem of Haagerup and Musat, the closure of these channels
consists of the channels that factor through an ultrapower of the hyperfinite $\mathrm{II}_1$ factor, and whether it contains every
factorizable channel, in every dimension, is equivalent to the Connes embedding problem~\cite[Theorems~3.6 and~3.7]{HaagerupMusat2015}.
That problem has a negative answer~\cite{JNVWY2020}, so some factorizable channels lie outside the closure. Inside it, Musat and R{\o}rdam
gave channels that need an infinite-dimensional ancilla~\cite{MusatRordam2020}. Our channel measures how expensive the finite
approximations can be.

\begin{theoremC}
The channel $\Phi$ of Theorem~B is factorizable and lies in the closure of the channels with finite tracial factorizations, and its minimal
Stinespring dilation implements it exactly with a pure environment of dimension $130$. Yet for small $u$, every channel with a finite tracial factorization within
accuracy $u$ of $\Phi$, in Choi or in diamond distance, and every finite mixture of such channels, needs a bath of dimension at least
$\exp\bigl(c\,u^{-1}/(\log(1/u))^{2\kappa}\bigr)$, while dimension $\exp(C\,u^{-1}\log(1/u))$ suffices.
\end{theoremC}

Corollary~\ref{cor:profile} gives the precise statement. Channels in the closure with no finite tracial factorization were
known~\cite{MusatRordam2020}, and for the Houghton channel of~\cite[Theorem~D.1]{DouglasMghirbiB} the exponent of the profile lies between
$0.069$ and $1/3$. What is new here is the quantitative profile: for $\Phi$ the exponent is $1-o(1)$, close to the most the method can
certify.

\paragraph{Physics.} A local lattice system in $\nu$ dimensions whose bath starts maximally mixed produces, in time $T$, only channels that can
be served with memory polynomial in $T$ and $\log n$ (Proposition~\ref{prop:light-cone}). So it needs time at least about $\delta^{-1/\nu}$ to
produce $\Phi$ to accuracy $\delta$ (Corollary~\ref{cor:preparation}). Autonomous dynamics with a bath in its tracial state have long-time
limits that are self-adjoint (Proposition~\ref{prop:autonomous};~\cite[Remark~10.3]{JungeLeMerdyXu2006},~\cite[Remark~5.2]{HaagerupMusat2011}), yet such a limit can still be
expensive: a bounded time-independent Hamiltonian, not local, produces in the long run a channel that trades memory and purity at the same
rate as $\Phi$ (Proposition~\ref{prop:lazy}). Whether local autonomous dynamics can produce an expensive channel is open.

\paragraph{How to read this paper.} Part~I (Sections~\ref{sec:group}--\ref{sec:ladder}) proves Theorem~A. A group theorist can read
Sections~\ref{sec:criterion}--\ref{sec:profile} and stop: Sections~\ref{sec:more} and~\ref{sec:ladder} add the finite presentation, the
embedding in $3V$ and the ladder; the bounds in parts~(a) and~(b) of Theorem~A do not use them. Part~II (Sections~\ref{sec:channel}--\ref{sec:main-channel-proof})
proves Theorem~B and needs from Part~I only the definitions of Section~\ref{sec:criterion}, the statement of Theorem~A and the pair-area
theorem. Part~III treats local physics and lists open problems. Long proofs are in the appendices. Appendix~\ref{app:imported} reproduces the proofs of the two results from the companion papers on which
the lower bounds of Theorems~A and~B stand, and Table~\ref{tab:dependencies} in Appendix~\ref{app:provenance} lists everything else we use from them. The
certificate of $25$ relations is in the ancillary files, with programs that check it. Logarithms are to base two unless written $\ln$.
\section{Counting approximately commuting copies of \texorpdfstring{$S_3$}{S3}}\label{sec:criterion}

Every lower bound in this paper comes from one principle, proved in~\cite{DouglasMghirbiA}. If finitely many relations force many displaced
copies of $S_3$ to commute cheaply, every model of a fixed chunk must be large, because approximately commuting copies of $S_3$ carry
entropy. We fix the conventions used in both halves of the paper, state the principle in the form used here, with $a,b$ in place of the
letters $u,v$ of~\cite{DouglasMghirbiA}, and note the ceiling $\exp(O(r))$ that it can never pass. The scheme, many cheaply transported
copies of a finite group, relations whose cost is controlled by area, and a dimension bound, follows Slofstra~\cite[Sections~2--3]{Slofstra2018}.
The criterion needs only pairwise approximate commutation of displaced copies of $S_3$.

\paragraph{Conventions.} Functions are composed from right to left, groups act on the left, $[x,y]=xyx^{-1}y^{-1}$, and $\log$ is to base
two. For a finite set $S$ and a finite set $R$ of words in $S^{\pm1}$, $\Area_R(w)$ is the least number of conjugates of words of $R^{\pm1}$
whose product is $w$ in the free group on $S$, and $\infty$ if there is none. A \emph{chunk} of a group is a finite subset $E\ni1$ with the
partial multiplication $xy=z$ whenever $x,y,z\in E$ and $xy=z$ in the group. Its \emph{sofic profile} $\sigma_E(r)$ is the least $n$ for which
some $f:E\to\Sym_n$ with $f(1)=\mathrm{id}$ multiplies and separates correctly up to $1/r$: $d(f(x)f(y),f(xy))\le1/r$ whenever $xy$ is defined in
$E$, and $d(f(x),f(y))\ge1-1/r$ for distinct $x,y\in E$~\cite{Cornulier2013}. Here $d$ is the normalized Hamming distance, the fraction of points
on which two permutations differ. The normalized Hilbert--Schmidt norm of a $D\times D$ matrix is $\|X\|_2=(\Tr(X^*X)/D)^{1/2}$.

The criterion is Theorem~8.1 of~\cite{DouglasMghirbiA}, and Appendix~\ref{app:imported} reproduces its proof.

\begin{theorem}[Counting criterion]\label{thm:criterion}
Let $G$ be a group, $S\subset G$ a finite set containing $a,b$, and $R$ a finite set of words in $S^{\pm1}$, each representing $1$ in $G$, with
$a^2,b^2,(ab)^3\in R$ and $ab\ne1$ in $G$. Let $l$ be the largest length of a word in $R$, and let $E$ be the chunk consisting of $1$, $S^{\pm1}$,
the prefixes of the words of $R$ and of $ab$, and the inverses of these elements. Let $w_1,\dots,w_N$ be words in $S^{\pm1}$ and $\Lambda\ge1$ such
that
\[
\Area_R\bigl([w_ixw_i^{-1},w_jyw_j^{-1}]\bigr)\le\Lambda\qquad(i\ne j,\ x,y\in\{a,b\}).
\]
Then $\log\sigma_E(r)\ge N/16$ for every $r\ge10^8N\Lambda(2l-1)$. Moreover, for $0<\Delta\le1$ and
$K_\Delta=\sqrt{80000\,(1+\log(1/\Delta))}/\Delta$, every assignment $\phi$ of unitaries in $U(D)$ to $S$ under which every word of $R$ lies
within $e$ of $I$ in the normalized Hilbert--Schmidt norm and $\|\phi(ab)-I\|_2\ge\Delta$ satisfies $\log D\ge N\Delta^2/12$ whenever
$K_\Delta\sqrt N\,\Lambda e\le1$.
\end{theorem}

The proof in~\cite{DouglasMghirbiA} has four steps, and only the last is specific to $S_3$. Conjugating the images of $a$ and $b$ by the image
of $w_i$ gives $N$ copies of $S_3$, one for each word. Each copy satisfies the relations of $S_3$ up to the defect of a relator, because those
relations are conjugates of $a^2$, $b^2$ and $(ab)^3$, so it can be rounded to an exact representation of $S_3$. Its three-cycle is a conjugate
of the image of $ab$, so it is displaced from the identity as much as that image is. Two copies commute up to $\Lambda$ times the relator defect,
because each of their commutators is a product of at most $\Lambda$ conjugates of relators. Finally, exact copies of $S_3$ that nearly commute
and have displaced three-cycles cannot share a small space: they force a fixed amount of entropy per copy, at least $\Delta^2/12$ bits for
unitaries and $1/16$ bit for permutations~\cite[Section~4]{DouglasMghirbiA}. The chunk $E$ does not depend on $N$ or on the words $w_i$; only
the areas matter.

The theorem has a ceiling. Its hypotheses force $N\le r/10^8$ and $N\le(K_\Delta e)^{-2}$, so it never gives $\log\sigma_E(r)$ beyond order $r$,
nor $\log D$ beyond order $e^{-2}$. Below the ceiling, everything depends on how cheaply the group carries copies of $S_3$ to many places.
Transport by the powers of one element, with area of order $k^p$ at distance $k$, gives exponent $1/(p+1)$; this is Corollary~8.2
of~\cite{DouglasMghirbiA}.

\begin{corollary}[Transport by the powers of one element]\label{cor:cyclic}
Let $G$, $S$, $R$, $l$ and $E$ be as in Theorem~\ref{thm:criterion}, let $t\in S$, and suppose that
$\Area_R([x,t^kyt^{-k}])\le Ck^p$ for $x,y\in\{a,b\}$ and $k\ge1$, with $C\ge1$ and $p\ge0$. Then
\[
\log\sigma_E(r)\ \ge\ \frac1{16}\Big\lfloor\Big(\frac r{10^8C(2l-1)}\Big)^{1/(p+1)}\Big\rfloor\qquad(r>1),
\]
and every unitary assignment as in Theorem~\ref{thm:criterion} has $\log D\ge\frac{\Delta^2}{12}\bigl((CK_\Delta e)^{-2/(2p+1)}-1\bigr)$.
\end{corollary}

In Houghton's group $H_3$ the copies of $S_3$ along a ray have commutators of at least quadratic area~\cite[Proposition~3.4]{DouglasMghirbiA},
so transport along a ray never passes the exponent $1/3$. A group that does better must reach many sites cheaply.

One more fact of Cornulier's is used twice, in Sections~\ref{sec:more} and~\ref{sec:ladder}. The profile of a chunk is bounded exactly when the
chunk embeds in a finite group, and otherwise it is at least $r$~\cite[Fact~3.5]{Cornulier2013}, indeed at least
$2r$~\cite[Lemma~2.1]{DouglasMghirbiA}. A chunk with unbounded profile therefore certifies that its group is not locally embeddable into finite
groups, not LEF~\cite[Fact~3.8]{Cornulier2013}.

\part{The group (Theorem A)}
\section{Configuration lamps}\label{sec:group}

This section defines the group $\Gamma$ of Theorem~A, its five generators $p,q,t,c,a$ and the $25$ relations that certify its profile, and
states Theorem~A in its precise form. The idea is visible in Figure~\ref{fig:lamps}: the lamps, copies of $S_3$, sit on configurations rather than on
points, so $N$ points carry $2^N$ lamps while every generator touches at most three points. Two of the $2^N$ lamps are compared by about $N^2$ local moves,
each filled with the Dehn function of $H_3$ at scale $N$, so the area is polylogarithmic in the number of lamps. That is what the criterion
needs. Theorem~\ref{thm:criterion} turns $M$ copies of $S_3$ whose pairwise commutators have area at most $\Lambda$ into
$\log\sigma_E(r)\ge M/16$ at $r$ of order $M\Lambda$, and it approaches its ceiling $\exp(O(r))$ only when $\Lambda$ is tiny compared with $M$.

\label{sec:toward}Put differently, the criterion needs three things at once: exponentially many sites reachable by short words, pairwise
commutators of polylogarithmic area, and an amenable ambient group. The two groups of the ladder in Section~\ref{sec:ladder} each meet two of
the three. The dyadic group $\mathcal G$ is amenable and its orbit of $0$ grows exponentially, but we cannot commute two copies at points of
opposite parity cheaply. In the subgroup $\mathcal L$ of Thompson's group $V$ everything works except amenability, which is the question
whether $F$ is amenable. The configuration lamps meet all three by changing what a site is. The sites are the finitely supported binary
configurations of the points of Houghton's three rays, not the points themselves. Affine maps that touch a bounded number of coordinates at a
time carry the copies between sites. And the ambient group is locally finite-by-$\Z^2$.

\paragraph{Houghton's group.} Let $Y=\{1,2,3\}\times\{1,2,\dots\}$; the point $(i,j)$ lies on ray $i$ at height $j$. Houghton's group
$H=H_3$ consists of the permutations $g$ of $Y$ that are eventually translations on each ray: there are integers $m_1(g),m_2(g),m_3(g)$ with
$g(i,j)=(i,j+m_i(g))$ for all large $j$. The map $g\mapsto(-m_2(g),-m_3(g))$ is a surjective homomorphism $H\to\Z^2$ whose kernel is the group
$\FSym(Y)$ of finitary permutations~\cite{Houghton1978}. Let $p$ be the permutation with $p(1,j)=(1,j+1)$, $p(2,1)=(1,1)$ and
$p(2,j)=(2,j-1)$ for $j\ge2$, fixing ray~$3$, and let $q$ be defined in the same way with rays $2$ and $3$ exchanged. They generate $H$
(Appendix~\ref{sec:houghton}), and their images in $\Z^2$ are $(1,0)$ and $(0,1)$.

\paragraph{The group.} Let $V=\F_2^{(Y)}$ be the space of finitely supported functions $Y\to\F_2$, the \emph{configurations},\footnote{This $V$ is
unrelated to Thompson's group $V$ of Sections~\ref{sec:3V} and~\ref{sec:thompson}.} with basis $(e_y)_{y\in Y}$. The group $H$ acts on $V$ by permutation matrices, $g(e_y)=e_{g(y)}$, and we identify $H$ with its image in $\GL(V)$. Let
$\GL_{\rm fin}(Y)$ be the group of linear automorphisms of $V$ that fix all but finitely many basis vectors. It is the directed union of the finite
groups $\GL(\F_2^J)$ over finite $J\subset Y$, each extended by the identity on the other basis vectors. It is normalized by $H$, and
$H\cap\GL_{\rm fin}(Y)=\FSym(Y)$. Let $L=\GL_{\rm fin}(Y)\,H\le\GL(V)$. Write $\tau_u$ for the translation $v\mapsto v+u$ of $V$, and let
$A=V\rtimes L$ be the group of affine maps $v\mapsto g(v)+u$ with $g\in L$ and $u\in V$. Our group is
\[
\Gamma=\Bigl(\bigoplus_{v\in V}S_3\Bigr)\rtimes A,
\]
where $A$ permutes the summands through its action on $V$. The copy of $S_3$ at $v$ is the \emph{lamp} at $v$; writing $\lambda_v$ for the copy
at $v$ of $\lambda\in S_3$, we have $g\lambda_vg^{-1}=\lambda_{g(v)}$ for $g\in A$. Thus $\Gamma$ is a lamplighter group. The lamps are copies of
$S_3$, one at each finite configuration of lit points of the three rays. The lighter is the affine group $A$, which moves configurations by
permuting points, toggling a point, or adding one point into another. Equivalently, $\Gamma$ is the group of permutations of $V\times S_3$
generated by the maps $(v,s)\mapsto(g(v),s)$ for $g\in A$ and $(v,s)\mapsto(v,\lambda(v)s)$ for finitely supported $\lambda:V\to S_3$. This
action is faithful. Groups obtained from a permutation action by adjoining the finitary linear or symmetric groups of the permuted set, and lamps
over it, appear as enrichments~\cite{Bradford2022,AlekseevBradford2026} and as halo products~\cite{GenevoisTessera2024}.

For distinct points $y,y'\in Y$ let $E_{y\leftarrow y'}$ be the transvection that sends $e_{y'}$ to $e_{y'}+e_y$ and fixes the other basis vectors.
For $g\in H$ and a linear map $M$ of $V$,
\begin{equation}\label{eq:affine-identities}
g\tau_{e_y}g^{-1}=\tau_{e_{g(y)}},\qquad gE_{y\leftarrow y'}g^{-1}=E_{g(y)\leftarrow g(y')},\qquad M\tau_uM^{-1}=\tau_{M(u)} .
\end{equation}
Linear maps fix $0$, so they commute with the lamp at $0$. On the first ray we abbreviate: for $x\ge1$ we write $x$ for the point $(1,x)$,
$e_x$ for $e_{(1,x)}$ and $\tau_x$ for $\tau_{e_x}$.

Let $a$ and $b$ be the transpositions $(1\,2)$ and $(2\,3)$ in the lamp at $0$, so that $ab$ has order three, and put
\[
t=\tau_1,\qquad d=E_{2\leftarrow1},\qquad c=d\,ab ;
\]
thus $t$ adds the point $1$ to a configuration, and $d$ adds the point $2$ to every configuration that contains the point $1$.

\begin{lemma}\label{lem:structure}
\begin{enumerate}[label=(\alph*),itemsep=1pt]
\item The element $c$ has order six, $c^3=d$ and $c^4=ab$; hence $b=ac^4$.
\item $\Gamma$ is generated by the five elements $p,q,t,c,a$.
\item Let $\pi:\Gamma\to\Z^2$ be the composite of $\Gamma\to A\to L\to L/\GL_{\rm fin}(Y)\cong H/\FSym(Y)\cong\Z^2$. Its kernel
$\bigl(\bigoplus_VS_3\bigr)\rtimes\bigl(V\rtimes\GL_{\rm fin}(Y)\bigr)$ is locally finite. Hence $\Gamma$ is locally finite-by-$\Z^2$, and in
particular elementary amenable.
\end{enumerate}
\end{lemma}

\begin{proof}
(a) The linear map $d$ fixes $0$, so it commutes with $ab$; since $d$ has order two and $ab$ order three, $c$ has order six, $c^3=d^3(ab)^3=d$ and
$c^4=d^4(ab)^4=ab$.

(b) Let $\Gamma'$ be the subgroup generated by the five elements. It contains $H$, because $H=\langle p,q\rangle$. Since $\FSym(Y)$ acts
$2$-transitively on $Y$, \eqref{eq:affine-identities} shows that $\Gamma'$ contains every $\tau_{e_y}$, hence $V$, and, since $d=c^3$, every
transvection $E_{y\leftarrow y'}$. Transvections generate every $\GL(\F_2^J)$, so $\GL_{\rm fin}(Y)\le\Gamma'$, and $A\le\Gamma'$. Finally
$b=ac^4\in\Gamma'$, and the lamp at $u\in V$ is $\tau_u\langle a,b\rangle\tau_u^{-1}$.

(c) Since $H\cap\GL_{\rm fin}(Y)=\FSym(Y)$, we have $L/\GL_{\rm fin}(Y)\cong H/\FSym(Y)$, and the kernel of $\pi$ is as stated. It is an
extension of $\bigoplus_VS_3$ by $V\rtimes\GL_{\rm fin}(Y)$, the directed union of the finite groups $\F_2^J\rtimes\GL(\F_2^J)$; both are locally
finite. An extension of a locally finite group $K$ by a locally finite group is locally finite: a finitely generated subgroup $\Lambda$ has finite
image modulo $K$, so $\Lambda\cap K$ has finite index in $\Lambda$, is finitely generated, and is finite. Locally finite groups are directed
unions of finite groups, and elementary amenable groups are closed under directed unions and extensions~\cite{Chou1980}.
\end{proof}

\subsection{A certificate of 25 relations}\label{sec:certificate}

All fillings in $\Gamma$ are built from one finite set $R$ of words in the five generators of Lemma~\ref{lem:structure}, each representing $1$
in $\Gamma$. The set $R$ does not present $\Gamma$ (Section~\ref{sec:fp}). Areas are always taken with respect to $R$.

\paragraph{Abbreviations.} Put
\[
s=[p,q],\qquad B=psp^{-1},\qquad t_m=p^{m-1}\,t\,p^{1-m}\quad(m=1,2,3),\qquad d=c^3,\qquad b=ac^4 .
\]
By Lemma~\ref{lem:transpositions}, $s$ and $B$ are the transpositions $(1\ 2)$ and $(2\ 3)$, and $t_m$ represents $\tau_m$, since
$p^{m-1}(1,1)=(1,m)$. By Lemma~\ref{lem:stabilizers}, with $\zeta_1=s$ and $\zeta_2=Bs$, the pairs
\[
K_1=(sp,\ sq),\qquad K_2=(Bsp,\ Bsq)
\]
generate the pointwise stabilizers $H^{(1)}$ of the point $1$ and $H^{(2)}$ of the points $1,2$. The elements $t$ and $d$ have the
\emph{prototypes} $P_t=\{1\}$ and $P_d=\{1,2\}$, the points on which they depend: by~\eqref{eq:affine-identities}, an element of $H$ that fixes
$P_f$ pointwise commutes with $f$.

\begin{table}[ht]
\centering\small
\begin{tabular}{@{}l p{0.36\textwidth} p{0.29\textwidth} r r@{}}
\toprule
 & relations & says & number & length\tabularnewline
\midrule
\textsf{H} & \raggedright $s^2$,\ $(sps\,p^{-1})^3$,\ $[s,p^2sp^{-2}]$,\ $psp^{-1}qs^{-1}q^{-1}$ & \raggedright the presentation~\eqref{eq:PH} of Houghton's group & 4 & 64\tabularnewline
\textsf{Q} & \raggedright $t^2$ & \raggedright a translation is an involution & 1 & 2\tabularnewline
\textsf{S} & \raggedright $[t,k]$ for $k\in K_1$;\ \ $[d,k]$ for $k\in K_2$ & \raggedright the stabilizer of a prototype commutes with it & 4 & 56\tabularnewline
\textsf{Z} & \raggedright $[g,x]$ for $g\in\{p,q\}$, $x\in\{a,b\}$;\ \ $[d,a]$ & \raggedright linear maps fix the lamp at $0$ & 5 & 40\tabularnewline
\textsf{C} & \raggedright $[t_1,t_2]$ & \raggedright translations commute & 1 & 8\tabularnewline
\textsf{A} & \raggedright $d\,t_1d^{-1}(t_1t_2)^{-1}$,\ $[d,t_2]$,\ $[d,t_3]$ & \raggedright how $d$ moves the translations at $1,2,3$ & 3 & 39\tabularnewline
\textsf{L} & \raggedright $a^2$,\ $b^2$,\ $(ab)^3$ & \raggedright the lamp at $0$ is $S_3$ & 3 & 30\tabularnewline
\textsf{P} & \raggedright $[x,tyt^{-1}]$ for $x,y\in\{a,b\}$ & \raggedright the lamps at $0$ and at $e_1$ commute & 4 & 64\tabularnewline
\midrule
 & total & & 25 & 303\tabularnewline
\bottomrule
\end{tabular}
\caption{The certificate $R$: eight families of relations in the five generators $p,q,t,c,a$, written with the abbreviations $s$, $B$, $t_m$,
$d$ and $b$. Each relation is taken to be the freely reduced word obtained by expanding the abbreviations; lengths are those of these words, and
the longest has length $24$.}
\label{tab:certificate}
\end{table}

\begin{lemma}\label{lem:validity}
Every word of Table~\ref{tab:certificate} represents~$1$ in $\Gamma$.
\end{lemma}

\begin{proof}
\textsf{H}: the words of~\eqref{eq:PH} hold in $H$.
\textsf{Q}: $V$ has exponent two.
\textsf{S}: the elements of $H^{(1)}$ fix the point $1$, so they commute with $t=\tau_1$; those of $H^{(2)}$ fix the points $1$ and $2$, so they
commute with $d=E_{2\leftarrow1}$ (Lemma~\ref{lem:structure}(a) and~\eqref{eq:affine-identities}).
\textsf{Z}: $p$, $q$ and $d$ are linear, so they commute with $a$ and with $b=ac^4=a\cdot ab$, which lie in the lamp at~$0$.
\textsf{C}: translations commute.
\textsf{A}: $dt_md^{-1}=\tau_{d(e_m)}$, and $d(e_1)=e_1+e_2$, while $d$ fixes $e_2$ and $e_3$.
\textsf{L}: these are relations of $S_3$.
\textsf{P}: $tyt^{-1}$ lies in the lamp at $e_1$, which commutes with the lamp at~$0$.
\end{proof}

A program in the ancillary files checks the $25$ words independently of Lemma~\ref{lem:validity}, in the faithful action of $\Gamma$ on
$V\times S_3$; Appendix~\ref{app:provenance} says what it checks.

\subsection{Theorem A}\label{sec:theorem-A}

\begin{theorem}[Theorem A]\label{thm:main-group}
Let $E\subset\Gamma$ be the chunk consisting of $1$, the five generators $p,q,t,c,a$ and their inverses, the elements represented by the prefixes
of the $25$ freely reduced words of Table~\ref{tab:certificate}, and the inverses of these elements. Let $p_H=\log107+4$ and
\[
\kappa=2+p_H=6+\log107<12.75 .
\]
There are absolute constants $c,C>0$ such that for every $r\ge2$:
\begin{enumerate}[label=(\alph*),itemsep=1pt]
\item $\log\sigma_E(r)\ge c\,r/(\log r)^\kappa$;
\item $\log\sigma_E(r)\le C\,r\log r$; more generally, every chunk $E'$ of $\Gamma$ satisfies $\log\sigma_{E'}(r)\le C_{E'}\,r\log r$.
\end{enumerate}
In particular $\sigma_E(r)=\exp(r^{1+o(1)})$. Moreover:
\begin{enumerate}[label=(\alph*),itemsep=1pt,resume]
\item For every $0<\Delta\le1$ there are $c_\Delta,e_\Delta>0$ such that the following holds for $0<e\le e_\Delta$. Let $\phi$ assign unitaries in
$U(D)$ to the five generators so that every word of Table~\ref{tab:certificate} is mapped within $e$ of $I$ in the normalized Hilbert--Schmidt
norm, and put $\phi(b)=\phi(a)\phi(c)^4$. If $\|\phi(a)\phi(b)-I\|_2\ge\Delta$, then
\[
\log D\ \ge\ c_\Delta\,\frac{e^{-2}}{(\log(1/e))^{2\kappa}},\qquad 2\kappa<25.49 .
\]
\item $\Gamma$ is finitely presented and elementary amenable.
\end{enumerate}
\end{theorem}

The exponent $\kappa=2+p_H$ has two parts. Two of the $2^N$ lamps are brought together by about $N^2$ local moves
(Theorem~\ref{thm:pair-area}), each priced by the Dehn function of $H_3$ at scale $N$, which is at most $CN^{p_H}$
(Proposition~\ref{prop:dehn}(b)). The $107$ is the factor by which the doubling endomorphism of $H_3$ in~\cite[Section~3]{DouglasMghirbiA}
multiplies a reweighted area at each step, so that far commutators at distance $M$ have area $O(M^{\log107})$~\cite[Theorem~3.2]{DouglasMghirbiA}.
The $4$ comes from reducing the Dehn function to far commutators~\cite[Theorem~3.3]{DouglasMghirbiA}. The $107$ is the only constant in
$\kappa$ that is not structural, and a quadratic Dehn function for $H_3$ would give $\kappa=4$.

The proof occupies Sections~\ref{sec:fillings}--\ref{sec:more}.
\begin{enumerate}[label=(\roman*),itemsep=1pt]
\item Section~\ref{sec:fillings}: every local identity among transported translations, transvections and lamps is filled at a cost $L(n)$
governed by $\delta_H$ at scale $n$.
\item Section~\ref{sec:pair-areas}: two of the $2^N$ lamps on the first $N$ points commute at the cost of about $N^2$ local fillings
(Theorem~\ref{thm:pair-area}).
\item Section~\ref{sec:profile}: feeding these pair areas into Theorem~\ref{thm:criterion} gives the lower bounds (a) and (c).
\item Section~\ref{sec:profile}: models that hash configurations by random affine maps give the upper bound (b).
\item Section~\ref{sec:more}: finite presentation, part (d), by Cornulier's criterion for permutational wreath products.
\end{enumerate}
\section{Local fillings}\label{sec:fillings}

To compare lamps we carry the two prototypes, the translation $t$ and the transvection $d$, around the first $n$ points of ray~$1$ by short
words in $p$ and $q$. This section shows that any two short words that agree on the prototype give the same conjugate, at the cost of two
fillings in $H_3$ at scale $n$. The output is one number, the $L(n)$ of Lemma~\ref{lem:local-filling}: the price of every local identity used
in Section~\ref{sec:pair-areas}.

\paragraph{Three facts about Houghton's group.} All transport takes place among the first points of ray~$1$. For $n\ge1$ let $\Sigma_n\le H$
be the group of permutations of $Y$ supported in $\{(1,1),\dots,(1,n)\}$, identified with $\Sym(n)$, and write $(x\ y)$ for the transposition of
the points $x$ and $y$ of the first ray. Let $\delta_H$ be the Dehn function of the presentation $P_H$ of $H$ in~\eqref{eq:PH}, whose four
relators are the family \textsf{H}. That is, $\delta_H(n)$ is the largest area, with respect to these relators, of a word of length at most $n$
in $p^{\pm1},q^{\pm1}$ that represents $1$ in $H$. Appendix~\ref{sec:houghton} proves the three facts used here.
\begin{enumerate}[label=(H\arabic*),itemsep=1pt,leftmargin=*]
\item \emph{Short words for finite permutations} (Lemma~\ref{lem:transpositions}). The word $\sigma_k=p^ksp^{-k}$ represents the transposition
$(k{+}1\ \,k{+}2)$, and for $1\le x<y$ the transposition $(x\ y)$ is represented by a word $\omega_{x,y}$ of length less than $10y$ in $p,q$. For
distinct $x_1,\dots,x_l\in\{1,\dots,n\}$ there is a word $\pi[x_1,\dots,x_l]$ of length less than $10ln$ that sends $m$ to $x_m$ for $m\le l$; it
is a product of at most $l$ transpositions of points in $\{1,\dots,n\}$.
\item \emph{Point stabilizers} (Lemma~\ref{lem:stabilizers}). Let $H^{(k)}$ be the pointwise stabilizer in $H$ of the points $1,\dots,k$ of
ray~$1$, and put $\zeta_k=\sigma_{k-1}\cdots\sigma_1\sigma_0$. Every element of $\Sigma_n\cap H^{(k)}$ that moves $m\ge2$ points is represented
by a word of length less than $10(m-1)n$ in the two letters of $K_k=(\zeta_kp,\zeta_kq)$. For $k=1,2$ these are the pairs of
Section~\ref{sec:certificate}, whose letters have length at most $5$ and $11$ in $p,q$.
\item \emph{The Dehn function} (Proposition~\ref{prop:dehn}). There are absolute constants $K$ and $C$ with $\delta_H(n)\le e^{Kn}$, by
Lee~\cite[Theorem~D]{Lee2012}, and $\delta_H(n)\le Cn^{p_H}$ with $p_H$ as in Theorem~\ref{thm:main-group},
by~\cite[Theorems~3.2 and~3.3]{DouglasMghirbiA}.
\end{enumerate}

Let $n\ge3$. Three elementary properties of area are also used. The area is invariant under free equality, inversion and conjugation, and
subadditive under products. If $N$ represents~$1$, then $\mu N\nu=(\mu N\mu^{-1})\mu\nu$ for all words $\mu,\nu$, so replacing a subword $U$ by
$U'$ changes the area by at most $\Area_R(UU'^{-1})$. And $[w,z_1z_2]=[w,z_1]\cdot z_1[w,z_2]z_1^{-1}$ and $[w,z^{-1}]=z^{-1}[w,z]^{-1}z$, so
$\Area_R([w,z])$ is at most the sum of the areas $\Area_R([w,z_l])$ over the letters $z_l$ of any factorization $z=z_1^{\pm1}\cdots z_m^{\pm1}$.

\paragraph{Canonical transporters.} For points $i\ne j$ in $\{1,\dots,n\}$ put
\[
U_i=\pi[i],\qquad \tau_i=U_i\,t\,U_i^{-1},\qquad U_{ij}=\pi[j,i],\qquad X_{ij}=U_{ij}\,d\,U_{ij}^{-1},
\]
with the words of Lemma~\ref{lem:transpositions}(c). We use $\tau_i$ both for this word and for the translation by $e_i$ that it represents, since
$U_i(1)=i$; and $X_{ij}$ represents $E_{i\leftarrow j}$, since $U_{ij}$ maps $1$ to $j$ and $2$ to $i$. By Lemma~\ref{lem:transpositions},
$|U_i|<10n$ and $|U_{ij}|<20n$, and $U_i$ is a product of at most one transposition and $U_{ij}$ of at most two.

Lemma~\ref{lem:exchange} says that it does not matter which short word transports a prototype: two such words differ by an element of a
stabilizer, which commutes with the transported element by the family~\textsf{S}. We allow the transporting word to end with a power of $p$,
because the relations of Table~\ref{tab:certificate} use $t_2=ptp^{-1}$ and $t_3=p^2tp^{-2}$.

\begin{lemma}[Change of transporter]\label{lem:exchange}
Let $n\ge3$ and $f\in\{t,d\}$. Let $U$ and $W_0$ be words in $p^{\pm1},q^{\pm1}$ of length less than $30n$ that represent elements of $\Sigma_n$,
such that $W_0^{-1}U$ moves at most $10$ points, and let $W=W_0p^j$, where $j=0$ if $f=d$ and $j\in\{0,1,2\}$ if $f=t$. If $W$ and $U$ agree
on $P_f$, then
\[
\Area_R\bigl(WfW^{-1}\cdot Uf^{-1}U^{-1}\bigr)\ \le\ \mathcal E(n):=2\,\delta_H(1500n)+130n .
\]
\end{lemma}

\begin{proof}
Let $k=|P_f|$ and let $g$ be the element represented by $W^{-1}U$; it fixes $P_f$ pointwise, so $g\in H^{(k)}$. If $j=0$, then $g\in\Sigma_n$ moves at
most $10$ points, and by Lemma~\ref{lem:stabilizers}(b) it is represented by a word $Z$ with fewer than $90n$ letters from $K_k^{\pm1}$. If
$j\ge1$, then $f=t$, $k=1$, and $g=(sp)^{-j}g'$ with $g'=(sp)^jp^{-j}\cdot W_0^{-1}U$. Here $(sp)p^{-1}=s$ and $(sp)^2p^{-2}=sB$ lie in
$\Sigma_3$, so $g'\in\Sigma_n$ moves at most $13$ points; and $g'$ fixes the point $1$, since $g$ and $sp$ do. By Lemma~\ref{lem:stabilizers}(b),
$g'$ is represented by a word $Z'$ with fewer than $120n$ letters from $K_1^{\pm1}$, and we put $Z=(sp)^{-j}Z'$. In both cases $Z$ has fewer than
$122n$ letters from $K_k^{\pm1}$, each of length at most $11$ in $p,q$.

Let $N=W^{-1}UZ^{-1}$. It represents $1$ in $H$ and has length less than $30n+2+30n+1342n\le1500n$, so by the family \textsf{H} its area is at
most $\delta_H(1500n)$. By the family \textsf{S}, $\Area_R([f,Z])<122n$. Since $U=WNZ$ freely,
\[
WfW^{-1}\cdot Uf^{-1}U^{-1}=W\,(fNf^{-1})\,[f,Z]\,N^{-1}\,W^{-1},
\]
whose area is at most $2\delta_H(1500n)+122n\le\mathcal E(n)$.
\end{proof}

\begin{lemma}[Local fillings]\label{lem:local-filling}
Let $n\ge3$ and
\[
L(n)=10\,\delta_H(1500n)+700\,n .
\]
Let $i,j,k\in\{1,\dots,n\}$ be points with $i\ne j$, and $x,y\in\{a,b\}$, where $b$ stands for the word $ac^4$. Each of the following words
represents $1$ and has area at most $L(n)$ with respect to $R$.
\begin{enumerate}[label=(\alph*),itemsep=1pt]
\item $\tau_i^2$, of area $1$, and $[\tau_i,\tau_j]$.
\item $X_{ij}\tau_kX_{ij}^{-1}\tau_k^{-1}$ if $k\ne j$, and $X_{ij}\tau_jX_{ij}^{-1}\tau_i^{-1}\tau_j^{-1}$.
\item $[X_{ij},x]$, of area less than $40n+1$.
\item $[x,\tau_iy\tau_i^{-1}]$, of area less than $40n+1$.
\end{enumerate}
\end{lemma}

\begin{proof}
All the words represent $1$ by~\eqref{eq:affine-identities}. Call replacing a subword $UfU^{-1}$ by $WfW^{-1}$ as in Lemma~\ref{lem:exchange} a
\emph{change of transporter}; it costs at most $\mathcal E(n)$, and $5\mathcal E(n)+1\le L(n)$. In each part, $W_0$ is a product of at most three
transpositions, of length less than $30n$, and $U\in\{U_i,U_j,U_k,U_{ij}\}$ is a product of at most two, so $W_0^{-1}U$ moves at most $10$ points.

(a) $\tau_i^2=U_it^2U_i^{-1}$ freely, a conjugate of the relation \textsf{Q}. For $[\tau_i,\tau_j]$ let $W_0=\pi[i,j]$. Replace both occurrences
of $\tau_i^{\pm1}$ by $(W_0tW_0^{-1})^{\pm1}$, and both occurrences of $\tau_j^{\pm1}$ by $(W_0t_2W_0^{-1})^{\pm1}=(W_0p\,t\,p^{-1}W_0^{-1})^{\pm1}$.
These are four changes of transporter, since $W_0$ and $U_i$ both map $1$ to $i$, and $W_0p$ and $U_j$ both map $1$ to $j$. The result
$W_0[t_1,t_2]W_0^{-1}$ is a conjugate of the relation \textsf{C}, so the area is at most $4\mathcal E(n)+1$.

(b) Let $W_0=\pi[j,i,k]$, or $\pi[j,i]$ if $k\in\{i,j\}$, so that $W_0$ maps $1$ to $j$, $2$ to $i$, and $3$ to $k$ when $k\notin\{i,j\}$; let
$m\in\{1,2,3\}$ be the point with $W_0(m)=k$, so that $m=1$ exactly when $k=j$. Replace $X_{ij}^{\pm1}$ by $(W_0dW_0^{-1})^{\pm1}$, a change of
transporter because $W_0$ and $U_{ij}$ agree on $P_d=\{1,2\}$. Replace $\tau_k^{\pm1}$ by $(W_0t_mW_0^{-1})^{\pm1}=(W_0p^{m-1}\,t\,p^{1-m}W_0^{-1})^{\pm1}$,
and $\tau_i^{\pm1}$ by $(W_0t_2W_0^{-1})^{\pm1}$, again changes of transporter. If $k\ne j$, four of them turn the word into
$W_0(dt_md^{-1}t_m^{-1})W_0^{-1}$ with $m\in\{2,3\}$, a conjugate of a relation \textsf{A}. If $k=j$, five of them turn it into
$W_0(dt_1d^{-1}t_2^{-1}t_1^{-1})W_0^{-1}$, again a conjugate of a relation \textsf{A}. So the area is at most $5\mathcal E(n)+1$.

(c) Moving $x$ across a letter of $U_{ij}$ costs one relation \textsf{Z}, and moving it across $d=c^3$ costs one relation $[d,a]$, since
$[c^3,ac^4]=[c^3,a]$ freely. Hence $\Area_R([X_{ij},x])\le2|U_{ij}|+1<40n+1$.

(d) Freely the word is $U_i[U_i^{-1}xU_i,\,t(U_i^{-1}yU_i)t^{-1}]U_i^{-1}$. Replacing the two occurrences of $U_i^{-1}x^{\pm1}U_i$ and the two of
$U_i^{-1}y^{\pm1}U_i$ by $x^{\pm1}$ and $y^{\pm1}$ costs at most $4|U_i|$ relations \textsf{Z} and leaves $U_i[x,tyt^{-1}]U_i^{-1}$, a conjugate of a
relation \textsf{P}; so the area is at most $4|U_i|+1<40n+1$.
\end{proof}
\section{Pair areas}\label{sec:pair-areas}

We now bring together two of the $2^N$ lamps supported on the first $N$ points of ray~$1$. Their transport words are products of the $\tau_i$,
and the two configurations differ by a configuration that at most $N$ transvections clear down to a single point, where the commutator of the
two lamps is a local filling. Clearing naively would let the translation word grow; we keep it normalized after every step, so that it never
has more than $N$ factors, and each step costs $O(N)$ local fillings.

\paragraph{A worked example.} Let $N=3$ and take the lamps at $u=\{1,3\}$ and $w=\{2\}$, with transport words $T_u=\tau_1\tau_3$ and
$T_w=\tau_2$. Figure~\ref{fig:clearing} draws the clearing. Translations commute, so the commutator of the two lamps is $T_u[x,ZyZ^{-1}]T_u^{-1}$
with $Z=T_u^{-1}T_w=\tau_3^{-1}\tau_1^{-1}\tau_2$, and two relations \textsf{Q} and two swaps turn $Z$ into $T_{u_0}=\tau_1\tau_2\tau_3$, where
$u_0=u+w=\{1,2,3\}$. Choose the point $3$ of $u_0$ as the \emph{pivot}, into which the other points are cleared. The transvection $X_{13}$ adds
the value at $3$ into $1$. Conjugating $T_{u_0}$ by it leaves $\tau_1$ and $\tau_2$ alone and turns $\tau_3$ into $\tau_1\tau_3$; one swap brings
the new $\tau_1$ next to the old one, and one relation \textsf{Q} cancels them, leaving $\tau_2\tau_3$. Then $X_{23}$ turns $\tau_2\tau_3$ into
$\tau_3$ in the same way. Conjugating by $G=X_{23}X_{13}$ therefore carries the lamp at $u_0$ to the lamp at $e_3$, while $G$ is linear and
commutes with the lamp at $0$. What is left is $[x,\tau_3y\tau_3^{-1}]$, a local filling. Every identity used involves at most three points, so
Lemma~\ref{lem:local-filling} prices each of them at $L(3)$.

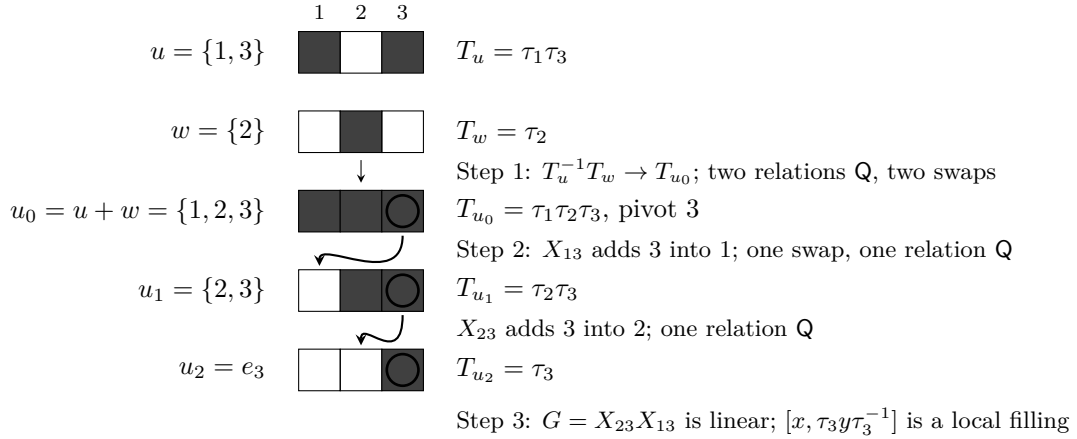
\begin{figure}[ht]
\centering
\begin{tikzpicture}[x=1cm,y=1cm,>=stealth]
\def\cw{0.55}   
\def\rh{1.05}   
\newcommand{\cellrow}[5]{%
  \foreach \c/\lit in {1/#2,2/#3,3/#4}{%
    \ifnum\lit=1 \fill[black!75] ({(\c-1)*\cw},#1) rectangle ({\c*\cw},{#1+\cw}); \fi
    \draw ({(\c-1)*\cw},#1) rectangle ({\c*\cw},{#1+\cw});
  }%
  \ifnum#5>0 \draw[line width=1pt] ({(#5-0.5)*\cw},{#1+0.5*\cw}) circle ({0.38*\cw}); \fi
}
\foreach \c in {1,2,3}{\node[font=\scriptsize] at ({(\c-0.5)*\cw},{5*\rh+\cw+0.25}) {$\c$};}
\cellrow{5*\rh}{1}{0}{1}{0}
\node[anchor=east,font=\small] at (-0.3,{5*\rh+0.5*\cw}) {$u=\{1,3\}$};
\node[anchor=west,font=\small] at ({3*\cw+0.3},{5*\rh+0.5*\cw}) {$T_u=\tau_1\tau_3$};
\cellrow{4*\rh}{0}{1}{0}{0}
\node[anchor=east,font=\small] at (-0.3,{4*\rh+0.5*\cw}) {$w=\{2\}$};
\node[anchor=west,font=\small] at ({3*\cw+0.3},{4*\rh+0.5*\cw}) {$T_w=\tau_2$};
\draw[->] ({1.5*\cw},{4*\rh-0.1}) -- ({1.5*\cw},{3*\rh+\cw+0.1});
\node[anchor=west,font=\footnotesize] at ({3*\cw+0.3},{3.5*\rh+0.5*\cw}) {Step 1: $T_u^{-1}T_w\to T_{u_0}$; two relations \textsf{Q}, two swaps};
\cellrow{3*\rh}{1}{1}{1}{3}
\node[anchor=east,font=\small] at (-0.3,{3*\rh+0.5*\cw}) {$u_0=u+w=\{1,2,3\}$};
\node[anchor=west,font=\small] at ({3*\cw+0.3},{3*\rh+0.5*\cw}) {$T_{u_0}=\tau_1\tau_2\tau_3$, pivot $3$};
\draw[->,thick] ({2.5*\cw},{3*\rh-0.05}) .. controls ({2.5*\cw},{2*\rh+0.5*\cw}) and ({0.5*\cw},{3*\rh-0.1}) .. ({0.5*\cw},{2*\rh+\cw+0.05});
\node[anchor=west,font=\footnotesize] at ({3*\cw+0.3},{2.5*\rh+0.5*\cw}) {Step 2: $X_{13}$ adds $3$ into $1$; one swap, one relation \textsf{Q}};
\cellrow{2*\rh}{0}{1}{1}{3}
\node[anchor=east,font=\small] at (-0.3,{2*\rh+0.5*\cw}) {$u_1=\{2,3\}$};
\node[anchor=west,font=\small] at ({3*\cw+0.3},{2*\rh+0.5*\cw}) {$T_{u_1}=\tau_2\tau_3$};
\draw[->,thick] ({2.5*\cw},{2*\rh-0.05}) .. controls ({2.5*\cw},{1*\rh+0.5*\cw}) and ({1.5*\cw},{2*\rh-0.1}) .. ({1.5*\cw},{1*\rh+\cw+0.05});
\node[anchor=west,font=\footnotesize] at ({3*\cw+0.3},{1.5*\rh+0.5*\cw}) {$X_{23}$ adds $3$ into $2$; one relation \textsf{Q}};
\cellrow{1*\rh}{0}{0}{1}{3}
\node[anchor=east,font=\small] at (-0.3,{1*\rh+0.5*\cw}) {$u_2=e_3$};
\node[anchor=west,font=\small] at ({3*\cw+0.3},{1*\rh+0.5*\cw}) {$T_{u_2}=\tau_3$};
\node[anchor=west,font=\footnotesize] at ({3*\cw+0.3},{0.5*\rh+0.5*\cw-0.15}) {Step 3: $G=X_{23}X_{13}$ is linear; $[x,\tau_3y\tau_3^{-1}]$ is a local filling};
\end{tikzpicture}
\caption{The clearing of the worked example, $N=3$: the lit cells are the support of a configuration on the points $1,2,3$ of ray~$1$, the ringed
cell is the pivot, and each arrow is a transvection that adds the value at the pivot into the cell it points to. After $X_{13}$ and $X_{23}$ the
configuration $u_0=\{1,2,3\}$ has become $e_3$, and the translation word has been renormalized after each step.}
\label{fig:clearing}
\end{figure}

\begin{theorem}[Pair areas]\label{thm:pair-area}
Let $N\ge1$ and $n=\max(3,N)$. For a configuration $u$ supported in the first $N$ points of ray~$1$ let $T_u=\tau_{i_1}\cdots\tau_{i_m}$,
where $i_1<\dots<i_m$ are the points of the support of $u$. Then $T_u$ represents $\tau_u$, so $T_uaT_u^{-1}$ and $T_ubT_u^{-1}$ generate the copy
of $S_3$ at $u$; $|T_u|\le21Nn$; and for distinct such $u,w$ and $x,y\in\{a,b\}$, with $b$ the word $ac^4$,
\[
\Area_R\bigl([T_uxT_u^{-1},\,T_wyT_w^{-1}]\bigr)\ \le\ \Lambda_N:=20N^2L(n)\ \le\ C_0N^2\bigl(\delta_H(C_0N)+N\bigr)
\]
with $C_0=4.5\cdot10^4$. Consequently there are absolute constants $K'$ and $C'$ with
\[
\Lambda_N\le e^{K'N}\qquad\text{and}\qquad \Lambda_N\le C'N^{2+p_H}\qquad(N\ge1),\qquad p_H=\log107+4 .
\]
\end{theorem}

\begin{proof}
Translations commute, so $T_u$ represents $\tau_u$, and $|T_u|\le N(2\cdot10n+1)\le21Nn$. Write $L=L(n)$; by Lemma~\ref{lem:local-filling},
each of its four local identities has area at most $L$.

\emph{Step 1: one translation word.} Freely, $[T_uxT_u^{-1},T_wyT_w^{-1}]=T_u[x,ZyZ^{-1}]T_u^{-1}$ with $Z=T_u^{-1}T_w$. Replace each of the at
most $N$ factors $\tau_i^{-1}$ of $T_u^{-1}$ by $\tau_i$, at one relation \textsf{Q} each. Sort the resulting product of at most $2N$ factors
$\tau_i$ into increasing order; this takes at most $\binom{2N}2$ swaps of adjacent distinct factors, each a local filling of
Lemma~\ref{lem:local-filling}(a). Cancel the at most $N$ adjacent equal pairs, again at one relation \textsf{Q} each. The result is $T_{u_0}$ with
$u_0=u+w\ne0$, and $\Area_R(ZT_{u_0}^{-1})\le2N+N(2N-1)L\le2N^2L$. Since $Z^{\pm1}$ occurs four times in $[x,ZyZ^{-1}]$,
\[
\Area_R\bigl([T_uxT_u^{-1},T_wyT_w^{-1}]\bigr)\le8N^2L+\Area_R\bigl([x,T_{u_0}yT_{u_0}^{-1}]\bigr).
\]

\emph{Step 2: normalized clearing.} Let $s\le N$ be the size of the support of $u_0$, choose a point $j$ in it as the pivot, and list the other
points of the support as $i_1,\dots,i_{s-1}$. Put $G_l=X_{i_lj}\cdots X_{i_1j}$ and $u_l=u_0+e_{i_1}+\dots+e_{i_l}$, whose support is that of $u_0$
with $i_1,\dots,i_l$ removed; it still contains $j$, and $u_{s-1}=e_j$. We claim
\[
\Area_R\bigl(G_lT_{u_0}G_l^{-1}\,T_{u_l}^{-1}\bigr)\le2lsL\qquad(0\le l\le s-1).
\]
For $l=0$ there is nothing to prove. Given the claim for $l$, write $X=X_{i_{l+1}j}$ and $i=i_{l+1}$. Freely
$G_{l+1}T_{u_0}G_{l+1}^{-1}=X(G_lT_{u_0}G_l^{-1})X^{-1}$, which differs from $XT_{u_l}X^{-1}$ by a conjugate of a word of area at most $2lsL$. Next,
$XT_{u_l}X^{-1}$ is freely the product of the words $X\tau_{i'}X^{-1}$ over the support of $u_l$, in order. By Lemma~\ref{lem:local-filling}(b),
at a cost of at most $sL$, each of these becomes $\tau_{i'}$, except that $X\tau_jX^{-1}$ becomes $\tau_j\tau_i$. The new factor $\tau_i$ has a
partner $\tau_i$ in $T_{u_l}$, because $i$ lies in the support of $u_l$; at most $s-1$ swaps (Lemma~\ref{lem:local-filling}(a)) make the two
adjacent, and one relation \textsf{Q} cancels them. What remains is $T_{u_{l+1}}$. This step costs at most $sL+(s-1)L+1\le2sL$, which proves the
claim. The translation word is renormalized after every transvection and never has more than $s\le N$ factors.

\emph{Step 3: a single point.} Let $G=G_{s-1}$ and let $M=GT_{u_0}G^{-1}\tau_j^{-1}$, which has area at most $2s(s-1)L\le2N^2L$ by Step~2. Freely
$T_{u_0}=G^{-1}M\tau_jG$. Substituting this into the four occurrences of $T_{u_0}^{\pm1}$ in $[x,T_{u_0}yT_{u_0}^{-1}]$ and deleting the four
occurrences of $M^{\pm1}$ costs at most $8N^2L$. By Lemma~\ref{lem:local-filling}(c), $\Area_R([G,z])\le(s-1)(40n+1)$ for $z\in\{a,b\}$, and the
same bound holds for $[G,z^{-1}]$ and $[z,G^{-1}]$. Moving $y^{\pm1}$ across $G$ and $G^{-1}$, and then $x^{\pm1}$, leaves
$G^{-1}[x,\tau_jy\tau_j^{-1}]G$, at a cost of at most $4(s-1)(40n+1)$; and $[x,\tau_jy\tau_j^{-1}]$ has area at most $40n+1$ by
Lemma~\ref{lem:local-filling}(d). Altogether, since $L\ge700n$,
\[
\Area_R\bigl([x,T_{u_0}yT_{u_0}^{-1}]\bigr)\le8N^2L+4N(40n+1)+40n+1\le9N^2L,
\]
and with Step~1 the area is at most $17N^2L\le\Lambda_N$.

\emph{Growth.} Since $n\le3N$, $L(n)\le10\delta_H(4500N)+2100N$, which gives the first bound on $\Lambda_N$. With Proposition~\ref{prop:dehn}(a)
it is at most $e^{K'N}$ for a suitable $K'$. With Proposition~\ref{prop:dehn}(b) it is at most $C_0N^2(C(C_0N)^{p_H}+N)\le C'N^{2+p_H}$, since $p_H>1$.
\end{proof}

The factor $N^2$ counts local moves: at most $N$ transvections, each followed by at most $N$ commutations. The Dehn function of $H$ enters
through the local fillings, at the scale $N$ of the transport. Theorem~\ref{thm:pair-area} is the only transport input of the paper: the lower
bounds of Theorem~A (Section~\ref{sec:profile}) and of Theorem~B (Section~\ref{sec:lamp-channel}) feed it into the criterion and into the
one-round theorem.
\section{Hash models, and the proof of Theorem A}\label{sec:profile}

The lower bound of Theorem~A follows by feeding $2^N$ lamps of pair area polynomial in $N$ (Theorem~\ref{thm:pair-area}) into the criterion.
The upper bound needs models far smaller than Følner sets, which we build by hashing configurations through random affine maps, a universal
hash family that the affine group moves covariantly. The two halves meet in the proof of Theorem~A at the end of the section.

\subsection{Affine hashing and the upper bound}\label{sec:hashing}

The Følner sets of $\Gamma$ are doubly exponential (Proposition~\ref{prop:folner-lower}), but its chunks have far smaller models. A Følner set
gives models of every chunk, by Cornulier's comparison of the sofic profile with the Følner function~\cite[Proposition~3.14]{Cornulier2013}, but
it records the lamp at every configuration of a large finite part of $V$. A model need only record the lamps that are lit, and it can store them
through a random affine hash of the configurations, which the affine group moves covariantly. Random $\F_2$-linear maps form a universal
hash family of Carter and Wegman~\cite{CarterWegman1979}, and a random translation does not change collision probabilities.

Let $W$ be a finite-dimensional vector space over $\F_2$, and let $K_W=S_3^W\rtimes\AGL(W)$, where $S_3^W$ is the group of functions
$W\to S_3$ under pointwise multiplication, and $g\in\AGL(W)$ acts by $(g\lambda)(w)=\lambda(g^{-1}w)$. Thus $(\lambda,g)(\mu,h)=(\lambda\cdot g\mu,\,gh)$.
The support of $\lambda$ is the set of $w$ with $\lambda(w)\ne1$. Fix $k\ge1$ and a total order on $W$, let $\operatorname{Aff}(W,\F_2^k)$ be the
set of affine maps $W\to\F_2^k$, and let
\[
\Omega=\operatorname{Aff}(W,\F_2^k)\times S_3^{\F_2^k},\qquad |\Omega|=2^{k(\dim W+1)}\,6^{2^k}.
\]
For $\eta\in\operatorname{Aff}(W,\F_2^k)$ and $\lambda\in S_3^W$ let $\lambda^\eta:\F_2^k\to S_3$ send $z$ to the product, in the order of $W$, of the
$\lambda(w)$ with $\eta(w)=z$. For $\gamma=(\lambda,g)\in K_W$ define a permutation $\Theta_\gamma$ of $\Omega$ by
\[
\Theta_\gamma(\eta,u)=\bigl(\eta\circ g^{-1},\ \lambda^{\eta\circ g^{-1}}\cdot u\bigr);
\]
its inverse sends $(\eta',u')$ to $(\eta'\circ g,(\lambda^{\eta'})^{-1}u')$, and $\Theta_1$ is the identity.

\begin{lemma}[Affine hashing]\label{lem:hash}
Let $b\ge1$, let $\gamma_1=(\lambda,g)$ and $\gamma_2=(\mu,h)$ be elements of $K_W$ whose lamp parts $\lambda,\mu$ have supports of size at most
$b$, and let $\omega$ be uniformly distributed on $\Omega$. Then
\begin{enumerate}[label=(\roman*),itemsep=1pt]
\item $\Theta_{\gamma_1}\Theta_{\gamma_2}\omega=\Theta_{\gamma_1\gamma_2}\omega$ except with probability at most $\binom{2b}2\,2^{-k}$;
\item if $\gamma_1\ne\gamma_2$, then $\Theta_{\gamma_1}\omega\ne\Theta_{\gamma_2}\omega$ except with probability at most $\binom{2b}2\,2^{-k}$.
\end{enumerate}
\end{lemma}

\begin{proof}
Write $\omega=(\eta,u)$ and $\eta=\eta_0+z_0$ with $\eta_0$ linear; then $\eta_0$ and $z_0$ are independent and uniform. For $w\ne w'$ in $W$,
$\eta(w)-\eta(w')=\eta_0(w-w')$ is uniform on $\F_2^k$, because $\eta_0\mapsto\eta_0(w-w')$ is a surjective linear map; so
$\eta(w)=\eta(w')$ with probability exactly $2^{-k}$. Precomposition with a fixed affine bijection of $W$ permutes $\operatorname{Aff}(W,\F_2^k)$, so
$\eta\circ f$ is again uniform for every $f\in\AGL(W)$. Hence, for a set $S\subset W$ of size at most $2b$, $\eta\circ f$ is injective on $S$ except
with probability at most $\binom{2b}2\,2^{-k}$. If $\eta'$ is injective on a set $S$ containing the support of $\lambda$, then
$\lambda^{\eta'}(\eta'(w))=\lambda(w)$ for $w\in S$, and $\lambda^{\eta'}=1$ off $\eta'(S)$.

(i) Put $\eta'=\eta\circ(gh)^{-1}$. Then $\Theta_{\gamma_1}\Theta_{\gamma_2}(\eta,u)=(\eta',\lambda^{\eta'}\mu^{\eta\circ h^{-1}}u)$ and
$\Theta_{\gamma_1\gamma_2}(\eta,u)=(\eta',(\lambda\cdot g\mu)^{\eta'}u)$. Let $S$ be the union of the support of $\lambda$ and the image under $g$ of
the support of $\mu$, and suppose that $\eta'$ is injective on $S$. Since $\eta\circ h^{-1}(v)=\eta'(gv)$, the map $\eta\circ h^{-1}$ is injective on
the support of $\mu$, and $\mu^{\eta\circ h^{-1}}(\eta'(w))=\mu(g^{-1}w)$ for $w\in S$. The support of $\lambda\cdot g\mu$ lies in $S$. So all three
functions are trivial off $\eta'(S)$, and at $\eta'(w)$, $w\in S$, both sides equal $\lambda(w)\mu(g^{-1}w)$.

(ii) If $g\ne h$, choose $w_0$ with $h^{-1}g(w_0)\ne w_0$. If $\eta\circ g^{-1}=\eta\circ h^{-1}$, then $\eta=\eta\circ h^{-1}g$, so
$\eta(h^{-1}g(w_0))=\eta(w_0)$, which has probability $2^{-k}$. If $g=h$, then $\lambda\ne\mu$; put $\eta'=\eta\circ g^{-1}$ and let $S$ be the union
of the supports. If $\eta'$ is injective on $S$, then $\lambda^{\eta'}$ and $\mu^{\eta'}$ take the different values $\lambda(w)$ and $\mu(w)$ at
$\eta'(w)$ for some $w\in S$.
\end{proof}

To use the lemma in $\Gamma$ we need finite subgroups that contain the relevant cocycles. The quotient $\Z^2$ of Lemma~\ref{lem:structure}(c)
serves as a clock. A model is a box in $\Z^2$ times the hash space $\Omega$; a generator advances the clock and applies a cocycle, which lies in
a finite subgroup of the locally finite kernel of $\pi$. For $z\in\Z^2$ put $g_z=p^{z_1}q^{z_2}$, so that
$\pi(g_z)=z$, and for $m\ge1$ let $B_m=\{-m,\dots,m\}^2$. For $R\ge1$ let $Y_R\subset Y$ be the set of points of height at most $R$, and let
$W_R=\F_2^{Y_R}\subset V$. Let $K_R\le\ker\pi$ consist of the elements whose lamps lie at configurations in $W_R$ and whose affine part lies in
$W_R\rtimes\GL(W_R)$, with $\GL(W_R)$ extended by the identity on the other basis vectors. Restricting to $W_R$ identifies $K_R$ with $K_{W_R}$.

\begin{lemma}[The clock]\label{lem:clock}
Let $m\ge1$, $R=2m+3$, $z\in B_m$ and $s\in\{p,q,t,c,a\}^{\pm1}$. Then $g_{z+\pi(s)}^{-1}sg_z$ and $g_zsg_{z+\pi(s)}^{-1}$ lie in $K_R$, and their
lamps sit at no more than one configuration.
\end{lemma}

\begin{proof}
Each of $p^{\pm1},q^{\pm1}$ changes heights by at most one and acts on the points of height at least $2$ as a translation along the rays, and these
translations commute. Hence for $(y_1,y_2)\in\Z^2$ the commutator $[p^{y_1},q^{y_2}]$ fixes every point of height greater than $|y_1|+|y_2|+1$:
while it is applied, such a point stays at height at least $2$. For $s=p^{\pm1}$ we have $\pi(s)=(\pm1,0)$, $g_{z+\pi(s)}^{-1}sg_z=1$ and
$g_zsg_{z+\pi(s)}^{-1}=p^{z_1}[q^{z_2},p^{\pm1}]p^{-z_1}$, which is supported in heights at most $|z_2|+2+|z_1|\le2m+2$. For $s=q^{\pm1}$ we have
$\pi(s)=(0,\pm1)$, $g_zsg_{z+\pi(s)}^{-1}=1$ and $g_{z+\pi(s)}^{-1}sg_z=q^{-z_2\mp1}[p^{-z_1},q^{\pm1}]q^{z_2\pm1}$, supported in heights at most
$|z_1|+2+|z_2|+1\le2m+3$. These are permutations of
$Y_R$, hence elements of $\GL(W_R)$. For $s\in\{t,c,a\}^{\pm1}$ we have $\pi(s)=0$, and $g_z^{\pm1}$ moves the points $1$ and $2$ to height at most
$2m+2$; so the conjugates of $t$ and $d$ by $g_z^{\pm1}$ are a translation and a transvection of $W_R$, while those of $a$ and $ab$ are $a$ and
$ab$, at the configuration~$0$.
\end{proof}

\begin{proposition}[Hash models]\label{prop:upper}
Let $E'$ be a chunk of $\Gamma$ each of whose elements is a product of at most $\ell\ge1$ of the generators $p,q,t,c,a$ and their inverses. Then
\[
\log\sigma_{E'}(r)\ \le\ 500\,\ell^2\,r\log(2+r)\qquad(r>1).
\]
\end{proposition}

\begin{proof}
Let $m=\lceil8\ell r\rceil$, $R=2m+3$, $D=|Y_R|=6m+9$ and $k=\lceil\log(4\ell^2r)\rceil$, and let $\Omega$ be the hash space for $W=W_R$. For $g\in\Gamma$
and $z\in\Z^2$ put $\psi(g,z)=g_{z+\pi(g)}^{-1}gg_z\in\ker\pi$; then $\psi(gh,z)=\psi(g,z+\pi(h))\,\psi(h,z)$. If $g$ is a product of at most $\ell$
generators and inverses and $z\in B_{m-\ell}$, then $\psi(g,z)$ is the product of the $\psi(s,z')$ for the letters $s$ of the word, at points
$z'\in B_m$, so by Lemma~\ref{lem:clock} it lies in $K_R$ with lamps at no more than $\ell$ configurations.

Let $X=B_m\times\Omega$. For $g\in E'$ let $F_g$ be a permutation of $X$ that extends the injective partial map
$(z,\omega)\mapsto(z+\pi(g),\Theta_{\psi(g,z)}\omega)$, defined for $z\in B_{m-\ell}$, with $F_1=\mathrm{id}$. Let $g,h,gh\in E'$ and
$z\in B_{m-2\ell}$. Then $z$ and $z+\pi(h)$ lie in $B_{m-\ell}$, so $F_gF_h(z,\omega)=(z+\pi(gh),\Theta_{\psi(g,z+\pi(h))}\Theta_{\psi(h,z)}\omega)$ and
$F_{gh}(z,\omega)=(z+\pi(gh),\Theta_{\psi(g,z+\pi(h))\psi(h,z)}\omega)$; by Lemma~\ref{lem:hash}(i) these agree for all but a fraction
$\binom{2\ell}2\,2^{-k}$ of the $\omega$. Since $|B_{m-2\ell}|\ge(1-8\ell/(2m+1))|B_m|$,
\[
d\bigl(F_gF_h,F_{gh}\bigr)\ \le\ \varepsilon:=\frac{8\ell}{2m+1}+\binom{2\ell}2\,2^{-k}<\frac1{2r}+\frac1{2r}=\frac1r .
\]
If $g\ne h$ in $E'$ and $z\in B_{m-\ell}$, then either $\pi(g)\ne\pi(h)$ and $F_g(z,\omega)\ne F_h(z,\omega)$, or $\psi(g,z)\ne\psi(h,z)$ and
Lemma~\ref{lem:hash}(ii) applies; so $d(F_g,F_h)\ge1-\varepsilon$. Hence $\sigma_{E'}(r)\le|X|$.

Finally $\log|X|=2\log(2m+1)+k(D+1)+2^k\log6$. Here $2^k<8\ell^2r$, $k<\log(8\ell^2r)\le5\ell\log(2+r)$ and $D+1\le64\ell r$, so
$k(D+1)\le320\ell^2r\log(2+r)$; also $2^k\log6\le21\ell^2r\le14\ell^2r\log(2+r)$ and $2\log(2m+1)\le10\ell^2r\log(2+r)$. The sum is less than
$500\ell^2r\log(2+r)$.
\end{proof}

\subsection{Proof of Theorem~\ref{thm:main-group}}

We now feed the pair areas of Theorem~\ref{thm:pair-area} into the criterion and the hash models of Proposition~\ref{prop:upper} into the
upper bound.

\begin{proof}[Proof of Theorem~\ref{thm:main-group}]
\emph{Setting up the criterion.} Theorem~\ref{thm:criterion} asks for $a$ and $b$ among the generators, while in $R$ the letter $b$ abbreviates
$ac^4$. So let $S^+$ consist of the five generators and $b=ac^4\in\Gamma$, and let $R^+$ consist of $R$ and the three words $b^{-1}ac^4$, $b^2$
and $(ab)^3$ in the letters of $S^+$. These words represent $1$, $a^2,b^2,(ab)^3\in R^+$, $ab\ne1$, and the longest word of $R^+$ has length
$l=24$. The chunk $E^+$ that Theorem~\ref{thm:criterion} associates with $S^+$ and $R^+$ is contained in $E$. The reason, in one line, is that
every new element is a power of $c$ or is represented, using $a^2=1$, by a prefix of the expanded words $b^2$ and $(ab)^3$. In detail, besides
the elements of $E$, the chunk $E^+$ contains $b^{\pm1}$, the prefixes $b^{-1},c^{-4},c^{-3},c^{-2},c^{-1}$ of $b^{-1}ac^4$, and the prefixes $b$,
$ab=c^4$, $aba=c^4a$, $abab=c^2$ and $ababa=c^2a$ of $b^2$, $(ab)^3$ and $ab$. These lie in $E$, because $ac^4$, $a^2c^4$, $a^2c^4a$ and
$a^2c^4a^2c^4a$ are prefixes of $b^2$ and $(ab)^3$ in Table~\ref{tab:certificate}, $c^2$ and $c^3$ are prefixes of the relations $[d,k]$ of the
family \textsf{S}, $a^2=1$, and $c^6=1$. Restricting a model of $E$ to $E^+$ gives a model of $E^+$, so $\sigma_E\ge\sigma_{E^+}$.

For $N\ge1$ take the $M=2^N$ words $T_u$ of Theorem~\ref{thm:pair-area}, one for each configuration $u$ supported in the first $N$ points of
ray~$1$. A commutator $[T_uxT_u^{-1},T_wyT_w^{-1}]$ with $x,y\in\{a,b\}$ contains the letter $b^{\pm1}$ at most four times, and replacing each
occurrence by $(ac^4)^{\pm1}$ costs one relation $b^{-1}ac^4$. So its area with respect to $R^+$ is at most $\Lambda^+_N=\Lambda_N+4$, and
Theorem~\ref{thm:criterion} gives
\begin{equation}\label{eq:criterion-applied}
\log\sigma_E(r)\ge2^N/16\qquad\text{whenever}\qquad r\ge4.7\cdot10^{9}\cdot2^N\Lambda^+_N .
\end{equation}
Since $a\ne1$, a $(1-1/r)$-expansive map needs at least two points, so $\log\sigma_E(r)\ge1$; for $r$ below any fixed bound $r_0$ claim (a)
therefore holds once $c$ is small enough, and we may assume $r$ large.

(a) By Theorem~\ref{thm:pair-area}, $\Lambda^+_N\le C_1N^\kappa$ with $C_1$ absolute. Let $r'=r/(4.7\cdot10^9C_1)$ and
$N=\lfloor\log r'-\kappa\log\log r'\rfloor$, which is at least $1$ for large $r$. Then
$2^N\Lambda^+_N\le C_1r'(\log r')^{-\kappa}(\log r')^\kappa=r/(4.7\cdot10^9)$, so~\eqref{eq:criterion-applied} applies, and
$\log\sigma_E(r)\ge2^N/16\ge r'/(32(\log r')^\kappa)$. Since $\log r'\le\log r$, this is~(a).

(b) Every element of $E$ is a generator, the inverse of one, or represented by a prefix of a word of length at most $24$ or by its inverse.
Proposition~\ref{prop:upper} with $\ell=24$ gives $\log\sigma_E(r)\le500\cdot24^2r\log(2+r)$, and $\log(2+r)\le2\log r$ for $r\ge2$. The same
proposition applies to every chunk.

(c) Extend $\phi$ to $S^+$ by $\phi(b)=\phi(a)\phi(c)^4$. Then $b^{-1}ac^4$ is mapped to $I$ exactly, and $b^2$ and $(ab)^3$ are mapped to the
images of the words $b^2$ and $(ab)^3$ of Table~\ref{tab:certificate}; so every word of $R^+$ is mapped within $e$ of $I$. With
$\Lambda^+_N\le C_1N^\kappa$ as in (a), the condition $K_\Delta\sqrt M\,\Lambda^+_Ne\le1$ of Theorem~\ref{thm:criterion} holds when
$2^NN^{2\kappa}\le Y:=(C_1K_\Delta e)^{-2}$, which is the case for $N=\lfloor\log Y-2\kappa\log\log Y\rfloor$; this $N$ is at least $1$ for $e$ small
in terms of $\Delta$. Then $2^N\ge Y/(2(\log Y)^{2\kappa})$, and since $C_1K_\Delta\ge1$ we have $\log Y\le2\log(1/e)$. Hence
\[
\log D\ \ge\ \frac{\Delta^2}{12}\,2^N\ \ge\ \frac{\Delta^2}{24\cdot4^\kappa C_1^2K_\Delta^2}\cdot\frac{e^{-2}}{(\log(1/e))^{2\kappa}}.
\]

(d) Proposition~\ref{prop:fp} gives finite presentation and Lemma~\ref{lem:structure} elementary amenability.
\end{proof}

\begin{remark}\label{rem:lee}
Part~(a) uses the polynomial Dehn bound of Proposition~\ref{prop:dehn}(b). With Lee's exponential bound alone, Theorem~\ref{thm:pair-area} gives
$\Lambda^+_N\le2^{BN}$ for a constant $B$ that depends on the constant in Lee's inequality, and $N=\lfloor\log(r/(4.7\cdot10^9))/(1+B)\rfloor$
in~\eqref{eq:criterion-applied} gives only $\log\sigma_E(r)\ge c\,r^{1/(1+B)}$: a superpolynomial profile, with an exponent we cannot make
explicit. Transport at scale $N$ is too expensive for an exponential Dehn bound. The grid host transports at scale $\sqrt N$ and reaches
$\log\sigma\ge c\,r\,e^{-C\sqrt{\log r}}$ with Lee's bound alone (conditionally, Remark~\ref{rem:grid}).
\end{remark}
\section{More about $\Gamma$}\label{sec:more}

Theorem~A asserts that $\Gamma$ is finitely presented, and the introduction claims that $\Gamma$ embeds in Brin's group $3V$. This section
proves both, and then shows that the F\o lner function of $\Gamma$ is doubly exponential, so that the hash models of Section~\ref{sec:profile}
are exponentially smaller than any model built from a F\o lner set. Nothing here is used in the profile bounds of
Theorem~\ref{thm:main-group}, parts (a) to (c), or in any proof of Part~II, where the embedding in $3V$ serves only as a second
check of the Choi rank; part~(d) of that theorem, finite presentation, is
Proposition~\ref{prop:fp}. A reader heading for Theorem~B can skip this section and Section~\ref{sec:ladder}.

\subsection{A finite presentation}\label{sec:fp}

The proof that $\Gamma$ is finitely presented has three steps. A presentation lemma for groups of the form $F(\mathcal X)\rtimes H$ reduces
everything to finitely many $H$-orbits of relations. Steinberg's relations for the transvections, with Johnson's presentation of $H$ and the
relations of the translations, then present the affine group $A$, and omitting the translations presents $L$. Cornulier's criterion for
permutational wreath products presents $\Gamma$, because $A$ acts doubly transitively on the configurations with finitely generated
stabilizers.

The certificate of Section~\ref{sec:certificate} is a finite set of relations, but it does not present $\Gamma$. The group it presents has
abelianization $\Z^2\times\Z/2\times\Z/4$, in which $c$ has order four; this is a computation with exponent sums, for instance $a^2$, $b^2=(ac^4)^2$
and $(ab)^3$ give $2a=8c=12c=0$. In $\Gamma$, by contrast, $c=d\,ab$ is a product of commutators, since $d=[E_{2\leftarrow3},E_{3\leftarrow1}]$
and $ab$ is a three-cycle.

\begin{proposition}[Finite presentation]\label{prop:fp}
The groups $L$, $A$ and $\Gamma$ are finitely presented.
\end{proposition}

\begin{proof}
\emph{Stabilizers.} For a finite set $D\subset Y$ the pointwise stabilizer $H_{(D)}$ of $D$ in $H$ is isomorphic to $H$. Renumbering the points of
each ray outside $D$ in increasing order turns $Y\setminus D$ into three rays, and a permutation fixing $D$ is eventually a translation on the
old rays exactly when it is on the new ones. So $H_{(D)}$ is finitely generated. Since $H$ contains $\FSym(Y)$, it acts transitively on ordered
$k$-tuples of distinct points for every $k$.

\emph{A presentation lemma.} Let a finitely presented group $H$ act on a set $\mathcal X$ with finitely many orbits and finitely generated
stabilizers, and let $F(\mathcal X)$ be the free group on $\mathcal X$. Then $F(\mathcal X)\rtimes H$, with $H$ permuting the free generators, is
finitely presented. Take a presentation of $H$, adjoin one representative $x_i$ of each orbit, and add the relations $[x_i,k]=1$ for $k$ in a
finite generating set of the stabilizer of $x_i$. In the group so presented, $hx_ih^{-1}$ depends only on $h\cdot x_i$, and the group is
isomorphic to $F(\mathcal X)\rtimes H$: both are the iterated amalgam of $H$ with the groups $\mathrm{Stab}(x_i)\times\langle x_i\rangle$ over
the stabilizers $\mathrm{Stab}(x_i)$. Imposing in addition a set of relations that is a finite union of $H$-orbits keeps the group finitely
presented, since one representative of each orbit suffices.

\emph{A presentation of $A$.} For distinct $i,j\in Y$ let $x_{ij}$ be the transvection $e_j\mapsto e_j+e_i$, let $t_i=\tau_{e_i}$, and let
$\varsigma_{ij}\in\FSym(Y)\le H$ be the transposition of $i$ and $j$. Let $\widetilde A$ be the quotient of $F(\mathcal X)\rtimes H$, where
$\mathcal X=\{x_{ij}\}\sqcup\{t_i\}$ and $H$ relabels indices, by the relations
\begin{gather*}
x_{ij}^2=1,\qquad [x_{ij},x_{kl}]=1\ (i\ne l,\ j\ne k),\qquad [x_{ij},x_{jk}]=x_{ik}\ (i,j,k\text{ distinct}),\\
\varsigma_{ij}=x_{ij}x_{ji}x_{ij},\qquad t_i^2=1,\qquad [t_i,t_j]=1,\\
x_{ij}t_jx_{ij}^{-1}=t_it_j,\qquad x_{ij}t_kx_{ij}^{-1}=t_k\ (k\ne j).
\end{gather*}
They hold in $A$, so there is a surjection $\widetilde A\to A$. Each relation involves at most four points, so by transitivity the relations form
finitely many $H$-orbits; the stabilizers of $x_{ij}$ and $t_i$ are $H_{(\{i,j\})}$ and $H_{(\{i\})}$. Since $H$ is finitely
presented~\cite{Brown1987}, for instance by Johnson's presentation~\eqref{eq:johnson}~\cite{Johnson1999}, $\widetilde A$ is finitely presented, and
it remains to show that $\widetilde A\to A$ is injective. The linear relations are Steinberg's~\cite{Milnor1971}; we give the short direct argument.

Fix a finite set $D\subset Y$, number its points $1,\dots,m$ in an order unrelated to heights, and let $U_D\le\widetilde A$ be generated by the
$x_{ij}$ with $i<j$ in $D$. By the first three relations, $x_{1m},\dots,x_{m-1,m}$ are commuting involutions, and they generate a subgroup
normalized by the $x_{ij}$ with $i<j<m$. By induction $|U_D|\le2^{m-1}2^{(m-1)(m-2)/2}=2^{m(m-1)/2}$, the order of the upper unitriangular
group $\mathrm{UT}_m(\F_2)$, onto which $U_D$ maps. So $U_D\cong\mathrm{UT}_m(\F_2)$. Let $W_D=\Sym(D)\le H$ and $\varsigma_i=\varsigma_{i,i+1}$. By
the relation $\varsigma_{ij}=x_{ij}x_{ji}x_{ij}$ and relabeling, the subgroup $G_D$ generated by all $x_{ij}$ with $i,j\in D$ contains $W_D$ and is
generated by $U_D$ and the $\varsigma_i$. We claim $G_D=U_DW_DU_D$; since this set contains $1$, it suffices that it is closed under right
multiplication by $U_D$ and by each $\varsigma_i$. Every $b\in U_D$ is $x_{i,i+1}^\epsilon b_0$ with $b_0$ in the subgroup generated by the other
positive roots, the kernel of the entry $(i,i+1)$, which $\varsigma_i$ normalizes. So it suffices to treat $wx_{i,i+1}\varsigma_i$ with $w\in W_D$;
put $k=w(i)$ and $l=w(i+1)$. If $k<l$, then $wx_{i,i+1}\varsigma_i=x_{kl}w\varsigma_i\in U_DW_D$. If $k>l$, then $w\varsigma_i=\varsigma_{kl}w$ and
$x_{kl}\varsigma_{kl}=x_{lk}x_{kl}$, so $wx_{i,i+1}\varsigma_i=x_{kl}\varsigma_{kl}w=x_{lk}wx_{i,i+1}\in U_DW_DU_D$. The map $G_D\to\GL(\F_2^D)$ is
onto, since transvections generate $\GL(\F_2^D)$. If $b_1wb_2$ maps to the identity of $\GL(\F_2^D)$, the permutation matrix of $w$ is upper
unitriangular, so $w=1$ in $H$, and then $b_1b_2=1$ in $U_D$. Hence $G_D\cong\GL(\F_2^D)$. Every element of the subgroup $\widetilde J$ generated
by all $x_{ij}$ lies in some $G_D$, so $\widetilde J\cong\GL_{\rm fin}(Y)$. By relabeling, every element of the subgroup generated by $\widetilde J$
and $H$ has the form $jh$. If it maps to $1$, then $h\in\GL_{\rm fin}(Y)\cap H=\FSym(Y)$, so $h$ is a product of transpositions, each equal in
$\widetilde A$ to a word in the $x_{ij}$, and $jh\in\widetilde J$ maps to $1$, hence is $1$. Finally the $t_i$ generate a normal elementary abelian
subgroup. It is a quotient of $\F_2^{(Y)}$ and maps onto $V$ by $t_i\mapsto e_i$, hence isomorphically, and the linear part acts on it as on $V$.
Hence every element of $\widetilde A$ is $t\ell$ with $t$ in that subgroup. If its image is the identity affine map, evaluating at $0$ gives
$t=1$, and then $\ell=1$. So $\widetilde A\cong A$. Omitting the $t_i$ gives $L$.

\emph{The lamps.} The group $A$ acts on $V$ doubly transitively: a translation carries $v$ to $0$, and a finitary linear map carries any nonzero
vector to $e_1$. The stabilizer of $0$ is $L$, which is generated by $p$, $q$ and $d$, because the $H$-conjugates of $d$ are all the transvections
$E_{y\leftarrow y'}$. So $A$ is finitely presented, its point stabilizers are finitely generated, and it has two orbits on $V\times V$. By
Cornulier's criterion for permutational wreath products~\cite[Theorem~1.1]{Cornulier2006}, $\Gamma=S_3\wr_VA$ is finitely presented. Explicitly,
let $\mathcal R_A$ present $A$ on generators that include $p,q,d,t$. Add the generators $a,b$ and the relations $a^2$, $b^2$, $(ab)^3$, $[a,\ell]$
and $[b,\ell]$ for $\ell\in\{p,q,d\}$, and $[u,tvt^{-1}]$ for $u,v\in\{a,b\}$. The result presents $\Gamma$. The lamp at $v\in V$ is well defined
as a conjugate of $\langle a,b\rangle$, because $L$ centralizes $\langle a,b\rangle$; and double transitivity carries the commutation of the lamps
at $0$ and at $e_1$ to every pair of lamps.
\end{proof}

\subsection{An embedding in a Brin--Thompson group}\label{sec:3V}

This subsection places $\Gamma$ inside a finitely presented simple group whose soficity is open, so that this group has a chunk with finite,
nearly exponential profile. Let $\mathfrak C=\{0,1\}^{\N}$ be Cantor space. For $n\ge1$, Brin's group $nV$~\cite{Brin2004} consists of the
homeomorphisms $h$ of $\mathfrak C^n$ given by a finite table of prefix replacements. Such a table is a pair of partitions of $\mathfrak C^n$ into
the same number of rectangles $u_1\mathfrak C\times\dots\times u_n\mathfrak C$, with $u_1,\dots,u_n$ finite binary words, matched in pairs. The
map $h$ sends each rectangle of the first partition onto its partner by replacing the $n$ prefixes and keeping the $n$ suffixes. Thus $1V$ is Thompson's group
$V$, and $nV$ embeds in $n'V$ for $n\le n'$ by acting trivially on the additional coordinates. The groups $nV$ are finitely presented, by Brin for
$2V$~\cite{Brin2005} and by Hennig and Matucci for every $n$~\cite{HennigMatucci2012}, and simple~\cite{Brin2010}. As for $V$, it is not known
whether they are sofic. The letter $V$ in these names has nothing to do with our space of configurations, which is also called $V$.

\begin{proposition}[Embedding in $3V$]\label{prop:3V}
$\Gamma$ embeds in $3V$, and hence in $nV$ for every $n\ge3$. Consequently $3V$ has a chunk $E_{3V}$ with
\[
\exp\Bigl(\frac{c\,r}{(\log r)^{12.75}}\Bigr)\ \le\ \sigma_{E_{3V}}(r)\ \le\ \exp(C\,r\log r)\qquad(r\ge2).
\]
\end{proposition}

\begin{proof}
\emph{Coding configurations.} For a finite set $S=\{j_1<\dots<j_m\}$ of positive integers, with $j_0=0$, let
\[
e(S)=\Bigl(\prod_{k=1}^m1\,0^{\,j_k-j_{k-1}-1}\,1\Bigr)\,0,
\]
so that $e(\emptyset)=0$, $e(\{1\})=110$ and $e(\{2\})=1010$. Read a word from the left: at a record boundary the letter $0$ ends the word and the
letter $1$ opens a record $10^k1$, which ends at its next letter $1$. So the words $e(S)$ form a prefix code. A finite word whose reading has not
ended can be extended to one that begins with some $e(S)$, by appending $0$ at a record boundary or $10$ inside a record. Hence the cylinders
$e(S)\mathfrak C$ are disjoint and their union is dense in $\mathfrak C$. For $v\in V$ and $i=1,2,3$ let $S_i(v)=\{j:v(i,j)=1\}$, let
\[
Z_v=e(S_1(v))\mathfrak C\times e(S_2(v))\mathfrak C\times e(S_3(v))\mathfrak C ,
\]
and let $\iota_v:\mathfrak C^3\to Z_v$ prefix the three codewords. The rectangles $Z_v$ are disjoint, and their union is dense in $\mathfrak C^3$.

\emph{Prepending a site.} Define $s_0,s_1:\mathfrak C\to\mathfrak C$ by $s_0(0\xi)=0\xi$, $s_0(1\xi)=10\xi$ and $s_1(\xi)=11\xi$. Their images
$0\mathfrak C\cup10\mathfrak C$ and $11\mathfrak C$ partition $\mathfrak C$, and, with $S+1=\{j+1:j\in S\}$,
\begin{equation}\label{eq:prepend}
s_0\bigl(e(S)\xi\bigr)=e(S+1)\,\xi,\qquad s_1\bigl(e(S)\xi\bigr)=e\bigl(\{1\}\cup(S+1)\bigr)\,\xi :
\end{equation}
an empty site in front lengthens the first gap by one, and an occupied site in front adds the record $11$.

\emph{The affine group.} For $b\in\{0,1\}$ put
\begin{gather*}
\hat P(x,s_b(y),z)=(s_b(x),y,z),\qquad \hat Q(x,y,s_b(z))=(s_b(x),y,z),\\
\hat T(s_b(x),y,z)=(s_{1-b}(x),y,z),
\end{gather*}
and let $\hat D$ fix $s_0(\mathfrak C)\times\mathfrak C^2$ pointwise and send $(11x,y,z)$ to $(11x',y,z)$, where $\hat T(x,y,z)=(x',y,z)$. These maps lie
in $3V$. In prefix tables, with $\epsilon$ the empty word, $\hat P$ replaces the prefixes $(\epsilon,11,\epsilon)$, $(0,0,\epsilon)$, $(1,0,\epsilon)$,
$(0,10,\epsilon)$ and $(1,10,\epsilon)$ by $(11,\epsilon,\epsilon)$, $(0,0,\epsilon)$, $(10,0,\epsilon)$, $(0,1,\epsilon)$ and $(10,1,\epsilon)$. $\hat Q$ does
the same with the second and third coordinates exchanged. $\hat T$ exchanges the first-coordinate prefixes $0\leftrightarrow110$ and
$10\leftrightarrow111$. $\hat D$ exchanges $110\leftrightarrow11110$ and $1110\leftrightarrow11111$ and keeps $0$ and $10$. The permutation $p$ moves
the first point of ray~$2$ to the first point of ray~$1$ and shifts the other points of these rays, and $q$ does the same with ray~$3$. The element
$t$ changes the value of a configuration at the point $1$, and $d$ adds the value at $1$ to the value at $2$. So~\eqref{eq:prepend} gives, for
every $v\in V$,
\[
\hat P\iota_v=\iota_{p(v)},\qquad \hat Q\iota_v=\iota_{q(v)},\qquad \hat T\iota_v=\iota_{t(v)},\qquad \hat D\iota_v=\iota_{d(v)} .
\]
By the proof of Lemma~\ref{lem:structure}(b), $p,q,t,d$ generate $A$. If a word in them represents $g\in A$, the same word in
$\hat P,\hat Q,\hat T,\hat D$ maps $\iota_v(\xi)$ to $\iota_{g(v)}(\xi)$ for all $v$ and $\xi$. If $g=1$, it is the identity on the dense set
$\bigcup_vZ_v$, hence the identity. So $p,q,t,d\mapsto\hat P,\hat Q,\hat T,\hat D$ defines a homomorphism $\rho:A\to3V$ with
\begin{equation}\label{eq:3V-equivariant}
\rho(g)\,\iota_v=\iota_{g(v)}\qquad(g\in A,\ v\in V),
\end{equation}
and $\rho$ is injective: if $g\ne1$, then $g(v)\ne v$ for some $v$, and $\rho(g)$ maps $Z_v$ onto the disjoint rectangle $Z_{g(v)}$.

\emph{The lamps.} Let $\beta_1=0$, $\beta_2=10$ and $\beta_3=11$, and let $S_3$ act faithfully on $\mathfrak C^3$ by the prefix replacements
$\theta(\lambda)(\beta_i\xi,y,z)=(\beta_{\lambda(i)}\xi,y,z)$. For $\lambda\in S_3$ and $v\in V$ let $\rho(\lambda_v)$ be $\iota_v\theta(\lambda)\iota_v^{-1}$ on
$Z_v$ and the identity off $Z_v$; it lies in $3V$, since the complement of a rectangle is a finite union of rectangles. Copies at different
configurations have disjoint supports and commute, each copy is faithful, and by~\eqref{eq:3V-equivariant}
$\rho(g)\rho(\lambda_v)\rho(g)^{-1}=\rho(\lambda_{g(v)})$ for $g\in A$. These are the defining relations of $\Gamma$ as a semidirect product, so $\rho$
extends to a homomorphism $\rho:\Gamma\to3V$. Suppose that $\rho(\gamma)=1$, where $\gamma$ has lamp part $\lambda:V\to S_3$ and affine part $g$. The
lamps preserve every $Z_v$, so $\rho(g)$ does too, and $g(v)=v$ for all $v$; so $g=1$. Then $\rho(\gamma)$ restricts on $Z_v$ to
$\iota_v\theta(\lambda(v))\iota_v^{-1}$, so $\lambda(v)=1$ for all $v$. Hence $\rho$ is injective. Explicitly, $\rho(a)$ exchanges the rectangles
$00\mathfrak C\times0\mathfrak C\times0\mathfrak C$ and $010\mathfrak C\times0\mathfrak C\times0\mathfrak C$. Also $\rho(c)=\hat D\rho(ab)$, where
$\rho(ab)$ permutes the rectangles with prefixes $(00,0,0)$, $(010,0,0)$ and $(011,0,0)$ cyclically in this order and $\hat D$ is the identity on
$Z_0$.

\emph{The chunk.} The sofic profile of a chunk depends only on its identity element and its partial multiplication
(Section~\ref{sec:criterion}), and an injective homomorphism preserves both. So the chunk $E_{3V}=\rho(E)$, for the chunk $E$ of
Theorem~\ref{thm:main-group}, has $\sigma_{E_{3V}}=\sigma_E$, and parts (a) and (b) of that theorem give the bounds.
\end{proof}

By Section~\ref{sec:thompson}, Thompson's group $V$ has chunks with finite superpolynomial profile, from its Houghton subgroups, and a chunk with
profile at least $\exp(c\,r/\log^2r)$, which is finite if $F$ is amenable. Proposition~\ref{prop:3V} gives finite, nearly exponential profile in
$3V$ with no assumption. Whether $\Gamma$ embeds in $V$ or in $2V$ is open (Section~\ref{sec:open}).

\subsection{The F\o lner function}\label{sec:folner}

The hash models of Proposition~\ref{prop:upper} are far smaller than any F\o lner set of $\Gamma$, and this subsection measures the gap. For
$r\ge1$ let $\textup{Føl}(r)$ be the least size of a finite nonempty set $F\subset\Gamma$ with $|Fs\setminus F|\le|F|/r$ for
$s\in\{p,q,t,c,a\}$.

\begin{proposition}[The F\o lner function]\label{prop:folner-lower}
For $r\ge12$,
\[
2^{\,2^{\lfloor r/12\rfloor-1}}\ \le\ \textup{Føl}(r)\ \le\ 2^{\,2^{3r+14}} .
\]
Thus the Følner function of $\Gamma$ is doubly exponential, while by Theorem~\ref{thm:main-group}(b) the sofic profile of every chunk is
$\exp(O(r\log r))$.
\end{proposition}

\begin{proof}
\emph{Lower bound.} Let $F$ be such a set and $l=\lfloor r/12\rfloor\ge1$. For $v\in\{0,1\}^l$ the word
$T_v=t^{v_1}p\,t^{v_2}p\cdots p\,t^{v_l}p^{1-l}$ is freely equal to $\prod_{j=1}^lp^{j-1}t^{v_j}p^{1-j}$, so it represents the translation by
$\sum_jv_je_j$, and it has length less than $3l$. The elements $a_v=T_vaT_v^{-1}$ are the copies of $a$ at $2^l$ distinct configurations, so they
are commuting involutions and form a basis of an elementary abelian group $Q$ of order $2^{M}$, $M=2^l$. Each $a_v$ is a product of fewer than
$6l$ generators and inverses. Since $|Fxy\setminus F|\le|Fx\setminus F|+|Fy\setminus F|$ and $|Fx^{-1}\setminus F|=|Fx\setminus F|$, we get
$|Fa_v\setminus F|\le6l|F|/r\le|F|/2$, and so $\sum_v|Fa_v\setminus F|\le M|F|/2$.

On the other hand, for every finite $S\subset Q$,
\begin{equation}\label{eq:cube}
\sum_v|Sa_v\setminus S|\ \ge\ |S|\,(M-\log|S|).
\end{equation}
Indeed, the left side is $M|S|-2e(S)$, where $e(S)$ is the number of pairs $\{x,xa_v\}\subset S$, and $e(S)\le\frac12|S|\log|S|$ by the
edge-isoperimetric inequality of the cube~\cite{Harper1964}. We recall the proof, by induction on $M$. Split $S$ by the last coordinate into
$S_0,S_1$ of sizes $s_0\le s_1$. Then $e(S)\le e(S_0)+e(S_1)+s_0\le\frac12(s_0\log s_0+s_1\log s_1)+s_0$, and this is at most $\frac12s\log s$ with
$s=s_0+s_1$, because the binary entropy function $h$ satisfies $h(x)\ge2x$ for $0\le x\le\frac12$. The right multiplications by the $a_v$ preserve
each coset $xQ$; applying~\eqref{eq:cube} to the sets $x^{-1}(F\cap xQ)$ and using $|F\cap xQ|\le|F|$ gives
$\sum_v|Fa_v\setminus F|\ge|F|(M-\log|F|)$. Hence $\log|F|\ge M/2=2^{l-1}$.

\emph{Upper bound.} Let $m=\lceil(r-1)/2\rceil$, $R=2m+3$ and $F=\bigcup_{z\in B_m}K_Rg_z$, a disjoint union because $\pi(g_z)=z$. By
Lemma~\ref{lem:clock}, $g_zs\in K_Rg_{z+\pi(s)}$ for $z\in B_m$ and each generator $s$, so $|Fs\setminus F|\le|K_R|(2m+1)=|F|/(2m+1)\le|F|/r$. With
$D=|Y_R|=6m+9$ we have $|K_R|=6^{2^D}\cdot2^D\cdot|\GL_D(\F_2)|$, so $\log|F|\le2^D\log6+D+D^2+2\log(2m+1)\le2^{D+2}\le2^{3r+14}$.
\end{proof}

The group $\Gamma$ is not LEF. This follows from $H_3\le\Gamma$ in the standard way. $H_3$ is finitely presented and not residually finite,
because the finitary alternating group of an infinite set is simple. A finitely presented LEF group is residually
finite~\cite{VershikGordon1997}, and subgroups of LEF groups are LEF. The chunk $E$ shows it too. Its profile is unbounded by
Theorem~\ref{thm:main-group}(a), and a chunk with unbounded profile embeds in no finite group (Section~\ref{sec:criterion}).
\section{The ladder below $\Gamma$}\label{sec:ladder}

Before $\Gamma$ we found two groups that each supply part of what the criterion needs. We record them here for comparison, and because each
leaves an open problem; nothing in this section is used in the proof of Theorem~A or in Part~II. The dyadic group $\mathcal G$ of
Section~\ref{sec:dyadic} is elementary amenable and transports copies of $S_3$ by the powers of one element, with far commutators of area at most
quadratic. Corollary~\ref{cor:cyclic} turns this into the exponent $1/3$, the most that transport along a ray of $H_3$ can give
(Section~\ref{sec:criterion}). The subgroup $\mathcal L$ of Thompson's group $V$ in Section~\ref{sec:thompson} transports along the exponentially
growing orbits of Thompson's group $F$: words of length $O(\log N)$ reach $N$ sites, and the quadratic Dehn function of $F$ fills all pairwise
commutators cheaply. That comes within a factor $\log^2r$ in the exponent of the ceiling of Theorem~\ref{thm:criterion}, but inside a group that
contains $F$, whose soficity is unknown. The proof that $\mathcal L$ is hyperlinear exactly when $F$ is uses tools that nothing else in the paper
needs, and it is in Appendix~\ref{app:hyperlinear}.

\subsection{An amenable group at exponent \texorpdfstring{$1/3$}{1/3}}\label{sec:dyadic}

The group $\mathcal G$ is a lamplighter over the dyadic rationals $\mathbb Z[\frac12]$: the lamps are copies of $S_3$, and the lighter is
generated by the shift of the integers and by doubling. In this section a group acting on a set $X$ acts on $\bigoplus_{x\in X}S_3$ by carrying
the copy of $S_3$ at $x$ to the copy at $g(x)$ through the identity of $S_3$.

Let $X=\mathbb Z[\frac12]$, and let $t$ and $u$ be the permutations of $X$ given by $t(x)=x+1$ for $x\in\mathbb Z$, $t(x)=x$ for
$x\notin\mathbb Z$, and $u(x)=2x$. Let $\mathcal B=\langle t,u\rangle$ and $\mathcal G=\bigl(\bigoplus_{x\in X}S_3\bigr)\rtimes\mathcal B$, and let $a,b$ be
generating involutions of the copy of $S_3$ at $0$. Put $s=utu^{-1}$ and $v=tst^{-1}$. Then $s$ adds $2$ to the even integers and $v$ adds $2$
to the odd integers, each fixing every other point, so $t^2=sv$ and $[s,v]=1$. Moreover $u$ and $v$ fix $0$, and $s$ fixes $1=t(0)$. Since
$\mathcal B$ acts transitively on $X$, the group $\mathcal G$ is generated by $a$, $b$, $t$ and $u$. By these identities the following seventeen words
in $a,b,t,u,s,v$ represent $1$ in $\mathcal G$; they form the set $R_{\mathcal G}$:
\[
a^2,\ b^2,\ (ab)^3,\quad s^{-1}utu^{-1},\ v^{-1}tst^{-1},\ t^2v^{-1}s^{-1},\ [s,v],\quad [u,x],\ [v,x],\ [s,txt^{-1}],\ [x,tyt^{-1}],
\]
where $x,y\in\{a,b\}$. Theorem~\ref{thm:dyadic} also records that $\mathcal G$ is not solvable, so the exponent $1/3$ is attained outside the
solvable class.

\begin{theorem}[The dyadic group]\label{thm:dyadic}
The group $\mathcal G$ is elementary amenable and not solvable, and every commutator $[x,t^kyt^{-k}]$ with $x,y\in\{a,b\}$ and
$k\ge1$ has area at most $5k^2$ with respect to $R_{\mathcal G}$. Consequently the chunk $E_{\mathcal G}$ that Theorem~\ref{thm:criterion}
associates with $R_{\mathcal G}$ has finite sofic profile with
\[
\log\sigma_{E_{\mathcal G}}(r)\ \ge\ \frac1{16}\Big\lfloor\Big(\frac r{75\cdot10^8}\Big)^{1/3}\Big\rfloor\qquad(r>1),
\]
every unitary assignment $\phi$ as in Theorem~\ref{thm:criterion} has $\log D\ge\frac{\Delta^2}{12}\bigl((5K_\Delta e)^{-2/5}-1\bigr)$, and
$\mathcal G$ is not LEF.
\end{theorem}

\begin{proof}
\emph{Amenability.} For $j\in\mathbb Z$ let $t_j=u^jtu^{-j}$; it adds $2^j$ to the points of $2^j\mathbb Z$ and fixes the others. The group $\mathcal K$
generated by all $t_j$ is normalized by $u$ and contains $t$, so it is normal in $\mathcal B$ and $\mathcal B/\mathcal K$ is cyclic. For $n\ge1$ the elements $t_j$ with
$|j|\le n$ move only points of $2^{-n}\mathbb Z$, and they commute with the translation $x\mapsto x+2^n$ of $2^{-n}\mathbb Z$, which has $4^n$
orbits. A permutation of $2^{-n}\mathbb Z$ that commutes with this translation maps each orbit onto an orbit by a translation, so these $t_j$ generate a subgroup of
$\mathbb Z^{4^n}\rtimes\Sym_{4^n}$, which is virtually abelian. Hence $\mathcal K$ is a directed union of virtually abelian groups, $\mathcal B$ is elementary
amenable, and so is $\mathcal G$, because $\bigoplus_XS_3$ is locally finite.

\emph{Not solvable.} For $n\ge1$ the elements $t_0,\dots,t_{n-1}$ preserve $\mathbb Z$ and commute with $x\mapsto x+2^n$, so they act on
$\mathbb Z/2^n$, where they permute the classes modulo $2^i$ for every $i\le n$. So they generate a subgroup $\Gamma_n$ of the iterated wreath
product $\mathsf W_n$ of $n$ copies of $\mathbb Z/2$, which has order $2^{2^n-1}$. On the even classes, identified with $\mathbb Z/2^{n-1}$ by
halving, $t_1,\dots,t_{n-1}$ act as $t_0,\dots,t_{n-2}$ act on $\mathbb Z/2^{n-1}$, and they fix the odd classes. Conjugation by $t_0$ carries
this action to the odd classes, and $t_0$ exchanges the two halves. So $|\Gamma_n|\ge2|\Gamma_{n-1}|^2$, and $\Gamma_n=\mathsf W_n$ by induction from $\Gamma_1=\mathsf W_1=\mathbb Z/2$. The derived
subgroup of $\mathsf W_n$ lies in $\mathsf W_{n-1}\times\mathsf W_{n-1}$ and projects onto each factor, since the commutator of $(g,1)$ with the
exchange is $(g^{-1},g)$. By induction the $(n-1)$st derived subgroup of $\mathsf W_n$ is nontrivial. Since $\Gamma_n$ is a quotient of the subgroup
$\langle t_0,\dots,t_{n-1}\rangle$ of $\mathcal B$, the derived length of $\mathcal B$ is at least $n$ for every $n$, so neither $\mathcal B$ nor
$\mathcal G$ is solvable.

\emph{Area.} Let $A_{\mathcal G}(k)$ be the largest area of $[x,t^kyt^{-k}]$ over $x,y\in\{a,b\}$, so that $A_{\mathcal G}(1)=1$. Five kinds of move
cost one cell each.
\begin{itemize}[nosep,leftmargin=*]
\item Replacing $t^2$ by $sv$.
\item Exchanging adjacent letters $s^{\pm1}$ and $v^{\pm1}$.
\item Replacing $s$ by $utu^{-1}$.
\item Moving a letter $x^{\pm1}$, $x\in\{a,b\}$, across $u^{\pm1}$ or $v^{\pm1}$.
\item Moving $txt^{-1}$ or its inverse across $s^{\pm1}$.
\end{itemize}
The commutator contains four powers $t^{\pm k}$.

For $k=2m$, write each $t^{\pm2m}$ as $(s^mv^m)^{\pm1}$, with $m$ replacements and $m(m-1)/2$ exchanges each, $2m^2+2m$ cells in all. Moving
$y^{\pm1}$ across $v^{\pm m}$ ($2m$ cells) leaves $[x,s^mys^{-m}]$. Replacing its $4m$ letters $s^{\pm1}$ gives
$[x,ut^mu^{-1}yut^{-m}u^{-1}]$, and moving $y^{\pm1}$ and $x^{\pm1}$ across $u$ (four cells) gives $u[x,t^myt^{-m}]u^{-1}$. Hence
$A_{\mathcal G}(2m)\le A_{\mathcal G}(m)+2m^2+8m+4$.

For $k=2m+1$, put $z=tyt^{-1}$, so that the commutator is freely $[x,t^{2m}zt^{-2m}]$. Write each $t^{\pm2m}$ as $(v^ms^m)^{\pm1}$, now with
$m(m+1)/2$ exchanges each, $2m^2+6m$ cells in all. Moving $z^{\pm1}$ across $s^{\pm m}$ ($2m$ cells) leaves $[x,v^mzv^{-m}]$, moving
$x^{\pm1}$ across $v^{\pm m}$ ($2m$ cells) leaves $v^m[x,z]v^{-m}$, and $[x,z]$ is a relation. Hence $A_{\mathcal G}(2m+1)\le2m^2+10m+1$.

By induction $A_{\mathcal G}(k)\le5k^2$, since $A_{\mathcal G}(2)\le15$, $A_{\mathcal G}(m)+2m^2+8m+4\le7m^2+8m+4\le20m^2$ for $m\ge1$, and
$2m^2+10m+1\le5(2m+1)^2$.

\emph{Profile.} The words of $R_{\mathcal G}$ have length at most $8$, so Corollary~\ref{cor:cyclic} with $C=5$, $p=2$ and $l=8$ gives the
bounds. Since $\mathcal G$ is amenable, it is sofic and every chunk has finite profile. The profile of $E_{\mathcal G}$ is unbounded, so
$E_{\mathcal G}$ embeds in no finite group, and $\mathcal G$ is not LEF (Section~\ref{sec:criterion}).
\end{proof}

The contrast with $H_3$ is instructive. There the doubling endomorphism of~\cite[Section~3]{DouglasMghirbiA}, which differs by a conjugation
from the doubling map of Burillo, Cleary, Martino and R\"over~\cite[Section~7]{BCMR2016}, multiplies a reweighted area by $107$ at each step;
hence the exponent $\log107$. In $\mathcal G$ the dilation $u$ relates distance $2m$ to distance $m$ exactly, and the two commuting shifts cost
$O(m^2)$ cells. The orbit of $0$ under $\mathcal B$ grows exponentially, as the orbits of Thompson's group $F$ do. Whether $\mathcal G$ also
supports the transport of many copies that we use for $F$ in Section~\ref{sec:thompson} is open (Section~\ref{sec:open}).

\subsection{Thompson's group \texorpdfstring{$V$}{V}}\label{sec:thompson}

R\"over proved that every Houghton group embeds in Thompson's group $V$~\cite[Proposition~2.6]{Rover1999}, and Genevois extended the
embedding to generalised Houghton groups~\cite[Corollary~5.11]{Genevois2019}. Since $H_3$ is amenable, Thompson's group $V$ therefore has a
chunk whose profile is finite at every $r$ and superpolynomial by~\cite[Theorem~1.1]{DouglasMghirbiA}. A different subgroup of $V$ gives a much
larger bound.

Let $\mathbb D$ be the set of dyadic rationals in $(0,1)$, and let Thompson's group $F$ act on $[0,1]$ by its piecewise-linear
homeomorphisms. Let $\mathcal L=\bigl(\bigoplus_{x\in\mathbb D}S_3\bigr)\rtimes F$. To embed $\mathcal L$ in $V$ we double every $1$ of a binary
word, so that the pattern $10$ never occurs inside a coded word. The first $10$ then marks a gap, the gaps are indexed by the dyadic rationals,
and $F$ acts on the gaps as it acts on $\mathbb D$.

\begin{lemma}[Gap coding]\label{lem:gap-coding}
$\mathcal L$ embeds in $V$.
\end{lemma}

\begin{proof}
Let $\mathsf c$ be the substitution $0\mapsto0$, $1\mapsto11$ on finite binary words. An infinite binary word either is a concatenation of
blocks $0$ and $11$, or it has a first occurrence of $10$ at a block boundary. So Cantor space is the disjoint union of the set $\mathcal C$ of
block sequences and the cylinders $Z_w=\mathsf c(w)10\{0,1\}^{\mathbb N}$, one for each finite word $w$. Index $Z_w$ by the dyadic
$x_w=0.w1\in\mathbb D$ in binary; this is a bijection. Let $f\in F$ have a tree-pair diagram with leaves $w_1,\dots,w_n$ and $w'_1,\dots,w'_n$,
so that $f(0.w_kz)=0.w'_kz$. Define $\hat f$ by the prefix replacements $\mathsf c(w_k)\mapsto\mathsf c(w'_k)$ and
$\mathsf c(v)10\mapsto\mathsf c(v')10$, where the interior vertices $v$ of the first tree are matched in left-to-right order, the $k$-th with the $k$-th, with the interior vertices $v'$ of
the second. Both families of prefixes are complete prefix codes, because splitting a leaf $w$ replaces $\mathsf c(w)$ by $\mathsf c(w0)$,
$\mathsf c(w)10$ and $\mathsf c(w1)$. Expanding a matched pair of leaves adds a matched pair of gaps and changes nothing else, so $\hat f$ depends
only on $f$. It maps $Z_w$ onto $Z_{w'}$ with $x_{w'}=f(x_w)$: below a leaf this is the formula for $f$, and interior vertices correspond to the
breakpoints between leaves. On $\mathcal C$ it acts as $f$ acts on binary sequences, so $f\mapsto\hat f$ is an injective homomorphism $F\to V$.
Inside each $Z_w$ let $S_3$ permute the three cylinders $\mathsf c(w)10y\{0,1\}^{\mathbb N}$, $y\in\{0,10,11\}$. Copies with different $w$ have
disjoint supports, and $\hat f$ carries the copy at $x$ to the copy at $f(x)$. An element of the group they generate that acts trivially acts
trivially on $\mathcal C$, so its $F$-part is trivial, and then each of its $S_3$ coordinates is trivial.
\end{proof}

Let $a,b$ generate the copy of $S_3$ at $\frac12$, let $t\in F$ be
\[
t(x)=\begin{cases}2x,&0\le x\le\frac14,\\ x+\frac14,&\frac14\le x\le\frac12,\\ \frac{x+1}2,&\frac12\le x\le1,\end{cases}
\]
so that $t(\frac12)=\frac34$, and for $m\ge1$ let $\mathbb D_m$ be the set of the $2^m-1$ points of $\mathbb D$ with denominator at most
$2^m$.

The group $\mathcal L$ is finitely presented. This is Cornulier's criterion for permutational wreath products~\cite[Theorem~1.1]{Cornulier2006},
applied to the action of $F$ on $\mathbb D$; he checks its hypotheses in~\cite[Example~3.4]{Cornulier2006}. Thompson's group $F$ is finitely
presented, and it acts transitively on increasing pairs of points of $\mathbb D$~\cite[Lemma~4.2]{CFP1996}, so it has finitely many orbits on
$\mathbb D\times\mathbb D$. The stabilizer $F_{1/2}$ of $\frac12$ is the product of the elements supported in $[0,\frac12]$ and those supported
in $[\frac12,1]$, two copies of $F$, so it is finitely generated. Nothing in Proposition~\ref{prop:thompson} uses finite presentation.

\begin{proposition}[Thompson transport]\label{prop:thompson}
The group $\mathcal L$ contains $F$ and contains no Houghton group $H_m$, so the bounds below do not come from the
Houghton subgroups of $V$. There are a finite set $R$ of relations of $\mathcal L$, containing $a^2,b^2,(ab)^3$, and a constant $C$ with the
following property: for every $m\ge1$ there are words $w_p$, $p\in\mathbb D_m$, with $w_p(\frac12)=p$ and
\[
\mathrm{Area}_R\bigl([w_pxw_p^{-1},w_qyw_q^{-1}]\bigr)\le Cm^2\qquad(p\ne q\in\mathbb D_m,\ x,y\in\{a,b\}).
\]
Consequently a chunk $E$ of $\mathcal L$, and hence of $V$, has $\log\sigma_E(r)\ge c\,r/\log^2r$ for all large $r$, and for $0<\Delta\le1$
every unitary assignment as in Theorem~\ref{thm:criterion} has $\log D\ge c_\Delta\,e^{-2}/\log^4(1/e)$ for all small $e$, where $c>0$ is a
constant and $c_\Delta>0$ depends only on $\Delta$.
\end{proposition}

\begin{proof}
\emph{Transport.} A standard dyadic interval is one of the form $[k2^{-j},(k+1)2^{-j}]$. Let two partitions of $[0,1]$ into standard dyadic
intervals have the same number $n$ of intervals. The map sending the $i$th interval of the first affinely onto the $i$th interval of the second
lies in $F$ and has a tree-pair diagram with $n-1$ carets. Every standard dyadic interval can be split into any given number of standard
dyadic intervals. For $p<q$ in $\mathbb D_m$, the binary digits of $p$ split $[0,p]$ and $[p,1]$ into at most $m$ standard dyadic
intervals each. The interval $[p,q]$, cut at its point $c$ with the shortest binary expansion, splits into at most $m$ standard dyadic intervals
on each side of $c$. Hence there are $w_p\in F$ with $w_p(\frac12)=p$ and fewer than $2m$ carets, and $g_{pq}\in F$ with $g_{pq}(\frac12)=p$, $g_{pq}(\frac34)=q$ and fewer than $4m$
carets. The word length in $F$ is comparable to the number of carets of the reduced tree-pair diagram~\cite{BCS2001}. Hence these elements
are represented by words of length $O(m)$ in the generating set of the finite presentation of $F$ that the set $R$ of relations uses. The Dehn
functions of two finite presentations of $F$ are equivalent, so Guba's quadratic bound holds for it. We write $w_p$ also for such a word. An
element fixing $\frac12$ maps $[0,\frac12]$ to itself, so the left and right subtrees of its reduced diagram are diagrams for its two
restrictions, which are elements of $F$ after affine rescaling, with fewer carets in total. Hence $F_{1/2}\cong F\times F$ is undistorted in
$F$: its elements of length $n$ in $F$ are words of length $O(n)$ in a fixed finite set of generators $h_j$ of $F_{1/2}$.

\emph{Relations.} Let $R$ consist of a finite presentation of $F$ on a generating set that includes $t$ and the $h_j$, the words
$a^2,b^2,(ab)^3$, the commutators $[x,h_j]$, and the commutators $[x,tyt^{-1}]$, for $x,y\in\{a,b\}$. They hold in $\mathcal L$: each $h_j$
fixes $\frac12$, and the copies at $\frac12$ and at $t(\frac12)=\frac34$ commute.

\emph{Area.} Let $p<q$ in $\mathbb D_m$ and $g=g_{pq}$. The elements $g^{-1}w_p$ and $t^{-1}g^{-1}w_q$ fix $\frac12$ and have length
$O(m)$, so they are represented by words $\omega,\omega'$ of length $O(m)$ in the $h_j$. The words $w_p^{-1}g\omega$ and $w_q^{-1}gt\omega'$ are null
in $F$
and have length $O(m)$, so Guba's quadratic Dehn function~\cite{Guba2006} fills each with $O(m^2)$ cells. Replacing the four
occurrences of $w_p^{\pm1}$ in $[w_pxw_p^{-1},w_qyw_q^{-1}]$ by $(g\omega)^{\pm1}$, and the four of $w_q^{\pm1}$ by $(gt\omega')^{\pm1}$, therefore
costs
$O(m^2)$ cells. Moving $x^{\pm1}$ across $\omega$ and $y^{\pm1}$ across $\omega'$ with the relations $[x,h_j]$ costs $O(m)$ cells and leaves
$g[x,tyt^{-1}]g^{-1}$, a conjugate of a relation. For $p>q$ the commutator is the inverse of the one with the two copies exchanged. So every
area is $O(m^2)$.

\emph{Profile.} Apply Theorem~\ref{thm:criterion} with $N=2^m-1$ and $\Lambda=Cm^2$, where $C\ge1$, and let $l$ be the largest
length of a word of $R$. Given $r$, let $m$ be the largest
integer with $2^mm^2\le r/(10^8C(2l-1))$; then $2^m$ is at least of order $r/\log^2r$, and $\log\sigma_E(r)\ge(2^m-1)/16$. Given
$e$, let $m$ be the largest integer with $2^{m/2}m^2\le1/(CK_\Delta e)$; then $2^m$ is at least of order
$(CK_\Delta)^{-2}e^{-2}/\log^4(1/e)$, and $\log D\ge(2^m-1)\Delta^2/12$.

\emph{No Houghton subgroup.} A torsion element of $\mathcal L$ maps to a torsion element of the torsion-free group $F$, so it lies in
$\bigoplus_{\mathbb D}S_3$, which has exponent six. Every $H_m$ contains a five-cycle. Finally, $\mathcal L$ embeds in $V$ by
Lemma~\ref{lem:gap-coding}.
\end{proof}

Since $\mathcal L$ contains $F$, whose soficity is open, this chunk may have no permutation models at all; if it has any, they are very large.
In particular, either $V$ is not sofic, or its models of this chunk need $2^{c\,r/\log^2r}$ points, close to the largest bound that
Theorem~\ref{thm:criterion} can give. If $F$ is amenable, so is $\mathcal L$, and then this chunk has finite profile. If $F$ is hyperlinear, so
is $\mathcal L$ (Proposition~\ref{prop:thompson-hyperlinear}), and then unitary models of the chunk exist at every defect.

Slofstra's bounds are of the same kind: they bound models that may not exist, since his group is not known to be hyperlinear. Any unitary
models of it need $\log D\gtrsim e^{-c}$ for every $c<\frac12$~\cite[Theorem~1.2 and the remark after it]{Slofstra2018}, and the proof
of~\cite[Corollary~1.3]{Slofstra2018} shows that any permutation models of a suitable chunk need $\log\sigma\gtrsim r^c$ for every $c<\frac14$.
These profiles may be infinite. The group $\mathcal L$ gives $e^{-2}$ and $r$ up to logarithms.

Since $\mathbb D$ is the coset space $F/F_{1/2}$, the group $\mathcal L=\bigl(\bigoplus_{F/F_{1/2}}S_3\bigr)\rtimes F$ has the shape
$\bigl(\bigoplus_{G/H}A\bigr)\rtimes G$ of the generalised wreath products of Kun and Thom~\cite[Theorem~A]{KunThom2026}. They prove such a
product non-sofic for $A=\Z/2$ when $H<G$ is infranormal and not normal, that is, the elements $g$ with $gHg^{-1}\subseteq H$ generate $G$, and
both $H$ and $G$ have property (T). Our input is different: cheap transport of non-abelian copies of $S_3$. Soficity and hyperlinearity pass to
wreath products whose acting group is sofic~\cite{HayesSale2018}. They pass to permutational wreath products when the acting group is sofic
(respectively hyperlinear) and the action is sofic~\cite[Theorems~3.6 and~3.8]{GKEP2025}, \cite{AlekseevBradford2026}. Since $F$ acts on
$\mathbb D$ transitively and faithfully, a sofic action of $F$ on $\mathbb D$ would already make $F$ sofic~\cite[Theorem~2.12]{GKEP2025}, so
these results give nothing here. Proposition~\ref{prop:thompson-hyperlinear} assumes only that $F$ is hyperlinear.

\begin{proposition}[Hyperlinearity transfer]\label{prop:thompson-hyperlinear}
The group $\mathcal L$ is hyperlinear if and only if $F$ is.
\end{proposition}

The proof, in Appendix~\ref{app:hyperlinear}, needs a copy of $S_3$ in the ultrapower $L(F)^\omega$ at $\frac12$ that commutes with the
stabilizer $F_{1/2}$. Near $\frac12$ the stabilizer acts by dilations. The proof therefore places a copy of $S_3$ in each of a sequence of
intervals shrinking to $\frac12$, and selects among them by a random index whose law is invariant under dilation in the limit.

\part{The channel (Theorem B)}
\section{Causal services and gadget channels}\label{sec:channel}\label{sec:channel-setting}

Part~II can be read on its own by a reader who comes from quantum information. From Part~I it needs three definitions, all in
Section~\ref{sec:criterion}. A \emph{chunk} of a group is a finite subset containing $1$, with the partial multiplication it inherits. A
\emph{sofic model} of the chunk at accuracy $1/r$ assigns a permutation of a finite set to each element, with products correct and distinct
elements separated on all but a fraction $1/r$ of the points. The \emph{profile} $\sigma_E(r)$ is the least number of points of such a model.
The configuration lamps enter mainly through two facts, quoted where they are used: their pair areas (Theorem~\ref{thm:pair-area}) and their
profile, Theorem~A in the precise form of Theorem~\ref{thm:main-group}. Concretely, the lower bounds use one black box, the transport of
Section~\ref{sec:lamp-channel}: for each $N$ there are $2^N$ words of length at most $63N^2$ whose copies of $S_3$ commute pairwise at area at
most $C_1N^\kappa$, with $\kappa=6+\log107$ (the introduction says where $107$ comes from). The upper bounds use another, the hash models
behind Theorem~A(b), which give sofic models with $\exp(O(r\log r))$ points.

A finite set of relations of a group becomes a quantum channel of fixed dimension. Each relation contributes a few small unitary blocks; the
channel evaluates the blocks in the regular representation of the group and traces out the group von Neumann algebra. Every channel on
$\mathbb C^d$ has a Stinespring dilation, a unitary on the input and a pure environment of dimension at most $d^2$ that is discarded after
use. Applied $n$ times with a fresh environment each time, the channel needs almost no memory, but the discarded environments carry away
entropy that fresh pure qubits paid for. We charge for that purity, and the memory is never reset. The question is how much memory a device
needs when its purity is far below what the Stinespring device spends.

In~\cite{DouglasMghirbiB} we showed that the memory needed to apply such a channel repeatedly, releasing each output before the next input
arrives, is bounded below by the size of approximate representations of its relations. Part~II sharpens that link, so that it loses a single
square root, and applies it to the configuration lamps. The outcome is Theorem~B, stated precisely as Theorem~\ref{thm:main-channel}: an
explicit channel in the closure of finite tracial channels whose repeated use with logarithmic purity needs memory $n^{1-o(1)}$, while memory
$O(n\log n)$ suffices. The lower bounds come from the sharpened link, the weighted one-round theorem of Section~\ref{sec:weighted-statement}.
The upper bounds come from the sofic models of $\Gamma$, which give devices directly (Proposition~\ref{prop:sofic-upper}) and, at the rank
rate, through Theorem~\ref{thm:streaming}.

We recall the framework of~\cite{DouglasMghirbiB} and keep its notation.

\paragraph{Channels.} Channels act on $M_d=M_d(\mathbb C)$. The normalized Choi state of a channel $\Psi$ is
$\mathcal J(\Psi)=(\Psi\otimes\mathrm{id})(|\Omega\rangle\langle\Omega|)$, with $|\Omega\rangle=d^{-1/2}\sum_i|ii\rangle$, and
$\ddia(\Psi,\Psi')=\frac12\|\Psi-\Psi'\|_\diamond$ is the diamond distance. We write $\tau_M$ for the normalized trace on $M_M$. A channel $\Psi$ has a
\emph{finite tracial factorization} with bath dimension $M$ if $\Psi(X)=(\mathrm{id}\otimes\tau_M)[U(X\otimes I_M)U^*]$ for a unitary $U$ on
$\mathbb C^d\otimes\mathbb C^M$, and $\cK_d$ denotes the closure of the set of channels on $M_d$ with a finite tracial factorization of some bath
dimension. It is also the closure of the convex hull of that set~\cite[Proposition~3.8(ii)]{Musat2025}, and it consists of unital channels. Channels
with a finite tracial factorization are exactly the noisy operations of~\cite{HHO2003} from $M_d$ to $M_d$, and $\cK_d$ is their closure.

\paragraph{Causal services.} A \emph{causal service} of a channel $\Phi$ on $M_d$ for $n$ rounds specifies a \emph{device}, its physical
realization, which serves one use of the channel in each round. The device has a memory register of $Q$ qubits, started in a fixed state
$\beta_0$. In each round $t$ it applies a unitary $U_t$, fixed in advance, to the input and the memory and then releases the $d$-dimensional output. After each round
but the last, a prescribed set of memory qubits is exported for good and replaced by the same number of imported qubits in a fixed state.
That state may be correlated within one step but is independent of everything else. An observer chooses the inputs adaptively, possibly
entangled with references of its own. The \emph{error} of the service is the largest half trace distance between the observer's final state in
this experiment and in the ideal experiment with $n$ independent uses of $\Phi$. Figure~\ref{fig:device} shows the device and its rounds.

The resources are the memory $Q$, the exchange $J$ and the purity $P+C$. The \emph{exchange} is the total number of exported qubits;
exporting is the only way the device can discard a qubit without handing it to the observer. The \emph{purity} is an entropy deficit in bits,
not $\Tr\rho^2$. Here $P=Q-S(\beta_0)$ is the deficit of the initial memory, and $C$ is the sum over the imports of their number of qubits
minus their entropy. A pure qubit is a resource, since it costs work to prepare, while a maximally mixed one is free; the Stinespring device
of a channel with Choi rank $r_*$ spends about $\log r_*$ bits of purity in every round. Every retained qubit counts; there is no free reset,
measurement record or postselection.

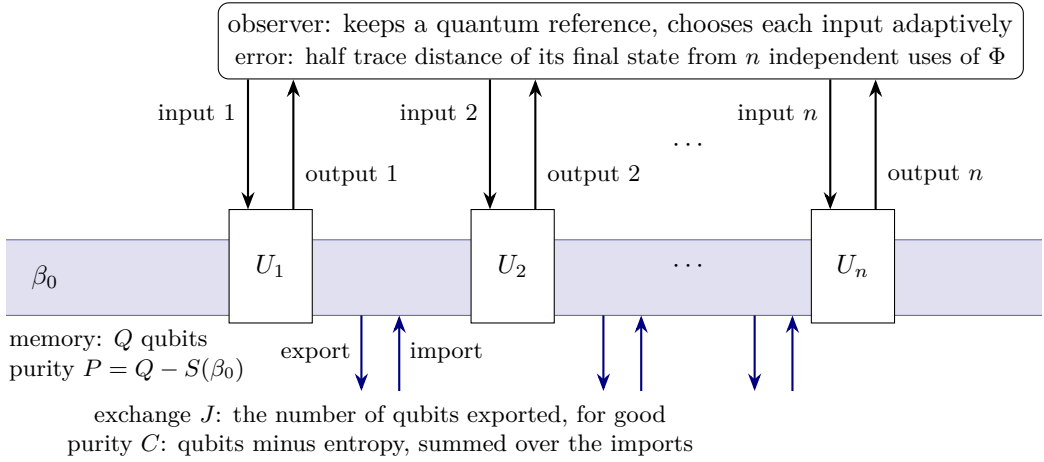
\begin{figure}[t]
\centering
\begin{tikzpicture}[x=1cm,y=1cm,font=\small,>=Stealth,
  box/.style={draw,fill=white,minimum width=1.1cm,minimum height=1.5cm,inner sep=2pt},
  wire/.style={->,thick},
  xw/.style={->,thick,accent},
  lbl/.style={font=\footnotesize,inner sep=1pt}]
\colorlet{accent}{blue!50!black}
\fill[accent!12] (0,-0.5) rectangle (13.7,0.5);
\draw[accent!60] (0,-0.5) -- (13.7,-0.5) (0,0.5) -- (13.7,0.5);
\node[anchor=west] at (0.2,0) {$\beta_0$};
\node[lbl,anchor=north west,align=left] at (0,-0.65) {memory: $Q$ qubits\\ purity $P=Q-S(\beta_0)$};
\node[draw,rounded corners,align=center,inner ysep=4pt,anchor=south west] (obs) at (2.8,2.62)
  {observer: keeps a quantum reference, chooses each input adaptively\\
   \footnotesize error: half trace distance of its final state from $n$ independent uses of $\Phi$};
\foreach \x/\t in {3.5/1,6.7/2,11.2/n}{
  \node[box] at (\x,0.15) {$U_{\t}$};
  \draw[wire] (\x-0.3,2.62) -- (\x-0.3,0.9);
  \draw[wire] (\x+0.3,0.9) -- (\x+0.3,2.62);
  \node[lbl,anchor=east] at (\x-0.42,2.15) {input $\t$};
  \node[lbl,anchor=west] at (\x+0.42,1.35) {output $\t$};
}
\node at (9.05,0.15) {$\cdots$};
\node at (9.05,1.75) {$\cdots$};
\foreach \x in {4.7,7.9,9.9}{
  \draw[xw] (\x,-0.5) -- (\x,-1.5);
  \draw[xw] (\x+0.5,-1.5) -- (\x+0.5,-0.5);
}
\node[lbl,anchor=east] at (4.6,-1.0) {export};
\node[lbl,anchor=west] at (5.3,-1.0) {import};
\node[lbl,anchor=north,align=center] at (4.95,-1.65) {exchange $J$: the number of qubits exported, for good\\ purity $C$: qubits minus entropy, summed over the imports};
\end{tikzpicture}
\caption{A causal service of $\Phi$ for $n$ rounds. In each round the input enters, a fixed unitary $U_t$ acts on it and the memory, and the output
is released before the next input arrives; after each round but the last, a prescribed set of memory qubits is exported for good and replaced by
the same number of fresh qubits in a fixed state. The resources are the memory $Q$, the exchange $J$, the purity $P+C$, and the error, the
largest half trace distance between the observer's final state and its final state after $n$ independent uses of $\Phi$.}
\label{fig:device}
\end{figure}

\paragraph{Gadget channels.} Let $G$ be a group generated by a finite set $S$, and let $R$ be a finite set of words in $S^{\pm1}$, each of length at
least two and each representing $1$ in $G$. Write $r_i=\sigma_{i,1}\cdots\sigma_{i,\ell_i}$ ($1\le i\le r_0$) for the words of $R$, and put
$s=|S|$ and $\Lambda=\sum_i\ell_i$. The \emph{symbols} are the $2s$ letters $x^{\pm1}$, $x\in S$, and the proper prefixes
$p_{i,l}=\sigma_{i,1}\cdots\sigma_{i,l}$, $2\le l\le\ell_i-1$, one for each pair $(i,l)$; we also put $p_{i,1}=\sigma_{i,1}$ and let $p_{i,\ell_i}$
be the empty word. The \emph{triangles} are the triples $(p_{i,l-1},\sigma_{i,l},p_{i,l})$, $2\le l\le\ell_i$, and $(x,x^{-1},\varnothing)$,
$x\in S$. For unitaries $\lambda(w)$ assigned to the symbols, with $\lambda(\varnothing)=I$, let
\[
U_\lambda=\mathrm{diag}\bigl(I,(\lambda(w))_{w}\bigr)\ \oplus\ \bigoplus_{(x,y,z)}\frac1{\sqrt2}\begin{pmatrix}I&\lambda(y)\\ \lambda(x)&-\lambda(x)\lambda(y)\end{pmatrix},
\]
where $w$ runs over the symbols and $(x,y,z)$ over the triangles. Each triangle block equals $\mathrm{diag}(I,\lambda(x))\,(H\otimes I)\,\mathrm{diag}(I,\lambda(y))$
with $H$ the Hadamard matrix, so $U_\lambda$ is unitary for all unitary $\lambda$. A triangle is a Hadamard test of one step $xy=z$ of one
relation. Its block is unitary whatever is put in it, and its lower right entry $-\lambda(x)\lambda(y)$ is compared, through the Choi support
defined next, with the entry that carries the label of $z$. A device that reproduces the channel must pass every test with the operators it
holds in memory. There are $2s+\Lambda-2r_0$ symbols and $\Lambda-r_0+s$ triangles, so the system dimension is
\begin{equation}\label{eq:gadget-dimension}
d=1+(2s+\Lambda-2r_0)+2(\Lambda-r_0+s)=1+4s+3\Lambda-4r_0 .
\end{equation}
The \emph{gadget channel} $\Phi_G$ on $M_d$ evaluates each symbol $w$ in the left regular representation, $\lambda(w)=\lambda([w])$ with $[w]\in G$ the
element that $w$ represents, and traces out the group von Neumann algebra $L(G)$ with its canonical trace $\tau$:
$\Phi_G(X)=(\mathrm{id}\otimes\tau)[U_\lambda(X\otimes1)U_\lambda^*]$. In particular $\Phi_G$ is factorizable, through $(L(G),\tau)$. Since
the relations hold, the lower right entry of the block of a triangle
$(x,y,z)$ is $-\lambda([z])$.

We write the gadget as $U_\lambda=\sum_{a\in\mathcal T}c_aE_a\otimes\lambda([w_a])$, where $\mathcal T$ is the set of \emph{slots}, the nonzero
entries, $E_a$ is the matrix unit of slot $a$, and $c_a\in\{1,\pm2^{-1/2}\}$. The word $w_a$ is empty, a symbol, or the concatenation $xy$ in
the lower right corner of a triangle. The elements $[w_a]$ are the \emph{labels}. With $A_g=\sum_{[w_a]=g}c_aE_a$, the matrices $A_g$ have disjoint
supports, $\Phi_G(X)=\sum_gA_gXA_g^*$ and $\mathcal J(\Phi_G)=\frac1d\sum_g|A_g\rangle\!\rangle\langle\!\langle A_g|$, because
$\tau(\lambda(g)\lambda(h)^*)=\delta_{gh}$. The Choi rank $r_*$ of $\Phi_G$ is therefore the number of distinct labels. Every label is carried by an
\emph{anchor}, a diagonal slot with coefficient $1$ whose row and column contain no other slot. The anchor of the label $1$ is the first
diagonal entry, and every other label has a symbol entry as anchor, since the lower right entry of a triangle carries the label of a prefix
symbol or $1$. Feeding one half of a maximally entangled state on the $r_*$ anchor positions gives a pure probe whose output is flat of rank
$r_*$, and $\log r_*$ is the largest entropy of a complementary output of $\Phi_G$~\cite[Proposition~9.1]{DouglasMghirbiB}.

\paragraph{The rank rate.} The probe fixes the least exchange that bounded memory allows. Over $n$ rounds the ideal experiment hands the
observer $n\log r_*$ bits of entropy. The observer's final state is a marginal of a pure state whose other half consists of the final memory,
the exported qubits and the purifiers of the initial memory and of the imports. Those are at most $2(Q+J)$ qubits in all, so a device hands
over at most $2(Q+J)$ bits. With memory sublinear in $n$ it must export about $\frac12\log r_*$ qubits per round at least. This rate is the optimal
exchange rate of a gadget channel in $\cK_d$~\cite[Propositions~3.1 and~9.1]{DouglasMghirbiB}, and we call it the \emph{rank rate}: a service is
at the rank rate if $J\le\frac n2\log r_*$. For the lamp channel, Theorem~\ref{thm:streaming} attains it with memory $n^{1/2+o(1)}$.

\paragraph{One round.} Let $U$ be a unitary on $\mathbb C^d\otimes\mathbb C^M$ and $\rho$ a state on $\mathbb C^M$, and put
\[
\Psi(X)=\mathrm{Tr}_M\,U(X\otimes\rho)U^*,\qquad T(\rho)=\mathrm{Tr}_d\,U(\tau_d\otimes\rho)U^* .
\]
Relative to a channel $\Phi$ with Choi support projection $\Pi_\Phi$, the \emph{support leakage} of $\Psi$ is
$\ell=\Tr[(I-\Pi_\Phi)\mathcal J(\Psi)]$. The \emph{entropy production} of the round is $\Delta S=S(T(\rho))-S(\rho)$, which is nonnegative because
$T$ is unital. A service produces such rounds, and the next result, from~\cite[Theorem~5.2]{DouglasMghirbiB}, finds one in which both quantities are
small. Appendix~\ref{app:imported} reproduces its proof.

\begin{theorem}[Sequential extraction]\label{thm:extraction}
Let a causal service of a channel $\Phi$ on $M_d$ for $n$ rounds have memory $Q$, error $\epsilon<1/2$ and $P+C\le B$, and put
\[
t_n=\frac{B+\epsilon+h_2(\epsilon)+\epsilon\log n}{n(1-\epsilon)} .
\]
Then some round $t$, with its unitary $U_t$ and the memory state $\rho_t$ entering it, defines a channel $\Psi_t(X)=\mathrm{Tr}_M\,U_t(X\otimes\rho_t)U_t^*$
with bath dimension $M=2^Q$, $\ddia(\Psi_t,\Phi)\le\epsilon/(1-\epsilon)$, and support leakage $\ell_t$ and entropy production $\Delta S_t$ satisfying
$\ell_t+\Delta S_t\le t_n$.
\end{theorem}

The bath of this round is the memory itself, of dimension $2^Q$, so $S(\rho_t)\le Q$. Every lower bound on memory in Part~II is a lower bound
on the entropy of a memory state that realizes one good round.

\paragraph{A worked example: the gadget of $S_3$.} Take $G=S_3$ with $S=\{a,b\}$ and $R=\{a^2,\,b^2,\,(ab)^3\}$, so that $s=2$, $r_0=3$ and
$\Lambda=2+2+6=10$. The symbols are the four letters $a^{\pm1},b^{\pm1}$ and the four proper prefixes $ab$, $aba$, $abab$, $ababa$ of
$(ab)^3$; the words $a^2$ and $b^2$ have no proper prefix of length two or more. The triangles are the nine triples $(a,a,\varnothing)$,
$(b,b,\varnothing)$, $(a,b,ab)$, $(ab,a,aba)$, $(aba,b,abab)$, $(abab,a,ababa)$, $(ababa,b,\varnothing)$, $(a,a^{-1},\varnothing)$ and
$(b,b^{-1},\varnothing)$. Formula~\eqref{eq:gadget-dimension} gives $d=1+8+2\cdot9=27$. The labels are the elements of $S_3$ that the slots
carry, namely $1$, $a$, $b$, $ab$, $aba$, $abab=ba$ and $ababa=b$; these are all six elements, so $r_*=6$. Since $S_3$ is finite, its group
von Neumann algebra is the regular representation on $\mathbb C^6$ with the normalized trace. The channel $\Phi_{S_3}$ therefore has a finite
tracial factorization with bath dimension $6$ by definition, and one round with a flat bath on six points reproduces it exactly. Serving $n$
rounds costs little more. The regular action of $S_3$ on itself is a model with no bad points, and Theorem~\ref{thm:streaming} with $k=k'=6$ gives a
service at the rank rate with purity at most $2$ and memory $O(\log n)$. The lamp channel of Section~\ref{sec:lamp-channel} adds five letters
and twenty-six words to these three. What those words cost a device is what Theorem~B measures.
\section{The weighted one-round theorem}\label{sec:weighted-statement}

The engine of Part~II reads an approximate representation of the relations directly off the memory state of one round, weighted by that
state, and turns many cheaply commuting copies of $S_3$ into entropy of the memory. We state it here with its amplified form; the proof is in
Appendix~\ref{sec:weighted-proof}. Transport words $u_1,\dots,u_N$ over $S^{\pm1}$ have \emph{length at most $L$} and \emph{pairwise areas at
most $A$} if $|u_i|\le L$ for all $i$ and
\[
\Area_R\bigl([u_ixu_i^{-1},\,u_kyu_k^{-1}]\bigr)\le A\qquad(i\ne k,\ x,y\in\{a,b\}),
\]
where $\Area_R$ is the area in the free group on $S$, as in Section~\ref{sec:criterion}.

\begin{theorem}[Weighted one-round theorem]\label{thm:weighted}
Let $G$, $S$ and $R$ be as in Section~\ref{sec:channel-setting}, and suppose that $S$ contains $a$, $b$ and $j$, that $a^2$, $b^2$, $(ab)^3$ and
$j^{-1}ab$ lie in $R$, and that $j\ne1$ in $G$. Put $m=\max\{6,\ell_1,\dots,\ell_{r_0}\}$. Let $u_1,\dots,u_N$ be transport words of length at
most $L$ and pairwise areas at most $A$. Let $\Psi(X)=\mathrm{Tr}_M\,U(X\otimes\rho)U^*$ be one round with $\ddia(\Psi,\Phi_G)\le1/8$, support
leakage $\ell$ relative to $\Phi_G$ and entropy production $\Delta S$, and put $\nu=\ell+\Delta S$. If
\begin{equation}\label{eq:weighted-condition}
2^{25}\,m^2\sqrt d\,\sqrt N\,(A+L+1)\sqrt\nu\ \le\ 1,
\end{equation}
then $S(\rho)\ge N/16$.
\end{theorem}

The theorem reads the round at relator defect of order $\sqrt\nu$, a single square root of the quantity that Theorem~\ref{thm:extraction}
controls; a service with purity $O(\log n)$ is thus read at defect about $n^{-1/2}$. It does so without tracialization, working with the
actual memory state $\rho$, which need not be flat. The general correspondence of~\cite{DouglasMghirbiB} instead replaces the round by a nearby
finite tracial mixture, which costs a second square root and reads the profile at defect about $n^{-1/4}$. That paper carried the weighted
route out for a single gadget, built from Houghton's group. Theorem~\ref{thm:weighted} holds for every relation set that contains the
relations of $S_3$ and a name for $ab$.

One round reveals a representation because the Choi support of $\Phi_G$ records which slots carry the same label. A round with small leakage
must therefore put nearly the same bath operator at every slot of a given label; the anchors make these operators unitaries, and the
triangle tests make them satisfy the relations. Small entropy production makes them nearly commute with $\sqrt\rho$, and that is what lets
every error be measured in the norm weighted by the memory state rather than in the flat one. Nearly commuting copies of $S_3$ then cost
entropy. Measure the $N$ three-cycles one after another on the memory. Each outcome is nearly a fair coin, because the involution of its copy
flips it, yet a purifying reference nearly knows it, because the copy nearly commutes with $\sqrt\rho$. Each copy thus contributes about one
bit of mutual information with the reference, and that mutual information is at most $S(\rho)$.

Applied to a family of transport words that grows with a parameter, Theorem~\ref{thm:weighted} gives a lower bound on the entropy of one round
in terms of $\nu$ alone.

\begin{corollary}[Amplification]\label{cor:amplification}
Keep the assumptions of Theorem~\ref{thm:weighted} on $G$, $S$, $R$ and the round, and let $K\ge1$ and $\kappa\ge0$.
\begin{enumerate}[label=(\alph*),itemsep=1pt]
\item Suppose that for every integer $k\ge1$ there are $2^k$ transport words of length at most $L_k$ and pairwise areas at most $A_k$, with
$A_k+L_k+1\le Kk^\kappa$. Put $C_R=2^{50}m^4dK^2$. There is $y_0$ depending only on $\kappa$ such that
\[
S(\rho)\ \ge\ \frac{y}{32\,(\log y)^{2\kappa}}\qquad\text{for every }y\ge y_0\text{ with }C_R\,\nu\,y\le1 .
\]
\item Suppose instead that $A_k+L_k+1\le2^{K\sqrt k}$ for every $k\ge1$, and put $C'_R=2^{50}m^4d$. There is $y_1$ depending only on $K$ such that
\[
S(\rho)\ \ge\ \frac{y}{32}\,2^{-2K\sqrt{\log y}}\qquad\text{for every }y\ge y_1\text{ with }C'_R\,\nu\,y\le1 .
\]
\end{enumerate}
\end{corollary}

\begin{proof}
(a) Let $y_0\ge4$ be such that $k=\lfloor\log y-2\kappa\log\log y\rfloor\ge1$ for $y\ge y_0$. Then $k\le\log y$ and $2^k\le y(\log y)^{-2\kappa}$, so
$2^kk^{2\kappa}\le y$. With $N=2^k$, the left side of~\eqref{eq:weighted-condition} is at most
$2^{25}m^2\sqrt d\,K\,(2^kk^{2\kappa})^{1/2}\sqrt\nu\le(C_R\nu y)^{1/2}\le1$. So $S(\rho)\ge2^k/16$, and $2^k>y(\log y)^{-2\kappa}/2$.

(b) Let $y_1$ be such that $k=\lfloor\log y-2K\sqrt{\log y}\rfloor\ge1$ for $y\ge y_1$. Then $K\sqrt k\le K\sqrt{\log y}$, so
$2^k(A_k+L_k+1)^2\le2^{k+2K\sqrt{\log y}}\le y$, and~\eqref{eq:weighted-condition} holds with $N=2^k$ as in part~(a). Hence
$S(\rho)\ge2^k/16>y\,2^{-2K\sqrt{\log y}}/32$.
\end{proof}
\section{The channel of the configuration lamps, and Theorem B}\label{sec:lamp-channel}

This section builds the channel of Theorem~B from the certificate of Theorem~A and states the theorem in full.
Theorem~\ref{thm:weighted} needs $a$, $b$ and a name for $ab$ among the generators. Let $R$ be the certificate of Table~\ref{tab:certificate}, in
the five generators $p,q,t,c,a$, with $b$ standing for the word $ac^4$. Let $S_\Phi=\{p,q,t,c,a,b,j\}$, where the new letters represent
$b=ac^4$ and $j=ab$ in $\Gamma$, and let $R_\Phi$ consist of the $25$ words of $R$ and the four words
\[
b^{-1}ac^4,\qquad b^2,\qquad (ab)^3,\qquad j^{-1}ab
\]
in the letters of $S_\Phi$. They represent $1$ in $\Gamma$, $a^2\in R$, and $j\ne1$ because $ab$ has order three (Lemma~\ref{lem:structure}). The
channel of the configuration lamps is the gadget channel $\Phi=\Phi_\Gamma$ of $R_\Phi$.

\paragraph{Dimension and rank.} The channel by the numbers:
\begin{center}\small
\begin{tabular}{cccccccc}
\toprule
letters & words & total length & symbols & triangles & longest word & dimension & Choi rank\\
\midrule
$s=7$ & $r_0=29$ & $\Lambda=320$ & $276$ & $298$ & $m=24$ & $d=873$ & $r_*=130$\\
\bottomrule
\end{tabular}
\end{center}
The total length is $\Lambda=303+6+2+6+3$, and~\eqref{eq:gadget-dimension} gives $d=1+4\cdot7+3\cdot320-4\cdot29=873$. The Choi rank $r_*$ is
the number of distinct labels. Of the $277$ candidates, the identity and the $276$ symbols, exactly $130$ are distinct elements of
$\Gamma$, so $r_*=130$. Most coincidences are shared prefixes: $pq$ and $pqp^{-1}$, for instance, begin several words. The others are relations
such as $a^2=1$ and $j=ab=c^4$. The script \texttt{anc/choi\_rank\_check.py} computes this partition with two independent engines,
the faithful action of $\Gamma$ on configurations and the embedding in $3V$ (Appendix~\ref{app:machine-checks}). The rank rate is therefore
$\frac12\log130\approx3.51$ qubits per use.

\paragraph{Transport.} For $N\ge1$, the proof of Theorem~\ref{thm:main-group} provides $2^N$ words $T_u$, one for each configuration $u$
supported in the first $N$ points of ray~$1$, with $|T_u|\le21N\max(3,N)\le63N^2$ (Theorem~\ref{thm:pair-area}) and
$\Area_{R_\Phi}([T_uxT_u^{-1},T_wyT_w^{-1}])\le\Lambda_N+4\le C_1N^{\kappa}$ for $u\ne w$ and $x,y\in\{a,b\}$, where
\[
p_H=\log107+4,\qquad \kappa=2+p_H<12.75,
\]
and $C_1$ is absolute; the letter $b$ is converted to $ac^4$ by the word $b^{-1}ac^4$, at most four times, and areas with respect to
$R_\Phi\supseteq R$ are no larger than with respect to $R$. So the hypothesis of Corollary~\ref{cor:amplification}(a) holds with this $\kappa$ and
$K=C_1+64$.

\begin{theorem}[Theorem B]\label{thm:main-channel}
The channel $\Phi$ on $M_{873}$ lies in $\cK_{873}$, and its Choi rank is $r_*=130$. Let $\kappa=2+p_H=6+\log107$. There are
constants $c,c_0,K>0$ and $n_0$ such that the following hold; all services have error at most $1/16$.
\begin{enumerate}[label=(\alph*),itemsep=1pt]
\item \emph{Memory against purity.} For $n\ge n_0$ and $\log n\le B\le c_0n$, every causal service of $\Phi$ for $n$ rounds with purity
$P+C\le B$ has, at any exchange,
\[
Q\ \ge\ c\,\frac{n}{B\,(\log n)^{2\kappa}} .
\]
For every $n\ge1$ and $0<B\le n$ there is an exact causal service, with error $0$, purity $P+C<B$, exchange $J=(n-1)Q$ and memory
$Q\le K(n/B)\log(2n/B)$. In particular, for every $A_0\ge0$ there is $c'>0$ such that services with $P+C\le A_0\log n$ have
$Q\ge c'n/(\log n)^{2\kappa+1}$ for all large $n$, where $2\kappa+1=2p_H+5<26.49$, and there are services with $P=C=0$ and $Q\le Kn\log n$.
\item \emph{The rank rate.} For $n\ge n_0$, every causal service of $\Phi$ for $n$ rounds at the rank rate, $J\le\frac n2\log r_*$, has, whatever
its purity,
\[
Q\ \ge\ c\,\frac{\sqrt n}{(\log n)^{\kappa}} .
\]
For every $n\ge2$ there is a causal service at the rank rate with $C=0$ and $P,Q\le K\sqrt{n\log n}$.
\end{enumerate}
So memory and purity trade at a fixed rate: for $n\ge n_0$ and $\log n\le B\le c_0n$ the least memory of a service with purity at most $B$ lies
between $c\,n/(B(\log n)^{2\kappa})$ and $K(n/B)\log(2n/B)$, and the product of memory and purity is $n^{1+o(1)}$. At the rank rate the least
memory is $n^{1/2+o(1)}$, and it is attained with purity of the same order: the rank rate sits where memory and purity are equal.
\end{theorem}

\begin{figure}[t]
\centering
\begin{tikzpicture}[x=1cm,y=1cm,font=\small,>=Stealth,lbl/.style={font=\footnotesize,inner sep=1pt}]
\colorlet{accent}{blue!50!black}
\fill[black!8] (1.0,5.65) -- (6.0,0.65) -- (6.0,0) -- (1.0,0) -- cycle;
\fill[accent!22] (1.0,6.35) -- (6.0,1.35) -- (6.0,0.65) -- (1.0,5.65) -- cycle;
\draw[accent] (1.0,6.35) -- (6.0,1.35);
\draw[accent] (1.0,5.65) -- (6.0,0.65);
\node[lbl,anchor=west] at (6.1,1.35) {$K(n/B)\log(2n/B)$: exact services};
\node[lbl,anchor=west] at (6.1,0.65) {$c\,n/(B(\log n)^{2\kappa})$};
\node[lbl,rotate=-45] at (2.55,4.45) {$Q\cdot B=n^{1+o(1)}$};
\node[lbl] at (3.3,1.1) {no service (Theorem~B(a))};
\draw[dashed] (2.6,2.2) -- (6.9,6.5);
\node[lbl,anchor=west,align=left] at (6.95,6.5) {$P+C=3Q+6$: rank-rate services\\ lie to the left (purity budget)};
\draw[->] (0,0) -- (8.0,0);
\node[lbl,anchor=north east] at (8.0,-0.12) {purity $B$};
\draw[->] (0,0) -- (0,7.4);
\node[lbl,anchor=south west] at (0.1,7.3) {memory $Q$};
\draw (0.42,-0.12) -- (0.58,0.12) (0.52,-0.12) -- (0.68,0.12);
\foreach \x/\l in {0/{$0$},1.0/{$\log n$},3.5/{$\sqrt n$},6.0/{$c_0n$}}{
  \draw (\x,0.06) -- (\x,-0.06); \node[lbl,anchor=north] at (\x,-0.12) {\l};
}
\foreach \y/\l in {3.5/{$\sqrt n$},6.0/{$n$},6.6/{$n\log n$}}{
  \draw (0.06,\y) -- (-0.06,\y); \node[lbl,anchor=east] at (-0.12,\y) {\l};
}
\fill (0,6.6) circle (2pt);
\node[lbl,anchor=west] at (0.2,6.75) {zero purity: $Q\le Kn\log n$};
\fill (3.5,3.5) circle (2pt);
\draw (3.5,3.5) -- (4.55,3.55);
\node[lbl,anchor=west,align=left] at (4.6,3.55) {rank rate: $Q$ and $P+C$\\ both about $\sqrt n$};
\end{tikzpicture}
\caption{Memory against purity for the lamp channel, on logarithmic axes. For $\log n\le B\le c_0n$ the least memory of a service with purity at
most $B$ lies in the band $Q\cdot B=n^{1+o(1)}$, between $c\,n/(B(\log n)^{2\kappa})$ and $K(n/B)\log(2n/B)$ (Theorem~\ref{thm:main-channel}(a)),
and at zero purity memory $Kn\log n$ suffices. At the rank rate the purity is less than $3Q+6$ (Proposition~\ref{prop:purity-budget}), so the
band forces memory about $\sqrt n$, and streaming services attain this with purity of the same order (Theorem~\ref{thm:main-channel}(b)).}
\label{fig:tradeoff}
\end{figure}
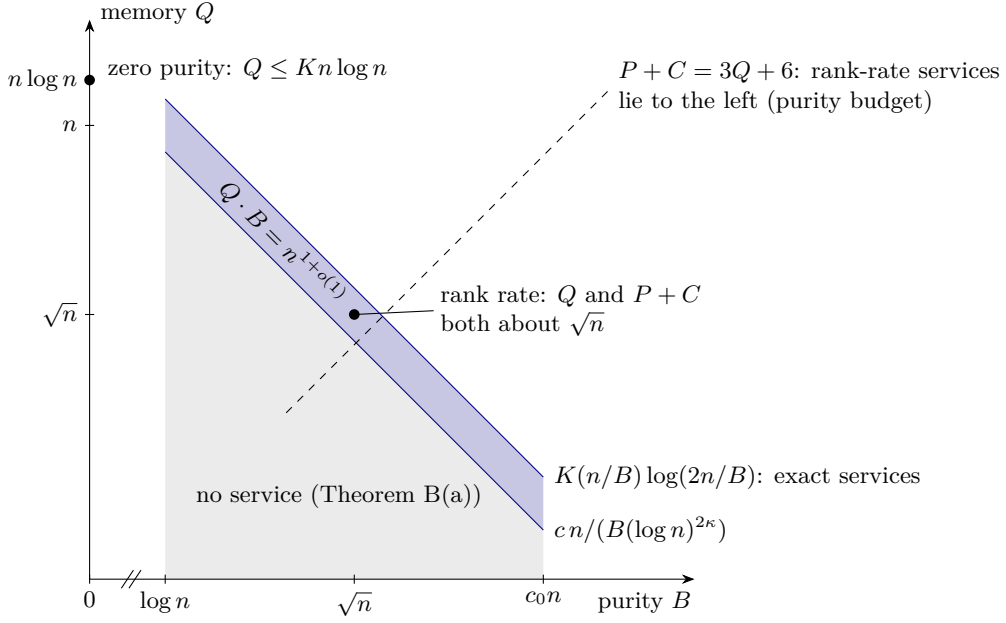

The exponent $2\kappa$ in (a) is the unitary exponent of Theorem~\ref{thm:main-group}(c): a service with purity $B$ is read at relator defect about
$(B/n)^{1/2}$. The lower bound in (b) needs no assumption on purity, because at the rank rate purity is paid for with memory
(Proposition~\ref{prop:purity-budget}): there $P+C<3Q+6$, and the argument of (a) with $B$ of order $Q$ gives $Q^2\gtrsim n$ up to logarithms. The
services in (a) and (b) are the two devices of Section~\ref{sec:devices}. We prove Theorem~\ref{thm:main-channel} in
Section~\ref{sec:main-channel-proof}, after the purity budget and the upper bound, and read it as a statement about finite tracial
approximations in Corollary~\ref{cor:profile}.
\section{Sofic models give devices, and streaming at the rank rate}\label{sec:devices}

\subsection{An upper bound from sofic models}\label{sec:sofic-upper}

Sofic models of the labels give devices, for two reasons. The triangle blocks are exactly unitary for arbitrary permutations. And at most
points of a model the pure bath vector reproduces the target channel exactly, not approximately, so a memory that is maximally mixed on the
good points of a model serves the channel exactly. This gives every upper bound in Theorem~B. Recall from Section~\ref{sec:criterion}
that $\sigma_E(r)$ is the least $M$ for which some $f:E\to\Sym(\Omega)$, $|\Omega|=M$, with $f(1)=\mathrm{id}$, is a
\emph{$(1-r^{-1})$-expansive $r^{-1}$-morphism}. This means that $f(x)f(y)$ and $f(xy)$ differ on at most $M/r$ points whenever $x,y,xy\in E$,
and that $f(x)$ and $f(y)$ agree on at most $M/r$ points whenever $x\ne y$ in $E$.

\begin{proposition}[Sofic models give devices]\label{prop:sofic-upper}
Let $G$, $S$ and $R$ be as in Section~\ref{sec:channel-setting}, with $t_\triangle$ triangles and $r_*$ labels, let $E$ be a chunk of $G$ containing
all labels, and put $K_G=t_\triangle+\binom{r_*}2$. Every $(1-r^{-1})$-expansive $r^{-1}$-morphism $f:E\to\Sym(\Omega)$, $|\Omega|=M$, gives a unitary $U_f$ on
$\mathbb C^d\otimes\mathbb C^\Omega$ whose tracial channel $\Phi_f(X)=(\mathrm{id}\otimes\tau_M)[U_f(X\otimes I)U_f^*]$ satisfies
\[
\Phi_f=(1-\vartheta)\Phi_G+\vartheta\,\Xi
\]
for some channel $\Xi$ and some $0\le\vartheta\le K_G/r$. Consequently, for all $n\ge1$ and $0<\epsilon\le1$ with $\sigma_E(2K_Gn/\epsilon)<\infty$,
$\Phi_G$ has
\begin{enumerate}[label=(\alph*),itemsep=1pt]
\item a causal service with error at most $\epsilon$, $P=C=0$, $J=(n-1)Q$ and $Q=\bigl\lceil\log\sigma_E(2K_Gn/\epsilon)+\log(2n/\epsilon)\bigr\rceil$;
\item a causal service with error at most $\epsilon$, $J=P=C=0$ and $Q=\bigl\lceil n\log\sigma_E(2K_Gn/\epsilon)+\log(2/\epsilon)\bigr\rceil$;
\end{enumerate}
and, for all $n\ge1$ and $0<B\le n$ with $\sigma_E(4K_Gn/B)<\infty$,
\begin{enumerate}[label=(\alph*),itemsep=1pt,resume]
\item an exact causal service, with error $0$, purity $P+C<B$, $J=(n-1)Q$ and $Q=\bigl\lceil\log\sigma_E(4K_Gn/B)+\log(4n/B)\bigr\rceil$.
\end{enumerate}
\end{proposition}

\begin{proof}
\emph{The permutation gadget is unitary.} For a symbol $w$ let $P_w$ be the permutation matrix of $f([w])$ on $\mathbb C^\Omega$, so that symbols
with the same label get the same matrix, and let $U_f$ be the gadget $U_\lambda$ with $\lambda(w)=P_w$. Each triangle block is
$\frac1{\sqrt2}\bigl(\begin{smallmatrix}I&P_y\\P_x&-P_xP_y\end{smallmatrix}\bigr)=\mathrm{diag}(I,P_x)(H\otimes I)\mathrm{diag}(I,P_y)$, so $U_f$ is unitary
although no relation holds among the $P_w$. Writing $\pi_a$ for the permutation of $\Omega$ in slot $a$, namely $f([w_a])$ for every slot except the
lower right corner of a triangle $(x,y,z)$, where it is $f([x])f([y])$, we have $U_f=\sum_ac_aE_a\otimes P_{\pi_a}$.

\emph{Good points reproduce $\Phi_G$.} Since $\tau_M$ is the average of the vector states of the points of $\Omega$, $\Phi_f=\frac1M\sum_{\omega\in\Omega}\Phi_\omega$ with
$\Phi_\omega(X)=\mathrm{Tr}_\Omega\,U_f(X\otimes|\omega\rangle\langle\omega|)U_f^*$. The isometry $|\psi\rangle\mapsto U_f(|\psi\rangle\otimes|\omega\rangle)$ is
$\sum_ac_aE_a\otimes|\pi_a(\omega)\rangle$, so
\begin{gather*}
\Phi_\omega(X)=\sum_{a,b\in\mathcal T}c_a\overline{c_b}\;\mathbf 1[\pi_a(\omega)=\pi_b(\omega)]\;E_aXE_b^*,\\
\Phi_G(X)=\sum_{a,b\in\mathcal T}c_a\overline{c_b}\;\mathbf 1\bigl[[w_a]=[w_b]\bigr]\;E_aXE_b^*,
\end{gather*}
the second because $\tau(\lambda(g)\lambda(h)^*)=\delta_{gh}$. Call $\omega$ \emph{good} if $f([x])f([y])(\omega)=f([z])(\omega)$ for every triangle
$(x,y,z)$ and $f(g)(\omega)\ne f(h)(\omega)$ for all distinct labels $g,h$. At a good point $\pi_a(\omega)=f([w_a])(\omega)$ for every slot, since
$[w_a]=[z]$ at the lower right corner of a triangle $(x,y,z)$; so $\pi_a(\omega)=\pi_b(\omega)$ exactly when $[w_a]=[w_b]$, and $\Phi_\omega=\Phi_G$.

\emph{Bad points are few.} The labels of a triangle satisfy $[x][y]=[z]$ with all three in $E$, so $f([x])f([y])$ and $f([z])$ differ on at most
$M/r$ points; and distinct labels $g,h\in E$ have $f(g)(\omega)=f(h)(\omega)$ for at most $M/r$ points. So the fraction $\vartheta$ of bad points is
at most $K_G/r$, and $\Phi_f=(1-\vartheta)\Phi_G+\vartheta\,\Xi$, where $\Xi$ is the average of $\Phi_\omega$ over the bad points. If there are no
bad points, $\Xi$ may be any channel.

\emph{Refreshing the memory gives part~(a).} Let $r=2K_Gn/\epsilon$, let $f$ have degree $M=\sigma_E(r)$, and let $Q$ be as in part~(a), so that
$M\le2^Q\epsilon/(2n)$. Split the $Q$-qubit memory as $(\mathbb C^\Omega\otimes\mathbb C^{k_0})\oplus\mathbb C^{g_0}$ with $k_0=\lfloor2^Q/M\rfloor$ and
$g_0=2^Q-k_0M<M$. In every round apply $U_f\otimes I_{k_0}$ on the first summand and the identity on the second. After every round but the
last, export all $Q$ memory qubits and import $Q$ fresh qubits in the maximally mixed state; the initial state is maximally mixed as well, so
$P=C=0$. Every round then meets a fresh flat memory, and the observer faces $n$ independent uses of the channel
$(1-g_0/2^Q)\Phi_f+(g_0/2^Q)\mathrm{id}$, which is within $\vartheta+g_0/2^Q\le\epsilon/(2n)+\epsilon/(2n)$ of $\Phi_G$ in diamond distance, by
convexity. A hybrid argument over the rounds bounds the error by $\epsilon$.

\emph{A memory flat on the good points gives part~(c).} Let $0<B\le n$, put $\eta=B/(4n)\le\frac14$ and $r=K_G/\eta$, let $f$ have degree
$M=\sigma_E(r)$, and let $Q$ be as in part~(c), so that $M\le\eta2^Q$. Let $\Omega_0\subseteq\Omega$ be the set of good points, so that
$|\Omega_0|\ge(1-\eta)M$. Split the memory as $(\mathbb C^\Omega\otimes\mathbb C^{k_0})\oplus\mathbb C^{g_0}$ as in part~(a), with
$k_0=\lfloor2^Q/M\rfloor$ and $g_0<M\le\eta2^Q$, and let $\beta$ be the maximally mixed state on
$\mathrm{span}\{|\omega\rangle:\omega\in\Omega_0\}\otimes\mathbb C^{k_0}$. Since $k_0M=2^Q-g_0$, its deficit is
\[
Q-\log(|\Omega_0|k_0)=\log\frac{2^Q}{2^Q-g_0}+\log\frac M{|\Omega_0|}\le2\log\frac1{1-\eta}\le\frac{8\eta}{3\ln2}<\frac Bn .
\]
Start the memory in $\beta$ and apply $U_f\otimes I_{k_0}$ on the first summand in every round. After every round but the last, export all $Q$
memory qubits and import $Q$ fresh qubits in the state $\beta$. The $n$ registers used have total deficit $P+C<B$. In every round the memory is a
fresh copy of $\beta$, a uniform mixture of vector states $|\omega\rangle\otimes|j\rangle$ with $\omega$ good, and $\Phi_\omega=\Phi_G$ at good points.
So every round applies exactly $\Phi_G$, independently of the other rounds, and the error is $0$.

\emph{One factor per round gives part~(b).} Let $f$ be the model of part~(a) and $Q$ as in part~(b), so that $M^n\le2^Q\epsilon/2$. Split the
memory as $(\mathbb C^{\Omega^n}\otimes\mathbb C^{k_0})\oplus\mathbb C^{g_0}$ with $k_0=\lfloor2^Q/M^n\rfloor$ and $g_0<M^n$, start it maximally
mixed, and in round $s$ apply $U_f$ to the input and the $s$-th factor $\mathbb C^\Omega$ of the first summand, and the identity on the second.
With probability $1-g_0/2^Q\ge1-\epsilon/2$ the memory starts in the first summand, uniformly distributed, and then the $n$ factors are
independent flat baths, each used once. The observer then faces $n$ independent uses of $\Phi_f$, which is within $n\vartheta\le\epsilon/2$ of
the ideal experiment. The remaining branch has weight at most $\epsilon/2$, so the error is at most $\epsilon$.
\end{proof}

In particular, if $\sigma_E(r)$ is finite for every $r$, as it is when $G$ is sofic, then $\Phi_G\in\cK_d$, since
$\ddia(\Phi_f,\Phi_G)\le K_G/r$ for models at every accuracy. This is a quantitative form, for sofic groups, of the fact that gadget channels of
hyperlinear groups lie in $\cK_d$~\cite{DouglasMghirbiB}.

\subsection{Streaming at the rank rate}\label{sec:streaming}

The services of Proposition~\ref{prop:sofic-upper} refresh the whole memory in every round, so their exchange is far above the rank rate. At
the rank rate the memory has to be kept, and a sofic model cannot simply be reused, since its bad points would then accumulate error.
Theorem~\ref{thm:streaming} keeps one sofic model for all rounds and pays for its bad points with purity. It uses the good points of the model
as an exact \emph{rectangular} dilation, whose input bath is smaller than its output bath, and lets the part of the memory in use grow slightly
in every round. The mixing and export steps are those of the rolling reservoir of~\cite[Appendix~A]{DouglasMghirbiB}, and the proof is in
Appendix~\ref{app:streaming}.

\begin{theorem}[Streaming at the rank rate]\label{thm:streaming}
Let $\Phi$ be a channel on $M_d$ with linearly independent Kraus operators $A_1,\dots,A_r$, $r\ge2$, and let $1\le k\le k'$. Suppose that
$W=\sum_iA_i\otimes B_i$, with $B_i:\mathbb C^k\to\mathbb C^{k'}$, is an isometry from $\mathbb C^d\otimes\mathbb C^k$ to $\mathbb C^d\otimes\mathbb C^{k'}$ and that
\begin{equation}\label{eq:rectangular}
\Tr(B_i^*B_j)=k\,\delta_{ij},\qquad \sum_iB_i^*B_i=r\,I_k .
\end{equation}
Then for all $n\ge2$ and $0<\epsilon<1$ there is a causal service of $\Phi$ for $n$ rounds with error at most $\epsilon$, $C=0$,
$J=\lfloor\frac{n-1}2\log r\rfloor$ and
\[
P\le n\log\frac{k'}k+2,\qquad Q\le\log k+n\log\frac{k'}k+O_r\Bigl(\log\frac n\epsilon+\log(2+\log k)\Bigr).
\]
\end{theorem}

The first identity in~\eqref{eq:rectangular} says that $W$ with a flat input bath realizes $\Phi$ exactly:
$\Phi(X)=\mathrm{Tr}_{k'}[W(X\otimes I_k/k)W^*]=\sum_iA_iXA_i^*$. The second says that every Kraus history costs the same. For a channel with a
flat probe of rank $r$ it follows from $W^*W=I$ alone: apply $\Tr(\rho_*\,\cdot\,)$ to the first factor of $\sum_{i,j}A_i^*A_j\otimes B_i^*B_j=I$,
where $\rho_*$ is the input marginal of the probe, for which $\Tr(\rho_*A_i^*A_j)=\delta_{ij}/r$.

\begin{corollary}[The lamp channel at the rank rate]\label{cor:streaming-lamps}
There is $K$ such that for every $n\ge2$ the channel $\Phi$ of Theorem~\ref{thm:main-channel} has a causal service for $n$ rounds with error at most
$1/16$, $C=0$, $J=\lfloor\frac{n-1}2\log r_*\rfloor$ and $P,Q\le K\sqrt{n\log n}$.
\end{corollary}

\begin{proof}
Let $E$ be the chunk of labels of Proposition~\ref{prop:sofic-upper}, let $s\ge2K_G$, and let
$f:E\to\Sym(\Omega)$ be a $(1-s^{-1})$-expansive $s^{-1}$-morphism of degree $M=\sigma_E(s)$. Let $\Omega_0$ be its set of good points, so that
$k=|\Omega_0|\ge(1-K_G/s)M$, and for each label $g$ let $B_g:\mathbb C^{\Omega_0}\to\mathbb C^\Omega$ send $|\omega\rangle$ to $|f(g)\omega\rangle$. At a good point
$\omega$ every slot of $U_f$ acts through the permutation of its label, so $U_f(\xi\otimes|\omega\rangle)=\sum_gA_g\xi\otimes|f(g)\omega\rangle$, and
$W=\sum_gA_g\otimes B_g$ is the restriction of the unitary $U_f$ to $\mathbb C^d\otimes\mathbb C^{\Omega_0}$, an isometry. Each $B_g^*B_g$ is $I_k$, and
$\Tr(B_g^*B_h)=0$ for $g\ne h$ because distinct labels do not collide at good points; the $A_g$ are linearly independent, having disjoint
supports. So Theorem~\ref{thm:streaming} applies with $r=r_*$ and $k'=M$. Here $\log(M/k)\le-\log(1-K_G/s)\le2K_G/(s\ln2)$, and
$\log k\le\log\sigma_E(s)\le C_Es\log s$ by Theorem~\ref{thm:main-group}(b). With $\epsilon=1/16$ the theorem gives
$Q\le C_Es\log s+2K_Gn/(s\ln2)+O(\log n)$ and $P\le2K_Gn/(s\ln2)+2$; take $s=\max(2K_G,\lceil\sqrt{n/\log n}\,\rceil)$.
\end{proof}
\section{A purity budget at the rank rate}\label{sec:purity-budget}

At the rank rate a device cannot buy memory with purity: the purity it can spend is linear in its memory. This is the last tool for the lower
bound at the rank rate in Theorem~B. We sharpen the budget $14Q+3$ of~\cite{DouglasMghirbiB} by combining its rank count with a one-sided
second-moment bound.

\begin{proposition}[Purity budget]\label{prop:purity-budget}
Let $\Phi$ be a channel on $M_d$, and suppose that some pure state $\psi_*$ of the input and a reference has output
$\omega=(\Phi\otimes\mathrm{id})(\psi_*)$ flat of rank $r$. Every causal service of $\Phi$ for $n$ rounds with error $\epsilon\le1/16$ and exchange
$J\le\frac n2\log r$ satisfies
\[
P+C\ <\ 3Q+6 .
\]
\end{proposition}

Every gadget channel has such a probe with $r=r_*$ (Section~\ref{sec:channel-setting}). The proof uses the variance of the surprisal of a
state with small deficit, which is~\cite[Lemma~5.9]{DouglasMghirbiB}; we include the short argument.

\begin{lemma}[Surprisal variance]\label{lem:surprisal}
Let $\beta$ be a state on $q\ge1$ qubits with spectrum $(\beta_i)$ and deficit $D=q-S(\beta)$. The surprisal $-\log\beta_i$, drawn with probability
$\beta_i$, has variance at most $(q+2/\ln2)\,D$.
\end{lemma}

\begin{proof}
Put $\varphi(x)=x\ln x-x+1\ge0$. We claim $x(\ln x)^2\le(2+\ln_+x)\varphi(x)$ for $x\ge0$. For $x=e^t$ with $t\ge0$ the difference of the two sides
is $f(t)=e^t(t-2)+t+2$, with $f(0)=f'(0)=0$ and $f''(t)=te^t\ge0$. For $t\le0$ it is $g(t)=2(te^t-e^t+1)-t^2e^t$, with $g(0)=0$ and
$g'(t)=-t^2e^t\le0$; and $x=0$ is clear. Now let $Q_0=2^q$ and $x_i=Q_0\beta_i\le Q_0$, the index running over all $Q_0$ eigenvalues. Then
$\frac1{Q_0}\sum_ix_i=1$ and $D\ln2=\sum_i\beta_i\ln(Q_0\beta_i)=\frac1{Q_0}\sum_i\varphi(x_i)$. The variance is at most
\[
\sum_i\beta_i\bigl(\log(Q_0\beta_i)\bigr)^2=\frac1{Q_0(\ln2)^2}\sum_ix_i(\ln x_i)^2\le\frac{2+q\ln2}{(\ln2)^2}\,D\ln2=\Bigl(q+\frac2{\ln2}\Bigr)D .
\]
\end{proof}

\begin{proof}[Proof of Proposition~\ref{prop:purity-budget}]
Put $b=P+C$ and suppose, for a contradiction, that $b\ge3Q+6$.

\emph{Preloading.} The imports are fixed states, independent of everything else, and the exported qubits are never used again. So the observer's
final state does not change if every import is prepared at the start in a register of its own, and each replacement swaps it with the exported
qubits, which then stay in the device. The device becomes a register of $L=Q+J$ qubits in the product state
$\sigma_0=\beta_0\otimes\bigotimes_f\omega_f$ of the initial memory state and the imported states, and it applies a fixed unitary in each round.

\emph{The surprisal.} Let $(\sigma_i)$ be the spectrum of $\sigma_0$ and let $Z=L+\log\sigma_i$, with $i$ drawn with probability $\sigma_i$. Then
$\mathbb EZ=L-S(\sigma_0)=b$. Also $Z$ is the sum of the independent variables $Z_f=q_f+\log\beta^{(f)}_i$ attached to the factors of $\sigma_0$,
where $q_f$ is the number of qubits of the factor and $\beta^{(f)}$ its spectrum. Every factor has at most $Q$ qubits, since an import replaces
exported memory qubits, and the deficits $b_f=\mathbb EZ_f$ of the factors add up to $b$. By Lemma~\ref{lem:surprisal},
\[
\mathrm{Var}\,Z=\sum_f\mathrm{Var}\,Z_f\le\Bigl(Q+\frac2{\ln2}\Bigr)b .
\]

\emph{A tail bound from the rank.} Let the observer feed the input halves of $n$ copies of $\psi_*$ and keep the references. Its final state is
$\tau_O=\sum_i\sigma_i\tau_i$, where $\tau_i$ is its final state when the register starts in the eigenvector of $\sigma_i$. Each $\tau_i$ is a
marginal of a pure state of the observer's system and the $L$-qubit register, so $\rank\tau_i\le2^L$. Fix a real number $k$ and let $I_k$ be the
set of indices with $\sigma_i\ge2^{k-L}$, so that $|I_k|\le2^{L-k}$. Let $E$ be the support projection of $\sum_{i\in I_k}\tau_i$. Then
$\rank E\le2^{2L-k}$ and $\Tr(E\tau_O)\ge\sum_{i\in I_k}\sigma_i=\Pr(Z\ge k)$. The ideal final state $\omega^{\otimes n}$ is flat of rank $r^n$, so
$\Tr(E\omega^{\otimes n})\le2^{2L-k}r^{-n}\le2^{2Q-k}$, using $2J\le n\log r$. Since the error is at most $\epsilon$,
\begin{equation}\label{eq:rank-tail}
\Pr(Z\ge k)\ \le\ \epsilon+2^{2Q-k}\qquad(k\in\mathbb R).
\end{equation}

\emph{The contradiction.} Put $u=b-2Q-3$, so that $u\ge b/3+1>0$. By the one-sided Chebyshev inequality,
$\Pr(Z\ge b-u)\ge u^2/(\mathrm{Var}\,Z+u^2)$, while~\eqref{eq:rank-tail} with $k=b-u=2Q+3$ gives $\Pr(Z\ge b-u)\le\frac1{16}+\frac18=\frac3{16}$. Hence
$13u^2\le3\,\mathrm{Var}\,Z$, and with $u>b/3$,
\[
\frac{13}9\,b^2<3\Bigl(Q+\frac2{\ln2}\Bigr)b,\qquad\text{that is}\qquad b<\frac{27}{13}\Bigl(Q+\frac2{\ln2}\Bigr)<2.08\,Q+5.993 ,
\]
which contradicts $b\ge3Q+6$.
\end{proof}
\section{Proof of Theorem~B}\label{sec:main-channel-proof}

Theorem~B, in the form of Theorem~\ref{thm:main-channel}, is assembled from sequential extraction, the weighted one-round theorem, the purity
budget and the sofic-model devices. Two remarks after the proof say what it shows.

\emph{Sofic models give the closure and the services.} By Theorem~\ref{thm:main-group}(b) the chunk $E$ of $\Gamma$ that consists of $1$ and
the labels of the gadget of $R_\Phi$ has $\log\sigma_E(r)\le C_Er\log r$. Proposition~\ref{prop:sofic-upper} therefore gives finite tracial
channels within $K_G/r$ of $\Phi$ for every $r$, so $\Phi\in\cK_{873}$. This also follows from~\cite{DouglasMghirbiB}, since $\Gamma$ is amenable
(Lemma~\ref{lem:structure}) and hence hyperlinear. The exact services of part~(a) are those of Proposition~\ref{prop:sofic-upper}(c), and the
zero-purity services are those of Proposition~\ref{prop:sofic-upper}(a) with $\epsilon=1/16$. Their memory bounds follow from
$\log\sigma_E(R)\le C_ER\log R$ for $R\ge2$ and from $\log(4K_Gn/B)\le(1+\log2K_G)\log(2n/B)$, which holds because $\log(2n/B)\ge1$.
Corollary~\ref{cor:streaming-lamps} gives the services of part~(b). The bounds on $r_*$ were proved in Section~\ref{sec:lamp-channel}.

\emph{Part (a).} Let a service have error $\epsilon\le1/16$ and $P+C\le B$, where $\log n\le B$. Theorem~\ref{thm:extraction} gives a round with
bath dimension $2^Q$, $\ddia(\Psi_t,\Phi)\le\epsilon/(1-\epsilon)\le1/15$ and $\nu_t=\ell_t+\Delta S_t\le t_n$. Since $\epsilon+h_2(\epsilon)<1$,
$\epsilon\log n\le\log n$ and $1\le\log n\le B$,
\[
t_n\le\frac{16}{15}\cdot\frac{B+1+\log n}{n}\le\frac{16B}{5n}\qquad(n\ge2).
\]
Let $C_R\ge1$ be the constant of Corollary~\ref{cor:amplification}(a) for $R_\Phi$ and the transport of Section~\ref{sec:lamp-channel}, and put
$y_n=5n/(16C_RB)$, so that $C_R\nu_ty_n\le1$ and $y_n\le n$. Put $c_0=5/(16C_Ry_0)$. If $B\le c_0n$, then $y_n\ge y_0$, and the corollary gives
\[
Q\ \ge\ S(\rho_t)\ \ge\ \frac{y_n}{32(\log y_n)^{2\kappa}}\ \ge\ \frac{5}{512\,C_R}\cdot\frac{n}{B(\log n)^{2\kappa}} .
\]
For the second statement take $B=\max(1,A_0)\log n$, which lies between $\log n$ and $c_0n$ for large $n$.

\emph{Part (b).} Proposition~\ref{prop:purity-budget}, with the flat probe of rank $r_*$, gives $P+C\le B:=3Q+6$. Theorem~\ref{thm:extraction} now
gives a round with $\nu_t\le t_n\le\frac{16}{15}(3Q+7+\log n)/n$. Put $y_n=15n/(16C_R(3Q+7+\log n))\le n$. If $y_n<y_0$, then
$3Q+7+\log n>15n/(16C_Ry_0)$, so $Q\ge c'n$ for large $n$. Otherwise Corollary~\ref{cor:amplification}(a) gives $Q\ge y_n/(32(\log n)^{2\kappa})$,
that is,
\[
Q\,(3Q+7+\log n)\ \ge\ \frac{15}{512\,C_R}\cdot\frac{n}{(\log n)^{2\kappa}} .
\]
If $Q\le\log n+7$ the left side is at most $4(\log n+7)^2$, which is smaller than the right side for large $n$. So $Q>\log n+7$, the left side is
at most $4Q^2$, and $Q\ge c\sqrt n/(\log n)^\kappa$. \qed

\paragraph{Finite tracial approximations.} The same argument has a static reading. A finite tracial factorization of bath dimension $M$ is a
round with the flat bath state $\tau_M$: its entropy production vanishes, and its support leakage is at most its Choi distance to $\Phi$. So
Theorem~\ref{thm:weighted} also reads the static approximation profile of $\Phi$, and the hash models bound it from above. For channels
$\Psi,\Psi'$ on $M_d$ let $d_J(\Psi,\Psi')=\frac12\|\mathcal J(\Psi)-\mathcal J(\Psi')\|_1$ be the distance of their normalized Choi states, so
that $d_J\le\ddia\le d\,d_J$~\cite[Section~2]{DouglasMghirbiB}. For $u\ge0$ let $D_\Phi(u)$ be the least $M$ for which some channel $\Psi$ with a
finite tracial factorization of bath dimension $M$ has $d_J(\Psi,\Phi)\le u$; this is the Choi profile
of~\cite[Definition~2.1]{DouglasMghirbiB}. Let $D^\diamond_\Phi(u)$ be the same with $\ddia$ in place of $d_J$, and $D^{\mathrm{mix}}_\Phi(u)$
the same with $\Psi$ a finite convex combination of channels with finite tracial factorizations of bath dimensions at most $M$. Each is
$+\infty$ if no such $\Psi$ exists, each is nonincreasing in $u$, and $D^{\mathrm{mix}}_\Phi\le D_\Phi\le D^\diamond_\Phi\le D_\Phi(\,\cdot\,/d)$.

\begin{corollary}[Finite tracial approximations of the lamp channel]\label{cor:profile}
Let $\Phi$ be the channel of Theorem~\ref{thm:main-channel} on $M_{873}$, let $\kappa=6+\log107$, let $E$ be the chunk of $\Gamma$
consisting of $1$ and the labels of $\Phi$, and let $K_G=298+\binom{r_*}2=8683$ be its constant in Proposition~\ref{prop:sofic-upper}.
There are constants $c,C>0$ and $u_0>0$ such that for $0<u\le u_0$
\[
c\,\frac{u^{-1}}{(\log(1/u))^{2\kappa}}\le\log D^{\mathrm{mix}}_\Phi(u)\le\log D_\Phi(u)\le\log D^\diamond_\Phi(u)\le\log\sigma_E(K_G/u)\le C\,u^{-1}\log(1/u),
\]
where $2\kappa<25.49$. So the least bath dimension of a finite tracial approximation of $\Phi$ to accuracy $u$, in Choi or in diamond
distance, is $\exp(u^{-1+o(1)})$, and $D^{\mathrm{mix}}_\Phi(0)=\infty$: $\Phi$ is a limit of channels with finite tracial factorizations
but is not a finite convex combination of such channels, although it has a Stinespring dilation with a pure environment of dimension
$r_*=130$.
\end{corollary}

\begin{proof}
\emph{Upper bound.} Let $0<u\le\frac12$ and $r=K_G/u\ge2$. Since $\sigma_E(r)$ is finite by Theorem~\ref{thm:main-group}(b), there is a
$(1-r^{-1})$-expansive $r^{-1}$-morphism $f:E\to\Sym(\Omega)$ with $|\Omega|=\sigma_E(r)$. Proposition~\ref{prop:sofic-upper} gives a channel
$\Phi_f=(1-\vartheta)\Phi+\vartheta\,\Xi$ with a finite tracial factorization of bath dimension $\sigma_E(r)$ and $0\le\vartheta\le K_G/r=u$, so
$\ddia(\Phi_f,\Phi)=\vartheta\,\ddia(\Xi,\Phi)\le u$ and $D^\diamond_\Phi(u)\le\sigma_E(K_G/u)$. By Theorem~\ref{thm:main-group}(b),
$\log\sigma_E(r)\le C_Er\log r$, and $\log(K_G/u)\le(1+\log K_G)\log(1/u)$ since $\log(1/u)\ge1$; so $C=C_EK_G(1+\log K_G)$ serves. The
inequalities between the three profiles hold because $d_J\le\ddia$ and a single factorization is a convex combination with one term.

\emph{A mixture is a flat round.} Let $\Psi=\sum_{i=1}^kp_i\Psi_i$ with $d_J(\Psi,\Phi)\le u$, where
$\Psi_i(X)=(\mathrm{id}\otimes\tau_{m_i})[U_i(X\otimes I)U_i^*]$ and $m_i\le M$. By Carath\'eodory's theorem we may take $k\le d^4$, since
the channels on $M_d$ lie in a real affine space of dimension less than $d^4$. On $\mathbb C^d\otimes\bigoplus_i\mathbb C^{m_i}$ put
$U=\bigoplus_iU_i$ and $\rho=\bigoplus_ip_i\tau_{m_i}$. Then the partial trace of $U(X\otimes\rho)U^*$ over $\bigoplus_i\mathbb C^{m_i}$ is
$\Psi(X)$, so $\Psi$ is a round in the sense of Section~\ref{sec:channel-setting}, with $S(\rho)=H(p)+\sum_ip_i\log m_i\le4\log d+\log M$. Its
entropy production vanishes: $T(\rho)=\bigoplus_ip_i\,\mathrm{Tr}_d\,U_i(\tau_d\otimes\tau_{m_i})U_i^*=\rho$, because
$\tau_d\otimes\tau_{m_i}$ is the normalized identity and $U_i$ is unitary. Its support leakage is
$\ell=\Tr[(I-\Pi_\Phi)(\mathcal J(\Psi)-\mathcal J(\Phi))]\le\frac12\|\mathcal J(\Psi)-\mathcal J(\Phi)\|_1\le u$, because $\mathcal J(\Phi)$
is supported in $\Pi_\Phi$ and $\Tr[PX]\le\frac12\|X\|_1$ for a projection $P$ and a traceless Hermitian $X$. So $\nu=\ell+\Delta S\le u$, and
$\ddia(\Psi,\Phi)\le d\,u\le\frac18$ when $u\le1/(8d)$: the round satisfies the hypotheses of Theorem~\ref{thm:weighted}.

\emph{Lower bound.} Section~\ref{sec:lamp-channel} verifies the hypothesis of Corollary~\ref{cor:amplification}(a) for $R_\Phi$ with the
present $\kappa$ and $K=C_1+64$. Let $C_R=2^{50}m^4dK^2\ge1$ and $y_0$ be the constants of that corollary, and put $y=1/(C_Ru)$. Then
$C_R\nu y\le1$, and $y\ge y_0$ when $u\le1/(C_Ry_0)$. For $u\le\min\{1/(8d),1/(C_Ry_0)\}$ the corollary gives, since $1<y\le1/u$,
\[
4\log d+\log M\ \ge\ S(\rho)\ \ge\ \frac{y}{32(\log y)^{2\kappa}}\ \ge\ \frac1{32C_R}\cdot\frac{u^{-1}}{(\log(1/u))^{2\kappa}} .
\]
Let $u_0$ be so small that the right side is at least $8\log d$ for $u\le u_0$. Then $\log M\ge c\,u^{-1}(\log(1/u))^{-2\kappa}$ with
$c=1/(64C_R)$, and with $c=1/(32C_R)$ for a single factorization, where $H(p)=0$. Finally $D^{\mathrm{mix}}_\Phi(0)\ge D^{\mathrm{mix}}_\Phi(u)$
for every $u>0$, and the lower bound is unbounded as $u\to0$. The pure environment is the Stinespring dilation of Remark~\ref{rem:plus-one}.
\end{proof}

By a theorem of Haagerup and Musat~\cite[Theorem~3.6]{HaagerupMusat2015}, a channel on $M_d$ lies in $\cK_d$ exactly when it factors through an
ultrapower of the hyperfinite $\mathrm{II}_1$ factor, and Musat and R{\o}rdam showed that such a channel may have no finite tracial
factorization at all~\cite{MusatRordam2020}. The Houghton channel of~\cite[Theorem~D.1]{DouglasMghirbiB} is an explicit channel of this kind,
with profile exponent between $0.069$ and $1/3$. For $\Phi$ the exponent is $1-o(1)$, nearly the most that the one-round method can certify
(Remark~\ref{rem:sqrt-limit}), while channels in the relative interior of $\cK_d$ have profiles at most linear in
$1/u$~\cite[Section~10]{DouglasMghirbiB}.

\begin{remark}[Memory between two readings of the profile]\label{rem:sandwich}
Take a gadget channel of a sofic group that satisfies the hypotheses of Theorem~\ref{thm:weighted}. The least memory of its services with error
$1/16$ and purity $O(\log n)$ lies between two readings of the group at scale $n$. The lower reading is the entropy forced by as many cheaply
commuting copies of $S_3$ as Corollary~\ref{cor:amplification} fits at relator defect about $\sqrt{\log n/n}$. The upper reading is
$\log\sigma_E(Cn)+O(\log n)$, by Proposition~\ref{prop:sofic-upper}; it is attained already at zero purity, with exchange $(n-1)Q$. For the
configuration lamps both readings are $n^{1+o(1)}$. The upper reading is the sofic profile, not the F{\o}lner function. The latter is doubly
exponential for $\Gamma$ (Proposition~\ref{prop:folner-lower}), so models built from F{\o}lner sets~\cite[Proposition~3.14]{Cornulier2013} give
memory exponential in $n$, while the hash models behind Theorem~\ref{thm:main-group}(b) give $O(n\log n)$. The grid host gives a second pair
of readings, conditionally on the sketches of Appendix~\ref{sec:grid} (Remark~\ref{rem:grid-channel}).
\end{remark}

\begin{remark}[The rank rate is the diagonal of the trade-off]\label{rem:sqrt-limit}
At the rank rate a device with error at most $1/16$ has purity less than $3Q+6$ (Proposition~\ref{prop:purity-budget}).
Theorem~\ref{thm:main-channel}(a) with $B$ of order $Q$ therefore forces $Q\gtrsim\sqrt n$, and the streaming services attain this with purity
of the same order. In general, for every channel with a flat probe, padding a rank-rate service with error at most $1/16$ by $2Q+6$ maximally mixed
qubits turns it into a service whose purity is at most its memory. The least memory at purity at most memory is therefore at most three times
the least memory at the rank rate, plus six. For $\Phi$ both are $n^{1/2+o(1)}$. The one-round method cannot give more than $\sqrt n$ at the
rank rate, since condition~\eqref{eq:weighted-condition} certifies at most about $1/\nu$ copies; Theorem~\ref{thm:streaming} shows that for
$\Phi$ nothing more is true.
\end{remark}

\part{Physics and open problems}
\section{Local physics}\label{sec:physics}

Theorem~B concerns an abstract channel. This section asks which physical systems could produce the lamp channel $\Phi$, and what a
system left to itself can produce in the long run. The answers are given at the scale of Theorem~B. A channel is \emph{cheap} if for all
$n$ and $\epsilon$ it has causal services for $n$ rounds with error at most $\epsilon$, $P=C=0$ and memory polynomial in $\log(n/\epsilon)$.
It is \emph{expensive} if for some $A,c>0$ every service for $n$ rounds with error at most $1/16$ and $P+C\le A\log(n+1)$ needs memory at
least $n^c$ for all large $n$. Neither notion constrains the exchange.

Local dynamics with a maximally mixed bath are cheap at every finite time, so they need time about $\delta^{-1/\nu}$ to come within
$\delta$ of $\Phi$ (Section~\ref{sec:physics-local}). An autonomous dynamics with a tracial bath has only self-adjoint, positive channels as
long-time limits (Section~\ref{sec:physics-autonomous}). Those two properties do not force cheapness: a lazy doubling of $\Phi$ has both, is
the exact value of an autonomous dynamics at every time $t\ge1$, and is still expensive (Section~\ref{sec:physics-lazy}). Finer constraints
on autonomous limits, which exclude some exact finite factors, are in Appendix~\ref{sec:physics-finer}.

\subsection{Local reservoirs are cheap at every finite time}\label{sec:physics-local}

\paragraph{Local reservoirs.} A \emph{local reservoir} is a finite box of the lattice $\Z^\nu$ with a system of dimension at most $q$ at each
site, the probe $\mathbb C^d$ attached to the origin, and a Hamiltonian $H(t)=\sum_Xh_X(t)$, possibly time dependent. Its terms are
supported in sets $X$ of diameter at most $R_0$ and satisfy $\|h_X(t)\|\le J_0$. The reservoir starts in its tracial state, and the probe is
coupled to it for a time $T$. The resulting channel on the probe is $\Phi_T(X)=(\mathrm{id}\otimes\tau)[W_T(X\otimes I)W_T^*]$, where $W_T$ is
the time-ordered evolution and $\tau$ the normalized trace of the reservoir. We call $\nu,q,R_0,J_0$ and $d$ the \emph{parameters} of the
reservoir; the size of the box is arbitrary.

\begin{proposition}[Local reservoirs are cheap at every finite time]\label{prop:light-cone}
There is $K$, depending only on the parameters, such that for every local reservoir, every $T\ge0$, every $n\ge1$ and every $0<\epsilon\le1$, the
channel $\Phi_T$ has a causal service for $n$ rounds with error at most $\epsilon$, $P=C=0$, $J=(n-1)Q$ and
\[
Q\ \le\ K\,\bigl[1+T+\log(n/\epsilon)\bigr]^\nu .
\]
\end{proposition}

\begin{proof}
For $r\ge0$ let $\Phi_T^{(r)}$ be the channel produced by the terms $h_X$ with $X$ inside the ball $B_r$ of radius $r$ about the origin. The
Lieb--Robinson bound~\cite{LiebRobinson1972,BHV2006,NSY2019} gives constants $C_1,v,\mu>0$, depending only on the parameters, such that the Heisenberg
evolutions of an operator $A$ at the origin under the full and the truncated dynamics differ in norm by at most $C_1\|A\|e^{vT-\mu r}$. The bound
does not depend on the dimension of the site at the origin, so it holds for operators on the probe tensored with a reference system that the
dynamics does not touch. Hence the Heisenberg-picture maps of $\Phi_T$ and $\Phi_T^{(r)}$ differ in completely bounded norm by at most
$C_1e^{vT-\mu r}$, and by duality $\ddia(\Phi_T,\Phi_T^{(r)})\le C_1e^{vT-\mu r}$. The truncated evolution acts only on the probe and the sites
of $B_r$, and the other sites start in their tracial state, so $\Phi_T^{(r)}$ has a finite tracial factorization with bath dimension
$m\le q^{(2r+1)^\nu}$.

Choose $r=\lceil(vT+\ln(2C_1n/\epsilon))/\mu\rceil$, so that $\ddia(\Phi_T,\Phi_T^{(r)})\le\epsilon/(2n)$, and let $Q=\lceil\log m+\log(2n/\epsilon)\rceil$.
The exchange construction in the proof of Proposition~\ref{prop:sofic-upper}, with $U_f$ replaced by the unitary of this factorization and
$\vartheta=0$, gives a service of $\Phi_T^{(r)}$ in which every round is within $\epsilon/(2n)$ of $\Phi_T^{(r)}$, hence within $\epsilon/n$ of
$\Phi_T$. A hybrid argument bounds the error by $\epsilon$. The memory is $Q\le(2r+1)^\nu\log q+\log(2n/\epsilon)+1\le K[1+T+\log(n/\epsilon)]^\nu$.
\end{proof}

With Theorem~B this says that local physics cannot produce the lamp channel quickly.

\begin{corollary}[Preparation time]\label{cor:preparation}
Let $\Phi$ be the channel of Theorem~\ref{thm:main-channel}, and fix the parameters of a local reservoir with $d=873$. There are $c,\delta_0>0$ such
that, whenever such a reservoir produces at time $T$ a channel within $\delta\le\delta_0$ of $\Phi$ in diamond distance,
\[
T\ \ge\ c\,\delta^{-1/\nu}\bigl(\log(1/\delta)\bigr)^{-(2\kappa+1)/\nu},\qquad 2\kappa+1<26.49 .
\]
\end{corollary}

\begin{proof}
Let $n=\lfloor1/(32\delta)\rfloor$. Proposition~\ref{prop:light-cone} gives a service of $\Phi_T$ with error at most $1/32$, zero purity and
$Q\le K[1+T+\log(32n)]^\nu$. Replacing the $n$ uses of $\Phi_T$ by uses of $\Phi$ one at a time changes the observer's final state by at most
$n\delta\le1/32$, so the same device serves $\Phi$ with error at most $1/16$ and zero purity. For $\delta_0$ small, $n$ is large, and
Theorem~\ref{thm:main-channel}(a) with $A_0=0$ gives $Q\ge c'n/(\log n)^{2\kappa+1}$. Hence
$1+T+\log(32n)\ge\bigl(c'n/(K(\log n)^{2\kappa+1})\bigr)^{1/\nu}$, and since $n\ge1/(64\delta)$ this gives the claim for $\delta_0$ small.
\end{proof}

The maximally mixed start of the reservoir is essential: $\Phi$ has a Stinespring dilation with a pure environment of dimension $r_*$, so
from a pure state a bounded number of sites produce $\Phi$ exactly in bounded time. The same argument applies to a circuit of depth $T$ of
gates of bounded range, whose light cone is exact, and so bounds the depth of any local circuit that prepares $\Phi$. The corollary does not
exclude that local dynamics produce $\Phi$ in the limit of long times; it says how slowly. What the limit of long times can be, for a system
left to itself, is the question of the rest of the section.

\subsection{Autonomous limits are self-adjoint}\label{sec:physics-autonomous}

\paragraph{Autonomous dynamics.} Let $\mathcal B$ be a von Neumann algebra with a faithful normal tracial state $\tau_{\mathcal B}$. Put
$\mathcal M=M_d\otimes\mathcal B$ with the trace $\tau=\mathrm{tr}_d\otimes\tau_{\mathcal B}$, where $\mathrm{tr}_d$ is the normalized trace, and let
$\iota(a)=a\otimes1$ and $E=\mathrm{id}\otimes\tau_{\mathcal B}$. An \emph{autonomous dynamics} is a trace-preserving automorphism $\alpha$ of
$\mathcal M$, and the channel after $t\in\N$ steps is $\Phi_t=E\circ\alpha^t\circ\iota$. For $\alpha=\mathrm{Ad}\,U$ with $\mathcal B=M_m$ this is
the tracial factorization $\Phi_t(X)=(\mathrm{id}\otimes\tau_m)[U^t(X\otimes I)U^{-t}]$. Let $\mathcal N$ be the fixed-point algebra of $\alpha$
and $E_{\mathcal N}$ the $\tau$-preserving conditional expectation onto it.

The next proposition is standard. It combines the mean ergodic theorem with an observation of Junge, Le Merdy and
Xu~\cite[Remark~10.3]{JungeLeMerdyXu2006}, and of Haagerup and Musat~\cite[Remark~5.2]{HaagerupMusat2011}: compressions
$E\circ E_{\mathcal N}\circ\iota$ of trace-preserving conditional expectations are self-adjoint and positive.

\begin{proposition}[Autonomous limits are self-adjoint]\label{prop:autonomous}
The averages $\frac1T\sum_{t=0}^{T-1}\Phi_t$ converge in diamond norm to $\overline\Phi=E\circ E_{\mathcal N}\circ\iota$, and if $\Phi_t$ converges, its
limit is $\overline\Phi$. The channel $\overline\Phi$ is self-adjoint and positive semidefinite for the Hilbert--Schmidt inner product on $M_d$. In
particular, a channel $\Psi$ with $\Psi^*\ne\Psi$ is not the long-time limit of any autonomous dynamics with a tracial bath, local or not.
\end{proposition}

\begin{proof}
The automorphism extends to a unitary $V$ on $L^2(\mathcal M,\tau)$. By the mean ergodic theorem~\cite{Krengel1985}, $\frac1T\sum_{t<T}V^t$ converges
strongly to the orthogonal projection $P$ onto the invariant vectors. For $x\in\mathcal M$ the averages of $x$ lie in the ball of radius $\|x\|$ of
$\mathcal M$, which is closed in $L^2$, so $Px\in\mathcal M$, and $Px$ is invariant; hence $P(\mathcal M)\subseteq\mathcal N$. Since
$\mathcal N$ consists of invariant vectors and $\mathcal M$ is dense in $L^2(\mathcal M)$, $P$ is the orthogonal projection onto
$L^2(\mathcal N)$, whose restriction to $\mathcal M$ is $E_{\mathcal N}$. The map $\iota$ is an isometry from $(M_d,\mathrm{tr}_d)$ into
$L^2(\mathcal M,\tau)$ with adjoint $E$, so $\Phi_t=\iota^*V^t\iota$, and the averages converge to $\iota^*P\iota=\overline\Phi$ on every $a\in M_d$,
hence in diamond norm, since $M_d$ is finite dimensional. An ordinary limit equals the limit of the averages. And
$\mathrm{tr}_d(b^*\overline\Phi(a))=\langle P\iota(b),P\iota(a)\rangle$ is Hermitian in $(a,b)$ and nonnegative for $a=b$.
\end{proof}

Conversely, every channel of the form $\Psi\Psi^*$ with $\Psi$ in the tracial closure is a compression of a conditional expectation:
apply~\cite[Theorem~5.9]{HaagerupMusat2011} to $\Psi^*$, which is again factorizable. It is in fact an autonomous long-time limit. The dynamics
that produces it adapts K\"ummerer's tensor dilations with a Bernoulli shift~\cite{Kummerer1985}, and nothing in the construction makes it local.

\begin{proposition}[Autonomy without locality]\label{prop:nonlocal}
Let $\mathcal N_0$ be a von Neumann algebra with a faithful normal tracial state $\tau_0$, let $U$ be a unitary in $M_d\otimes\mathcal N_0$, and put
$\Psi(X)=(\mathrm{id}\otimes\tau_0)[U(X\otimes1)U^*]$. Let $\mathcal B=\overline\bigotimes_{k\in\Z}(\mathcal N_0,\tau_0)$, with $U$ acting on the factor
$k=0$, and let $\beta$ be the Bernoulli shift of $\mathcal B$. Then the autonomous dynamics $\alpha=\mathrm{Ad}\,U\circ(\mathrm{id}\otimes\beta)\circ\mathrm{Ad}\,U^*$
has $\Phi_t\to\Psi\Psi^*$, where $\Psi^*(a)=(\mathrm{id}\otimes\tau_0)[U^*(a\otimes1)U]$ is the Hilbert--Schmidt adjoint of $\Psi$.
\end{proposition}

\begin{proof}
We have $\alpha^t=\mathrm{Ad}\,U\circ(\mathrm{id}\otimes\beta^t)\circ\mathrm{Ad}\,U^*$. The shift is mixing: $\tau_{\mathcal B}(w\,\beta^t(z))\to\tau_{\mathcal B}(w)\tau_{\mathcal B}(z)$ for
$w,z\in\mathcal B$. Equality holds as soon as $w$ and $\beta^t(z)$ are elementary tensors with disjoint supports, and by Kaplansky's density
theorem such tensors of norm at most $\|w\|$ and $\|z\|$ approximate $w$ and $z$ in $2$-norm. Let $a\in M_d$ and $y=U^*(a\otimes1)U$, and write
$y=\sum_{i,j}e_{ij}\otimes y_{ij}$. For every $x=\sum_{k,l}e_{kl}\otimes x_{kl}$ in $M_d\otimes\mathcal B$,
\[
\tau\bigl(x\,(\mathrm{id}\otimes\beta^t)(y)\bigr)=\frac1d\sum_{k,l}\tau_{\mathcal B}\bigl(x_{kl}\beta^t(y_{lk})\bigr)\ \longrightarrow\ \tau\bigl(x\,(E(y)\otimes1)\bigr).
\]
Taking $x=U^*(e_{ji}\otimes1)U$ shows that every matrix entry of $\Phi_t(a)=E[U(\mathrm{id}\otimes\beta^t)(y)U^*]$ converges to the corresponding entry
of $E[U(E(y)\otimes1)U^*]=\Psi(\Psi^*(a))$. Convergence on a basis of $M_d$ gives convergence in diamond norm.
\end{proof}

By a theorem of Haagerup and Musat~\cite[Theorem~3.6]{HaagerupMusat2015}, every channel in $\cK_d$ has the form $\Psi$ of
Proposition~\ref{prop:nonlocal} for some $\mathcal N_0$, and the lamp channel has it with $\mathcal N_0=L(\Gamma)$. Nothing in
Proposition~\ref{prop:nonlocal} makes the dynamics local. The unitary $U$ is arbitrary in $M_d\otimes\mathcal N_0$, and whether a dynamics of
this kind can be given by a local rule is open (Section~\ref{sec:open}).

Applied to the lamp channel, Proposition~\ref{prop:nonlocal} makes $\Phi\Phi^*$ an autonomous limit. That limit is cheap, since $\Phi\Phi^*$
has an exact finite factorization with a bounded bath (Proposition~\ref{prop:compositions}). An expensive limit needs a construction that keeps
more of the group, and the next subsection gives one.

\subsection{An expensive autonomous limit: the lazy doubling}\label{sec:physics-lazy}

The limit we build is a self-adjoint, positive and expensive channel on a doubled system, reached in two steps. The first step doubles $\Phi$
to a self-adjoint channel $\Theta$ that keeps the lower bounds of Theorem~B. The second averages $\Theta$ with the identity, which makes it
positive and leaves it expensive.

\begin{remark}[Self-adjoint expensive channels]\label{rem:doubling}
A self-adjoint channel can still be expensive. Let $z$ be the unitary of order three in $L(\Z_3)$, so that $\tau(z)=\tau(z^2)=0$, let $U$ be the gadget
unitary of the channel $\Phi$ of Theorem~B, and let
$V=\bigl(\begin{smallmatrix}0&U\otimes z\\U^*\otimes z^*&0\end{smallmatrix}\bigr)$, a self-adjoint unitary over $\mathbb C^2\otimes\mathbb C^{873}$ and
$L(\Gamma\times\Z_3)$. The channel it defines on $M_{1746}$ is
$\Theta\bigl(\begin{smallmatrix}x&y\\u&w\end{smallmatrix}\bigr)=\mathrm{diag}\bigl(\Phi(w),\Phi^*(x)\bigr)$, which is self-adjoint. It keeps the lower bounds of
Theorem~B. The memory-against-purity bound transfers: a service of $\Theta$ becomes a service of $\Phi$ with one more memory qubit and one more
bit of purity, at the same exchange. Feed each input together with a retained tag qubit that starts in $|1\rangle$ and is flipped after every
round; this works because $\Theta$ maps $|1\rangle\langle1|\otimes\rho$ to $|0\rangle\langle0|\otimes\Phi(\rho)$. The bound at the rank rate transfers
as well. $\Theta$ has a flat probe of rank $2r_*$, with anchors in both blocks, so Proposition~\ref{prop:purity-budget} gives $P+C<3Q+6$ at its
own rank rate $J\le\frac n2\log(2r_*)$; the tag transfer and Theorem~\ref{thm:main-channel}(a) then give $Q\ge c\sqrt n/(\log n)^\kappa$. But
$\Theta$ maps $\mathrm{diag}(I,-I)$ to $-\mathrm{diag}(I,-I)$, so it is not positive and not an autonomous limit. Its square
$\Theta^2=\mathrm{diag}(\Phi\Phi^*,\Phi^*\Phi)$ is an autonomous limit, by Proposition~\ref{prop:nonlocal}, and it is cheap
(Proposition~\ref{prop:compositions}).
\end{remark}

\paragraph{A lazy doubling.} The doubled channel $\Theta$ is expensive but not positive, and its square is cheap. Averaging $\Theta$ with the
identity restores positivity and keeps it expensive. Put
\[
\Theta_+=\tfrac12(\mathrm{id}+\Theta),\qquad
\Theta_+\begin{pmatrix}x&y\\u&w\end{pmatrix}=\frac12\begin{pmatrix}x+\Phi(w)&y\\u&w+\Phi^*(x)\end{pmatrix}\qquad\text{on }M_{1746}.
\]
Remark~\ref{rem:doubling} uses only $zz^*=1$ and $\tau(z^2)=0$, so $z$ may also be the unitary $\mathrm{diag}(1,i)$ of a qubit with its
normalized trace. Since $\Theta$ is a unital channel that is self-adjoint for the Hilbert--Schmidt inner product, its spectrum lies in $[-1,1]$;
unital channels are contractions for the Hilbert--Schmidt norm. Hence $\Theta_+$ is self-adjoint and positive semidefinite, as
Proposition~\ref{prop:autonomous} requires of an autonomous limit.

\begin{proposition}[An expensive autonomous limit]\label{prop:lazy}
Let $\Phi$, $V$ and $\Theta$ be as in Remark~\ref{rem:doubling}, and $\kappa=6+\log107$.
\begin{enumerate}[label=(\alph*),itemsep=1pt]
\item \emph{Autonomy.} Let $Z$ be the coordinate of $L^\infty[-1,1]$ with the normalized Lebesgue measure, put
$\mathcal B=L(\Gamma\times\Z_3)\,\overline\otimes\,L^\infty[-1,1]$ with the product trace, and
$W=\exp\bigl(i\tfrac\pi2V\otimes Z\bigr)=\cos(\tfrac\pi2Z)+i\sin(\tfrac\pi2Z)\,V\in M_{1746}\otimes\mathcal B$. The autonomous dynamics
$\alpha=\mathrm{Ad}\,W$ has $\Phi_t=\Theta_+$ for every $t\ge1$. More generally, $\exp(isV\otimes Z)$ gives the channel
$\Xi_s=\frac{1+\mathrm{sinc}\,2s}2\,\mathrm{id}+\frac{1-\mathrm{sinc}\,2s}2\,\Theta$, with $\mathrm{sinc}\,x=\sin x/x$, and
$\ddia(\Xi_s,\Theta_+)=|\sin2s|/(4s)$.
\item \emph{Memory against purity, lower bound.} There are $c,c_0>0$ and $n_0$ such that for $n\ge n_0$ and $\log n\le B\le c_0n$ every causal
service of $\Theta_+$ for $n$ rounds with error at most $1/16$ and purity $P+C\le B$ has, at any exchange,
\[
Q\ \ge\ c\,\frac n{B\,(\log n)^{2\kappa}} .
\]
In particular, for every $A_0\ge0$ there is $c'>0$ such that services with $P+C\le A_0\log n$ have $Q\ge c'n/(\log n)^{2\kappa+1}$ for all large $n$.
\item \emph{Upper bound.} There is $K$ such that for every $n\ge1$ and $0<B\le n$, $\Theta_+$ has an exact causal service with purity $P+C<B$,
$J=(n-1)Q$ and $Q\le K(n/B)\log(2n/B)$. Moreover $\Theta_+\in\cK_{1746}$.
\end{enumerate}
\end{proposition}

Thus $\Theta_+$ is the long-time limit, indeed the exact value at every time $t\ge1$, of an autonomous dynamics with a tracial bath. It is
expensive in the sense of this section, and it trades memory and purity at the same rate as $\Phi$ (Theorem~\ref{thm:main-channel}(a)). Like
the dynamics of Proposition~\ref{prop:nonlocal}, the construction is not local. Part (a) is a computation, and we give it here.
Appendix~\ref{app:lazy} proves parts (b) and (c).

Part (b) loses nothing against Theorem~\ref{thm:main-channel}(a). A service of $\Theta_+$ yields, by Theorem~\ref{thm:extraction}, one round
$(U,\rho)$ on the bath $\mathbb C^M$, $M=2^Q$, with small leakage and entropy production. Write $U$ in $2\times2$ blocks with respect to the
qubit factor of the doubled system. The upper right block $U_{01}$ is where $\Theta_+$ applies $\Phi$. It is nearly a round of $\Phi$ on the same
bath, except that it is not unitary: in the norm weighted by $\rho$, $U_{01}^*U_{01}$ is close to $1\otimes R$ for an operator $0\le R\le1$ on the bath. Where $R$ is bounded away from $0$, the operator $U_{01}R^{-1/2}$ is nearly an isometry. A smooth cutoff restricts to that part of the bath,
and a polar decomposition produces a unitary $U'$ on $\mathbb C^{873}\otimes\mathbb C^M$ and a state $\rho'$, the state $\rho$ reweighted by the
cutoff (Lemma~\ref{lem:block-extraction}). The leakage of the new round and its commutator $\|[U',1\otimes\sqrt{\rho'}\,]\|_F^2$ are at most a
constant times the leakage and the commutator $\|[U,1\otimes\sqrt\rho\,]\|_F^2$ of the old round, a linear loss. The one-round theorem accepts
the commutator in place of entropy production, because its proof uses entropy production only to bound that commutator. It also accepts, in
place of a small diamond distance to $\Phi$, a bound on one coherence between two anchors, which the block extraction gives as well
(Lemma~\ref{lem:commutator-form}). Corollary~\ref{cor:amplification} then yields the bound of Theorem~\ref{thm:main-channel}(a) for $\Theta_+$,
with the same exponent.

For part (c), two maximally mixed qubits join the bath of a sofic model of the labels: a coin that chooses between the identity and $\Theta$,
and a phase qubit carrying $\mathrm{diag}(1,i)$ in place of $z$. The model then serves $\Theta_+$ exactly at its good points. The exact service of
Proposition~\ref{prop:sofic-upper}(c) then goes through unchanged, and the same unitary with a flat bath puts $\Theta_+$ in $\cK_{1746}$.

\begin{proof}[Proof of (a)]
Since $V^2=1$ and $Z$ commutes with $V$, $\exp(isV\otimes Z)=\cos(sZ)+i\sin(sZ)V$, so that $W^t=\exp(it\frac\pi2V\otimes Z)$ and
\[
W^t(X\otimes1)W^{-t}=\cos^2(\tfrac{t\pi}2Z)\,X+\sin^2(\tfrac{t\pi}2Z)\,V(X\otimes1)V+\tfrac i2\sin(t\pi Z)\bigl(V(X\otimes1)-(X\otimes1)V\bigr).
\]
The trace of $\mathcal B$ is the product of the trace of $L(\Gamma\times\Z_3)$ and the integral over $Z$. The last term integrates to $0$ because
$\sin(t\pi Z)$ is odd, $\int\cos^2(\frac{t\pi}2Z)=\frac12(1+\mathrm{sinc}\,t\pi)=\frac12$ for $t\ge1$, and $(\mathrm{id}\otimes\tau)[V(X\otimes1)V]=\Theta(X)$.
So $\Phi_t=\frac12(\mathrm{id}+\Theta)$. The same computation with $s$ in place of $\frac{t\pi}2$ gives $\Xi_s$. Finally
$\Xi_s-\Theta_+=\frac12\mathrm{sinc}(2s)(\mathrm{id}-\Theta)$, and $\|\mathrm{id}-\Theta\|_\diamond=2$: the input $|0\rangle\langle0|\otimes\sigma$ is mapped
by $\Theta$ into the block $|1\rangle\langle1|$. The automorphism $\mathrm{Ad}\,W$ is inner, hence trace preserving.
\end{proof}

The same channel is reached, for every $t\ge1$, by a Bernoulli-shift construction like that of Proposition~\ref{prop:nonlocal}. Take
$\mathcal B=L(\Gamma\times\Z_3)\,\overline\otimes\,\overline\bigotimes_{k\in\Z}(M_2,\tau_2)$, let $\beta$ be the shift of the chain, and let $p_0$ be
the projection $|1\rangle\langle1|$ of the qubit at site $0$. The automorphism $\mathrm{Ad}\bigl((1-p_0)+p_0V\bigr)\circ(\mathrm{id}\otimes\beta)$ raises
$V$ to the parity of $t$ fair bits.

Autonomous dynamics with a tracial bath can therefore produce expensive channels. The reduced dynamics of the probe matters. If the channels
$\Phi_t$ themselves form a semigroup of unital channels, they are mixed unitary for all large $t$~\cite{BhatDevendra2025}, hence cheap; here
$\Theta_+^2\ne\Theta_+$. Locality constrains how $\Theta_+$ can be approached. Proposition~\ref{prop:light-cone} and
Proposition~\ref{prop:lazy}(b) combine as in the proof of Corollary~\ref{cor:preparation}. Together they show that a local reservoir whose
probe is the doubled system, and which produces a channel within $\delta\le\delta_0$ of $\Theta_+$ at time $T$, has
$T\ge c\,\delta^{-1/\nu}(\log(1/\delta))^{-(2\kappa+1)/\nu}$, with constants $c,\delta_0>0$ depending only on the parameters of the reservoir. No
local reservoir, then, produces $\Theta_+$ exactly at a finite time or approaches it exponentially fast. Whether a local dynamics can produce an
expensive channel at all is open (Section~\ref{sec:open}). Self-adjointness and positivity, which $\Theta_+$ has, are not all that an
autonomous limit must satisfy: Appendix~\ref{sec:physics-finer} gives a sharper necessary condition and excludes an exact finite
factor by it.
\section{Open problems}\label{sec:open}

\paragraph{Groups.}

\begin{enumerate}[label=(\arabic*)]
\item \emph{The exact profile.} Theorem~A gives $\exp(r^{1+o(1)})$, with polylogarithmic factors in the exponent between
$c\,r/(\log r)^{\kappa}$, $\kappa=6+\log107$, and $C\,r\log r$. The exponent $\kappa$ comes from the Dehn function of $H_3$, known to be polynomial of
degree at most $\log107+4$ and at least $2$~\cite{DouglasMghirbiA}; a quadratic Dehn function would give $\kappa=4$. What is the correct
polylogarithmic factor; is $\log\sigma_E(r)=\Theta(r\log r)$?
\item \emph{Beyond the counting ceiling.} Does some amenable group have a chunk with profile $\exp(c\,r)$? Does any sofic group have a chunk with
profile $\exp(r^{1+c})$ for some $c>0$? The counting criterion cannot give either; a different lower-bound mechanism would be needed. For the
second question, grids of fixed dimension should not help: the argument sketched in Remark~\ref{rem:grid} would bound the profile of their
chunks by $\exp(O(r\log r))$. By
Remark~\ref{rem:sandwich}, a gadget channel of a sofic group can need superlinear memory at zero purity only if its chunk profile exceeds
$\exp(O(r))$. Every fixed chunk of a sofic group has a computable profile~\cite[Lemma~1]{Ivanov2021} (search the degrees $1,2,\dots$ and the finitely
many maps at each), so no chunk outgrows every computable function. How fast the profile of a chunk of a finitely presented amenable group can grow is
open.
\item \emph{The dyadic group.} Its profile is finite and at least $\exp(c\,r^{1/3})$. In $\mathcal G$ the words $w_0=1$,
$w_{2n}=uw_n$ and $w_{2n+1}=tuw_n$ reach the integer $n$ with length $O(\log n)$, and commutators of copies at points of equal parity reduce
freely to smaller indices. What is the area of the commutators at points of opposite parity?
\item \emph{Thompson's groups.} Is $\mathcal L$ sofic whenever $F$ is? The hyperlinearity transfer of Proposition~\ref{prop:thompson-hyperlinear}
does not obviously give permutations, and not every action of a sofic group is sofic: some coset actions of residually finite Kazhdan groups
are not~\cite{KunThom2026}. We do not know whether the action of $F$ on $\mathbb D$ is sofic when $F$ is. Does $\Gamma$ embed in $V$ or
in $2V$? The embedding of Proposition~\ref{prop:3V} uses one coordinate for each ray.
\end{enumerate}

\paragraph{Channels and physics.}

\begin{enumerate}[label=(\arabic*),resume]
\item \emph{The optimal exchange rate in general.} For $\Phi$ the least memory at the rank rate and the least memory with purity at most the
memory are both $n^{1/2+o(1)}$ (Remark~\ref{rem:sqrt-limit}). For every channel with a flat probe the second is at most three times the first, plus
six. Is the converse true up to polylogarithmic factors for every channel in $\cK_d$ with a flat probe? Theorem~\ref{thm:streaming} bounds the
rank-rate memory by $\log k+n\log(k'/k)+O(\log n)$, so the converse holds whenever some exact rectangular dilation makes this at most a
polylogarithmic factor times the least memory at purity at most memory.
\item \emph{Local autonomous physics.} Proposition~\ref{prop:lazy} gives an expensive long-time limit of an autonomous dynamics with a tracial
bath, but the dynamics is not local. Can an expensive channel be the long-time limit of a fixed translation-invariant local rule, such as a
quantum cellular automaton or a time-independent finite-range Hamiltonian, with the bath in its tracial state? By Proposition~\ref{prop:light-cone}
such a limit would be approached slowly; for $\Theta_+$ the distance at time $t$ is at least of order $t^{-\nu}$ up to logarithmic factors
(Section~\ref{sec:physics}). By~\cite{BhatDevendra2025} the reduced dynamics of the probe could not be a semigroup.
\end{enumerate}

\section*{Use of AI}
\addcontentsline{toc}{section}{Use of AI}

This work was carried out with extensive use of generative AI systems. They proposed constructions and proof strategies, wrote and
checked proofs, carried out computations, searched the literature, wrote the verification code, and drafted and edited the text. The
authors directed the research, assessed and selected the results, and take full responsibility for the content. The proofs are given in
full in the text; repeated adversarial checking and the accompanying programs support them but do not replace them. The verification
notes state which statements are checked by machine.

\appendix
\section{Tools from Houghton's group}\label{sec:houghton}

Section~\ref{sec:fillings} uses three facts about $H=H_3$, and this appendix proves them. Lemma~\ref{lem:transpositions} gives short words for
finite permutations, Lemma~\ref{lem:stabilizers} gives generators of point stabilizers with linear word length, and
Proposition~\ref{prop:dehn} bounds the Dehn function. A comparison of conventions with~\cite{DouglasMghirbiA} comes first, because permutations
act on the right there. All transport in Sections~\ref{sec:fillings} and~\ref{sec:pair-areas} takes place among the first points of ray~$1$, so
we need words only for permutations supported there. For $n\ge1$ let $\Sigma_n\le H$ be the group of permutations of $Y$ supported in
$\{(1,1),\dots,(1,n)\}$, identified with $\Sym(n)$; we write $(x\ y)$ for the transposition of the points $x$ and $y$ of the first ray.

\paragraph{Conventions.} Put $\alpha=[p,q]$; a direct check shows $\alpha=(1\ 2)$. In~\cite{DouglasMghirbiA} permutations act on the right.
Johnson's presentation of $H$~\cite{Johnson1999}, also in~\cite[Theorem~2.8]{Lee2012}, is written there in generators $a,b,\alpha$ that are the
inverses, as maps, of our $p,q,\alpha$. The map $g\mapsto g^{-1}$ is an isomorphism from the permutations of $Y$ with the product
of~\cite{DouglasMghirbiA} to the permutations of $Y$ under composition, and it sends $a,b,\alpha$ to $p,q,\alpha$. A word in $a,b,\alpha$
therefore represents $1$ in the conventions of~\cite{DouglasMghirbiA} if and only if the same word in $p,q,\alpha$ represents $1$ here. Hence
the relators of~\cite[Section~2]{DouglasMghirbiA}, read in our letters, give the presentation
\begin{equation}\label{eq:johnson}
P_J=\bigl\langle p,q,\alpha\bigm|\ \alpha^2,\ (\alpha p\alpha p^{-1})^3,\ [\alpha,p^2\alpha p^{-2}],\ [p,q]\alpha^{-1},\ p\alpha p^{-1}q\alpha^{-1}q^{-1}\bigr\rangle
\end{equation}
of $H$; in particular $H=\langle p,q\rangle$. Renaming the letters is an isomorphism of free groups that carries the relators
of~\cite{DouglasMghirbiA} to those of $P_J$, so every area computed in~\cite{DouglasMghirbiA} for Johnson's presentation holds verbatim for $P_J$.

Eliminating $\alpha$ by the relator $[p,q]\alpha^{-1}$, a Tietze transformation, gives the presentation
\begin{equation}\label{eq:PH}
P_H=\bigl\langle p,q\bigm|\ s^2,\ (sps\,p^{-1})^3,\ [s,p^2sp^{-2}],\ psp^{-1}qs^{-1}q^{-1}\bigr\rangle,\qquad s=[p,q],
\end{equation}
of $H$. Let $\delta_H$ be its Dehn function: $\delta_H(n)$ is the largest area, with respect to the four relators of $P_H$, of a word of length at most
$n$ in $p^{\pm1},q^{\pm1}$ that represents $1$ in $H$. It is nondecreasing.

\begin{lemma}[Finite permutations]\label{lem:transpositions}
Put $\sigma_k=p^k\alpha p^{-k}$ for $k\ge0$.
\begin{enumerate}[label=(\alph*),itemsep=1pt]
\item $\sigma_k=(k+1\ \,k+2)$, and for $1\le x<y$ the word
$\omega_{x,y}=p^{x-1}(\alpha p)^{y-x-1}(\alpha p^{-1})^{y-x-1}\alpha\,p^{1-x}$ represents $(x\ y)$. With $\alpha=pqp^{-1}q^{-1}$ it has length at most
$10y-8x-8<10y$ in $p,q$.
\item Every element of $\Sigma_n$ that moves $m\ge2$ points is represented by a word in $p,q$ of length less than $10(m-1)n$.
\item Let $x_1,\dots,x_l\in\{1,\dots,n\}$ be distinct. Put $\pi_0=1$ and, for $m=1,\dots,l$, $w_m=\pi_{m-1}^{-1}(x_m)$ and
$\pi_m=\pi_{m-1}(m\ w_m)$, where $(m\ m)=1$. Then $\pi[x_1,\dots,x_l]:=\pi_l$ sends $m$ to $x_m$ for $m\le l$, it is a product of at most $l$
transpositions of points in $\{1,\dots,n\}$, and the product of the corresponding words $\omega$ represents it with length less than $10ln$.
\end{enumerate}
\end{lemma}

\begin{proof}
(a) Since $p^k$ maps $(1,1),(1,2)$ to $(1,k+1),(1,k+2)$, $\sigma_k$ is the stated transposition. Hence
$(x\ y)=\sigma_{x-1}\sigma_x\cdots\sigma_{y-3}\,\sigma_{y-2}\,\sigma_{y-3}\cdots\sigma_{x-1}$, by induction on $y-x$ from
$(x\ y)=\sigma_{x-1}(x{+}1\ \,y)\sigma_{x-1}$. In the product the powers of $p$ telescope: $\sigma_k\sigma_{k+1}=p^k\alpha p\alpha p^{-k-1}$ and
$\sigma_{k+1}\sigma_k=p^{k+1}\alpha p^{-1}\alpha p^{-k}$, which gives $\omega_{x,y}$ after free reduction. Its length is
$2(x-1)+10(y-x-1)+4$.

(b) A permutation moving $m$ points is a product of at most $m-1$ transpositions of those points, each of length less than $10n$ by (a).

(c) Since $\pi_{m-1}(j)=x_j\ne x_m$ for $j<m$, we have $w_m\ge m$, so $(m\ w_m)$ fixes $1,\dots,m-1$ and $\pi_m(j)=x_j$ for $j\le m$. Each factor
has length less than $10n$ by~(a).
\end{proof}

That a point stabilizer of a Houghton group is isomorphic to the group is known~\cite[Example~3.6]{Cornulier2006}; part~(a) of
Lemma~\ref{lem:stabilizers} gives such isomorphisms explicitly on generators.

\begin{lemma}[Point stabilizers]\label{lem:stabilizers}
For $k\ge1$ let $H^{(k)}$ be the pointwise stabilizer in $H$ of $(1,1),\dots,(1,k)$, put $\zeta_k=\sigma_{k-1}\cdots\sigma_1\sigma_0$ and
$K_k=(\zeta_kp,\ \zeta_kq)$.
\begin{enumerate}[label=(\alph*),itemsep=1pt]
\item There is an isomorphism $H\to H^{(k)}$ that sends $p,q,\alpha$ to $\zeta_kp,\zeta_kq,\sigma_k$ and the transposition $(x\ y)$ to
$(x{+}k\ \ y{+}k)$. In particular $\zeta_kp$ and $\zeta_kq$ generate $H^{(k)}$.
\item Every element of $\Sigma_n\cap H^{(k)}$ that moves $m\ge2$ points is represented by a word of length less than $10(m-1)n$ in $\zeta_kp$ and
$\zeta_kq$.
\end{enumerate}
In $p,q$ the words $\zeta_1p=[p,q]\,p$ and $\zeta_2p=p[p,q]p^{-1}[p,q]\,p$ have length at most $5$ and $11$, and likewise for $q$.
\end{lemma}

\begin{proof}
(a) Let $\theta_k:Y\to Y\setminus\{(1,1),\dots,(1,k)\}$ send $(1,j)$ to $(1,j+k)$ and fix rays $2$ and $3$. For $g\in H$ let $g^{[k]}$ act as
$\theta_kg\theta_k^{-1}$ on the image of $\theta_k$ and fix $(1,1),\dots,(1,k)$. Then $g^{[k]}$ is eventually a translation on each ray, and
$g\mapsto g^{[k]}$ is an isomorphism $H\to H^{(k)}$, since an element of $H^{(k)}$ restricts to a permutation of the image of $\theta_k$. Clearly
$(x\ y)^{[k]}=(x{+}k\ \ y{+}k)$, so $\alpha^{[k]}=\sigma_k$. The element $\zeta_k$ is the cycle $(1,1)\mapsto(1,k+1)\mapsto(1,k)\mapsto\dots\mapsto(1,2)\mapsto(1,1)$.
Hence $\zeta_kp$ fixes $(1,j)$ for $j\le k$, maps $(1,j)$ to $(1,j+1)$ for $j>k$, maps $(2,1)$ to $(1,k+1)$, and agrees with $p$ elsewhere; that is,
$\zeta_kp=p^{[k]}$. Likewise $\zeta_kq=q^{[k]}$.

(b) An element $g\in\Sigma_n\cap H^{(k)}$ moving $m$ points is $g'^{[k]}$ for some $g'\in\Sigma_{n-k}$ moving $m$ points. Write $g'$ as in
Lemma~\ref{lem:transpositions}(b) and substitute $\zeta_kp,\zeta_kq$ for $p,q$.
\end{proof}

Part~(a) of the next proposition is Lee's theorem~\cite[Theorem~D]{Lee2012}, and part~(b) is~\cite[Theorems~3.2 and~3.3]{DouglasMghirbiA}. Both
pass to $P_H$ by a Tietze transformation.

\begin{proposition}[Dehn function]\label{prop:dehn}
Let $\delta_J$ be the Dehn function of $P_J$. Then $\delta_H\le\delta_J$, and:
\begin{enumerate}[label=(\alph*),itemsep=1pt]
\item there is an absolute constant $K$ with $\delta_H(n)\le e^{Kn}$ for $n\ge1$;
\item there is an absolute constant $C$ with $\delta_H(n)\le Cn^{\,p_H}$ for $n\ge1$, where $p_H=\log107+4\approx10.74$.
\end{enumerate}
\end{proposition}

\begin{proof}
Let $w$ be a word in $p,q$ of length at most $n$ representing $1$. In the free group on $p,q,\alpha$ it is a product of at most $\delta_J(n)$
conjugates of relators of $P_J$ and their inverses. Apply the homomorphism to the free group on $p,q$ that fixes $p,q$ and sends $\alpha$ to
$[p,q]$. It fixes $w$, sends $[p,q]\alpha^{-1}$ to the trivial word, and sends each of the other four relators to a relator of $P_H$. Hence
$\Area_{P_H}(w)\le\delta_J(n)$.

(a) Lee proved that $H_3$ satisfies an exponential isoperimetric inequality~\cite[Theorem~D]{Lee2012}. The Dehn functions $\delta,\delta'$ of two finite
presentations of the same group are equivalent: $\delta(n)\le C'\delta'(C'n)+C'n+C'$ for a constant $C'$~\cite{Bridson2002}. This bound
preserves the class of functions bounded by $e^{Kn}$.

(b) By~\cite[Theorem~3.3]{DouglasMghirbiA}, $\delta_J(n)\le400n^4(1+A^*(4n+2))$, where the far-commutator area satisfies
$A^*(M)\le135M^{\log107}$~\cite[Theorem~3.2]{DouglasMghirbiA}; so $\delta_J(n)=O(n^{\log107+4})$. By the comparison of conventions at the start
of this appendix, this holds for $P_J$ in our letters, and $\delta_H\le\delta_J$.
\end{proof}

Part~(a) rests on Lee's published exponential bound, part~(b) on our polynomial one. Both reach the lower bounds of Part~I through a single
quantity, $\delta_H$ at scale $N$, where $N$ is the number of points that carry the configurations.
\section{The counting criterion and sequential extraction}\label{app:imported}

The lower bounds of Theorems~A and~B stand on two results of the companion papers: the counting criterion
(Theorem~\ref{thm:criterion}, which is~\cite[Theorem~8.1]{DouglasMghirbiA}) and sequential extraction (Theorem~\ref{thm:extraction},
which is~\cite[Theorem~5.2]{DouglasMghirbiB}). This appendix reproduces both proofs in the notation of this paper, so that a reader can check
the main theorems without the companions. The remaining imports, chiefly the Dehn bound for $H_3$ (Proposition~\ref{prop:dehn}), are listed in
Table~\ref{tab:dependencies}. Within this appendix, $\|uv-vu\|_2$ measures how far two unitaries are from commuting, and $\ln$ is the natural
logarithm.

\subsection{Three facts about copies of \texorpdfstring{$S_3$}{S3}}

The first fact rounds an approximate representation of $S_3=\langle s,t\mid s^2,t^2,(st)^3\rangle$ to an exact one.

\begin{lemma}[Exact $S_3$ rounding]\label{lem:imp-s3round}
If $A,B\in U(D)$ satisfy $\|A^2-I\|_2,\|B^2-I\|_2,\|(AB)^3-I\|_2\le e$, there is a unitary representation $u$ of $S_3$ on $\mathbb C^D$
with $\|u(s)-A\|_2\le e$, $\|u(t)-B\|_2\le8e$ and $\|u(st)-AB\|_2\le9e$, and each normal form $1,s,t,st,ts,sts$ evaluated in $u$ is within $10e$
of the same word in $A,B$.
\end{lemma}

\begin{proof}
Round the spectra of $A$ and $B$ to the nearer of $\pm1$. For $|z|=1$ this moves $z$ by at most $|z^2-1|$, so the resulting Hermitian unitaries
$S,T$ satisfy $\|S-A\|_2,\|T-B\|_2\le e$. Then $C=ST$ satisfies $SCS=C^*$ and $\|C^3-I\|_2\le7e$. Round the spectrum of $C$ to the nearest cube root
of unity, choosing $1$ at ties, with one exception: send $-1$ to $1$, although it is not the nearest root there. At that point the
displacement is $2=|(-1)^3-1|$. This rule commutes with complex conjugation and moves $z$ by at most $|z^3-1|$, so $C_0$ satisfies $C_0^3=I$,
$SC_0S=C_0^*$ and $\|C_0-C\|_2\le7e$. Then $u(s)=S$ and $u(t)=SC_0$ are involutions with $u(s)u(t)=C_0$, and the bounds follow from the triangle
inequality.
\end{proof}

The second fact is the entropy that nearly commuting exact copies force. Write $h_2$ for the binary entropy.

\begin{lemma}[Entropic sector bound]\label{lem:imp-sector}
Let $u_1,\dots,u_N:S_3\to U(D)$ be unitary representations, not assumed to commute, such that for a three-cycle $c$
\[
\|u_i(c)-I\|_2\ge\Delta\quad(1\le i\le N),\qquad\|u_i(g)u_j(h)-u_j(h)u_i(g)\|_2\le\eta\quad(i\ne j,\ g,h\in S_3).
\]
Put $\theta=\sqrt{2N}\,\eta$. If $\theta\le\Delta$ and $\theta^2\le5/6$, then
\[
\log D\ \ge\ N\Big[\frac{(\Delta-\theta)^2}3-h_2(\theta^2)-\theta^2\log5\Big].
\]
\end{lemma}

\begin{proof}
Let $\Sigma=S_3^N$ and $f(g)=u_1(g_1)\cdots u_N(g_N)$. For $h$ in factor $i$, $\|f(gh)-f(g)u_i(h)\|_2=\|Y^*XY-X\|_2$ with $X=u_i(h)$ and
$Y=u_{i+1}(g_{i+1})\cdots u_N(g_N)$, and $\|Y^*XY-X\|_2^2=2-2\,\mathrm{Re}\,\tau(X^*Y^*XY)$. The averages
$T_j(Z)=\frac16\sum_au_j(a)^*Zu_j(a)$ are orthogonal projections on $M_D$ with the normalized Hilbert--Schmidt inner product, the average of
$Y^*XY$ over $g$ is $T_N\cdots T_{i+1}(X)$, and $\|(I-T_j)X\|_2^2=\frac1{12}\sum_a\|[X,u_j(a)]\|_2^2\le\eta^2/2$.
Lemma~\ref{lem:projection-product} therefore gives
\[
\mathbb E_g\|f(gh)-f(g)u_i(h)\|_2^2\le2(N-i)\eta^2\le\theta^2 .
\]
Let $\mathsf R$ be the right regular representation of $\Sigma$ and $W:\mathbb C^D\to\ell^2(\Sigma)\otimes\mathbb C^D$ the isometry
$Wv=|\Sigma|^{-1/2}\sum_g|g\rangle\otimes f(g)v$. Then $\|\mathsf R(h)W-Wu_i(h)\|_{HS}\le\sqrt D\,\theta$ for $h$ in factor $i$. Put $\Pi=WW^*$ and
$\rho=\Pi/D$, a state of entropy $\log D$. For the three-cycle $c_i$ of factor $i$ and $Z_i=2I-\mathsf R(c_i)-\mathsf R(c_i)^*$,
\[
\Tr(\rho Z_i)=D^{-1}\|(\mathsf R(c_i)-I)W\|_{HS}^2\ge(\Delta-\theta)^2 .
\]
Let $\mathcal E_i(X)=\frac16\sum_h\mathsf R(h_i)X\mathsf R(h_i)^*$, with $h_i$ the element $h$ placed in factor $i$. These conditional expectations
commute. Since $(I-\Pi)W=0$, $\Tr[(I-\Pi)\mathcal E_i(\rho)]\le\theta^2$. The state $\mathcal E_i(\rho)$ is an average of six states of rank $D$, one
of them $\rho$, so after pinching by $\Pi$ and $I-\Pi$ its two blocks have rank at most $D$ and $5D$. Pinching does not decrease entropy, and
$x\mapsto h_2(x)+x\log5$ increases on $[0,5/6]$, so $S(\mathcal E_i(\rho))\le\log D+h_2(\theta^2)+\theta^2\log5$. For a conditional expectation
$S(\mathcal E(\sigma))-S(\sigma)=D(\sigma\|\mathcal E(\sigma))$. Since the $\mathcal E_j$ commute, monotonicity of relative entropy under
$\mathcal E_{i-1}\cdots\mathcal E_1$ bounds the entropy increment of $\mathcal E_i$, applied after $\mathcal E_{i-1}\cdots\mathcal E_1$, by
$D(\rho\|\mathcal E_i\rho)$. Hence $\bar\rho=\mathcal E_N\cdots\mathcal E_1(\rho)$ satisfies $S(\bar\rho)\le\log D+N[h_2(\theta^2)+\theta^2\log5]$.

The state $\bar\rho$ commutes with $\mathsf R(\Sigma)$, so $\bar\rho=\bigoplus_\lambda p_\lambda(I_{d_\lambda}/d_\lambda\otimes\sigma_\lambda)$ over the
irreducible representations $\lambda$ of $S_3^N$, and $S(\bar\rho)\ge\sum_\lambda p_\lambda\log d_\lambda=\sum_\lambda p_\lambda k_\lambda$, where $k_\lambda$
is the number of factors of $\lambda$ that carry the two-dimensional representation of $S_3$. The operator $Z_i$ is central in the group
algebra, so $\Tr(\bar\rho Z_i)=\Tr(\rho Z_i)$, and it acts as $3I$ on the sectors whose factor $i$ is two-dimensional and as $0$ on the others.
Hence $N(\Delta-\theta)^2\le\sum_i\Tr(\bar\rho Z_i)=3\sum_\lambda p_\lambda k_\lambda\le3S(\bar\rho)$.
\end{proof}

The third fact runs the same count on the three-cycles alone. It needs no relation between the transpositions of different copies. Put
$\chi(x)=h_2(\min\{\frac23x,\frac23\})+\min\{\frac23x,\frac23\}$, a concave function.

\begin{theorem}[Phase-robust sector bound]\label{thm:imp-phase}
Let $c_i,x_i\in U(D)$, $1\le i\le N$, satisfy $c_i^3=x_i^2=I$ and $x_ic_ix_i=c_i^{-1}$, with no condition on how $x_i$ and $x_j$ commute. Put
$\Delta_i=\|c_i-I\|_2$ and $f(g)=c_1^{g_1}\cdots c_N^{g_N}$ for $g\in\Z_3^N$, let $g^{(i)}$ be $g$ with its $i$-th coordinate negated, and, for
uniform $g$, let
\[
\hat a_i^2=\mathbb E_g\|f(g+e_i)-f(g)c_i\|_2^2,\qquad \hat b_i^2=\mathbb E_g\|x_if(g^{(i)})-f(g)x_i\|_2^2 .
\]
Then $\hat a_i^2\le\frac43\sum_{j>i}\|c_ic_j-c_jc_i\|_2^2$, $\hat b_i^2\le\frac43\sum_{j\ne i}\|x_ic_j-c_jx_i\|_2^2$, and
\[
\log D\ \ge\ \sum_{i=1}^N\Big[\frac{(\Delta_i^2-\hat a_i^2)_+}3-\frac{\hat b_i^2}{2\ln2}-\chi(\hat a_i^2)\Big].
\]
\end{theorem}

\begin{proof}
\emph{The ordered-product errors.} For a unitary $h$, $\|(I-\mathcal A_j)h\|_2^2=\frac13\|hc_j-c_jh\|_2^2$ for the average $\mathcal A_j$ of conjugation by
the powers of $c_j$. Moving $c_i$ through the suffix $c_{i+1}^{g_{i+1}}\cdots c_N^{g_N}$ and averaging over its exponents gives
$\hat a_i^2=2-2\,\mathrm{Re}\langle c_i,\mathcal A_N\cdots\mathcal A_{i+1}c_i\rangle$, and Lemma~\ref{lem:projection-product} gives the first bound. For the
second, the maps $\mathcal B_j(Z)=\frac13\sum_kc_j^{-k}Zc_j^{\epsilon_jk}$, with $\epsilon_i=-1$ and $\epsilon_j=1$ for $j\ne i$, are orthogonal
projections. The local relation gives $\mathcal B_i(x_i)=x_i$, and $\|(I-\mathcal B_j)x_i\|_2^2=\frac13\|x_ic_j-c_jx_i\|_2^2$ for $j\ne i$. Moreover
$\hat b_i^2=2-2\,\mathrm{Re}\langle x_i,\mathcal B_N\cdots\mathcal B_1x_i\rangle$, and Lemma~\ref{lem:projection-product} gives the second bound.

\emph{Embedding.} Let $W\xi=3^{-N/2}\sum_gf(g)\xi\otimes|g\rangle$, let $\mathsf C_i$ shift the $i$-th coordinate of $\ell^2(\Z_3^N)$ by $-1$, let
$\mathsf F_i$ negate it, and put $\mathsf X_i=x_i\otimes\mathsf F_i$. The $\mathsf C_i$ commute, $\mathsf X_i\mathsf C_i\mathsf X_i=\mathsf C_i^{-1}$ and
$\mathsf X_i\mathsf C_j\mathsf X_i=\mathsf C_j$ for $j\ne i$. Moreover $D^{-1/2}\|\mathsf C_iW-Wc_i\|_{HS}=\hat a_i$ and
$D^{-1/2}\|\mathsf X_iW-Wx_i\|_{HS}=\hat b_i$.

\emph{Entropy.} Let $\rho=WW^*/D$, and let $\mathcal T_i$ average conjugation by the powers of $\mathsf C_i$. The twirl moves weight at most
$\min\{\frac23\hat a_i^2,\frac23\}$ outside the range of $WW^*$, into a block of rank at most $2D$, so $S(\mathcal T_i\rho)\le\log D+\chi(\hat a_i^2)$. As in
Lemma~\ref{lem:imp-sector} the twirls commute, and the increment of each twirl, the Holevo quantity of the conjugates of the current state, can
only decrease under the earlier twirls by data processing for the Holevo quantity. So $\bar\rho=\mathcal T_1\cdots\mathcal T_N\rho$ has
$S(\bar\rho)\le\log D+\sum_i\chi(\hat a_i^2)$.

\emph{Sectors.} Let $p$ be the joint distribution of the eigenvalue labels $Z_i\in\{1,\omega,\omega^2\}$ of the $\mathsf C_i$, the same for $\rho$
and $\bar\rho$. Then $S(\bar\rho)\ge H(p)\ge\sum_iH(Z_i\mid Z_{-i})$. Let $Q_\alpha$ and $P_\beta$ be the spectral projections of $c_i$ and
$\mathsf C_i$. The numbers $r_{\alpha\beta}=D^{-1}\|P_\beta WQ_\alpha\|_{HS}^2$ form a joint distribution whose marginals are the spectral distribution
of $c_i$ in the normalized trace and that of $\mathsf C_i$ in $\rho$. Distinct cube roots of unity are at squared distance $3$, so
$\hat a_i^2=3\Pr(\alpha\ne\beta)$ and $\Delta_i^2=3\Pr(\alpha\ne1)$, whence $s_i:=\Pr(Z_i\ne1)\ge(\Delta_i^2-\hat a_i^2)_+/3$. Next, the normalized
vectorizations of $W$ and of $\mathsf X_iWx_i^*$ are at squared distance $\hat b_i^2$, so their real overlap is $1-\hat b_i^2/2$. Measuring the joint
spectral projections of the $\mathsf C_j$ gives $p$ from the first and its flip $p^{(i)}$ in coordinate $i$ from the second, because $\mathsf X_i$
exchanges $\omega$ and $\omega^2$ in coordinate $i$ and fixes the other labels. By the Cauchy--Schwarz inequality, outcome by outcome,
$\sum_z\sqrt{p(z)p(z^{(i)})}\ge1-\hat b_i^2/2$. For $0\le x\le1$ one has $1-h_2(x)\le(1-2\sqrt{x(1-x)})/\ln2$: with $y=2x-1$, the difference of
$1-\sqrt{1-y^2}$ and $\ln2\,(1-h_2(x))$ vanishes with its derivative at $y=0$ and has second derivative $(1-y^2)^{-3/2}-(1-y^2)^{-1}\ge0$. So,
for fixed other labels with masses $m,m'$ on the two nontrivial values of $Z_i$, the conditional entropy receives at least
$m+m'-(m+m'-2\sqrt{mm'})/\ln2$. Summing gives $H(Z_i\mid Z_{-i})\ge s_i-\hat b_i^2/(2\ln2)$.
\end{proof}

\subsection{Proof of the counting criterion}

\begin{proof}[Proof of Theorem~\ref{thm:criterion}]
Both parts use the copies $A_i=\phi(w_i)\phi(a)\phi(w_i)^{-1}$ and $B_i=\phi(w_i)\phi(b)\phi(w_i)^{-1}$, $1\le i\le N$, of an assignment $\phi$ extended
to the free group. They have three properties. The $S_3$ relations of each pair are conjugates of $a^2,b^2,(ab)^3$, so they have the defect of a
relator. The product $A_iB_i=\phi(w_i)\phi(ab)\phi(w_i)^{-1}$ has the displacement of $\phi(ab)$. And a commutator of two letters from copies
$i\ne j$, or of their inverses, is a conjugate of $[w_ixw_i^{-1},w_jyw_j^{-1}]^{\pm1}$ with $x,y\in\{a,b\}$, so its area is at most $\Lambda$. A
word of area $A$ is freely a product of $A$ conjugates of relators, so under a unitary assignment it lies within $A$ times the relator defect of
$I$; for permutations the same holds for the normalized Hamming distance, which is bi-invariant.

\emph{Permutations.} Let $f:E\to\Sym_n$ have $f(1)=\mathrm{id}$ and multiply and separate correctly up to $1/r$, with $r\ge10^8N\Lambda(2l-1)$, and put
$\phi(s)=f(s)$ for $s\in S$. For a word $s_1\cdots s_m$ of $R$ with prefixes $p_k$, the chunk contains $p_ks_{k+1}=p_{k+1}$ and $p_m=1$, and
$s\cdot s^{-1}=1$, so $m-1$ merges and at most $m$ inverse replacements give $d(\phi(w),\mathrm{id})\le\delta:=(2l-1)/r$ for every $w\in R$. In the
same way $d(\phi(ab),f(ab))\le1/r$, and $d(f(ab),\mathrm{id})\ge1-1/r$ because $ab\ne1$; so $d(\phi(ab),\mathrm{id})\ge1-2/r\ge\frac12$. Regard
permutations as unitary matrices, with $\|P-I\|_2^2=2d(P,\mathrm{id})$. The $S_3$ relators of each pair then have defect at most $e=\sqrt{2\delta}$,
and Lemma~\ref{lem:imp-s3round} gives exact representations $u_i$ within $10e$ of the normal forms. Commuting two normal forms from different
pairs takes at most nine interchanges of letters, so their commutator has $d\le9\Lambda\delta$ and $\|\cdot\|_2\le3\sqrt{2\Lambda\delta}$, and rounding adds at
most $40e$; so $\eta\le43\sqrt{2\Lambda\delta}$. The element $u_i(st)$ lies within $9e$ of $A_iB_i$, whose displacement in $\|\cdot\|_2$ is at least $1$,
so $\Delta\ge1-9\sqrt{2\delta}\ge\frac12$. Then $\theta^2=2N\eta^2\le7396N\Lambda\delta<10^{-4}$, and Lemma~\ref{lem:imp-sector} gives
$\log n\ge N[(49/100)^2/3-h_2(10^{-4})-10^{-4}\log5]\ge N/16$. So $\log\sigma_E(r)\ge N/16$.

\emph{Unitaries.} Let $\phi$ map every word of $R$ within $e$ of $I$, with $\|\phi(ab)-I\|_2\ge\Delta$ and $K\sqrt N\,\Lambda e\le1$, where
$K=K_\Delta=\sqrt{80000\,(1+\log(1/\Delta))}/\Delta$. Lemma~\ref{lem:imp-s3round} gives exact pairs $x_i=u_i(s)$ and $c_i=u_i(st)$ with
$\|x_i-A_i\|_2\le e$ and $\|c_i-A_iB_i\|_2\le9e$, so $\Delta_i\ge\Delta-9e$. Four and two interchanges of letters give
$\|c_ic_j-c_jc_i\|_2\le4\Lambda e+36e\le40\Lambda e$ and $\|x_ic_j-c_jx_i\|_2\le2\Lambda e+20e\le22\Lambda e$ for $i\ne j$. By Theorem~\ref{thm:imp-phase} and the
concavity of $\chi$,
\[
\log D\ \ge\ N\Big[\frac{(\Delta-9/K)^2-\bar a}{3}-\frac{\bar b}{2\ln2}-\chi(\bar a)\Big],\qquad \bar a\le\frac{3200}{3K^2},\quad \bar b\le\frac{1936}{3K^2},
\]
where $\bar a$ and $\bar b$ are the averages of $\hat a_i^2$ and $\hat b_i^2$, since $e\le1/K$ and $N\Lambda^2e^2\le1/K^2$. It remains to show that the
bracket is at least $\Delta^2/12$. Since $K\ge36/\Delta$, $(\Delta-9/K)^2\ge9\Delta^2/16$, so it suffices that
$\bar a/3+\bar b/(2\ln2)+\chi(\bar a)\le5\Delta^2/48$. Here $\bar a/3+\bar b/(2\ln2)<821.2/K^2$. With $q=2\bar a/3<711.2/K^2$, the bound
$(1-q)\log\frac1{1-q}\le q\log e$ gives $\chi(\bar a)=h_2(q)+q\le q\log(2e/q)$, and $x\log(2e/x)$ increases for $x<2$. Put $L=\log(1/\Delta)$, so that
$K^2=80000(1+L)/\Delta^2$. Multiplying by $K^2$ and using $\log(1+L)\le L+\frac1{10}$ for $L\ge0$, the left side becomes at most
\[
821.2+711.2\Big[\log\frac{2e\cdot80000}{711.2}+\log(1+L)+2L\Big]\ \le\ 7476+2134L\ \le\ \frac5{48}\,80000\,(1+L),
\]
which is the required inequality. Hence $\log D\ge N\Delta^2/12$.
\end{proof}

\subsection{Proof of sequential extraction}

\begin{proof}[Proof of Theorem~\ref{thm:extraction}]
\emph{A test with perfect completeness.} In every round the observer submits half of a fresh maximally entangled pair, measures
$\{\Pi_\Phi,I-\Pi_\Phi\}$ on the output and the reference, keeps the outcome, discards the pair and continues. Let $p_t$ be the probability that
the tests of rounds $1,\dots,t-1$ all pass, $\ell_t$ the probability of failing round $t$ given that, and $F\in\{1,\dots,n,\infty\}$ the first
failure, with $\infty$ for none. The ideal experiment never fails, so $f=\Pr(F\le n)\le\epsilon$. Hence $p_t\ge1-\epsilon$ and
$\sum_tp_t\ell_t=f\le\epsilon$.

\emph{Entropy ledger.} Let $\rho_t$ be the memory state given that the first $t-1$ tests passed; for the bookkeeping, a branch that fails is
frozen. Since the input marginal is $\tau_d$, the round takes $\rho_t$ to $T_t(\rho_t)=\mathrm{Tr}_d\,U_t(\tau_d\otimes\rho_t)U_t^*$, raising the
entropy by $\Delta S_t$. Recording the outcome lowers the average conditional entropy by at most $h_2(\ell_t)$. Exporting $k$ qubits lowers the
entropy by at most $k$, and importing $k$ qubits in a state of entropy $k-c$ raises it by $k-c$, so an exchange lowers it by at most the imported
deficit $c$. The average entropy starts at $Q-P$ and ends at most at $Q$, so
\[
\sum_tp_t\Delta S_t\le P+C+\sum_tp_th_2(\ell_t)\le B+H(F)\le B+h_2(\epsilon)+\epsilon\log n,
\]
where $\sum_tp_th_2(\ell_t)=H(F)$ by the chain rule and $H(F)\le h_2(f)+f\log n$. Therefore
$\sum_tp_t(\ell_t+\Delta S_t)\le B+\epsilon+h_2(\epsilon)+\epsilon\log n$, and since $\sum_tp_t\ge n(1-\epsilon)$, some round has
$\ell_t+\Delta S_t\le t_n$.

\emph{The extracted channel.} Each round's unitary acts only on the input and the $Q$ memory qubits, and the memory state $\rho_t$ conditioned on
the earlier passes does not depend on the current input, since the earlier test registers are discarded. So the round defines a channel $\Psi_t$
with a bath of dimension $2^Q$, and its support leakage is exactly $\ell_t$. For an arbitrary current input with reference, the actual state is
$p_t$ times a normalized pass-block state $\sigma$ plus a failure block, while the ideal state $\omega$ lies in the pass block. Closeness of the
complete experiments gives $\frac12\|p_t\sigma-\omega\|_1+\frac{1-p_t}2\le\epsilon$, hence $p_t\frac12\|\sigma-\omega\|_1\le\epsilon$ and
$\ddia(\Psi_t,\Phi)\le\epsilon/(1-\epsilon)$.
\end{proof}

The device performs no test. The tests and the conditioning on success belong to one particular observer and to the proof.
\section{Proof of the weighted one-round theorem}\label{sec:weighted-proof}

This appendix proves Theorem~\ref{thm:weighted}, the weighted one-round theorem. Theorem~B uses it through
Corollary~\ref{cor:amplification} in Section~\ref{sec:main-channel-proof}, and Appendix~\ref{app:lazy} reuses its proof in a commutator form.

Four steps lead from the round to the entropy bound. A round with small leakage and entropy production yields unitaries on the bath under
which every relation holds up to $O(\sqrt\nu)$, with errors weighted by the memory state (Lemma~\ref{lem:weighted-extraction}). Weighted
errors add up over a van Kampen diagram, at a price paid once for each face and each vertex (Lemma~\ref{lem:charged-area}). Each transported
copy of $S_3$ is rounded to an exact representation on the same space, so the state and its entropy do not change
(Lemma~\ref{lem:rounding}). Measuring the rounded three-cycles one after another on one half of a purification of the memory state
produces about one bit per copy that the other half also knows (Theorem~\ref{thm:sector}).

The argument generalises~\cite[Appendix~B]{DouglasMghirbiB}, which proves the one-round theorem for the Houghton
gadget~\cite[Theorem~5.5]{DouglasMghirbiB}, to every gadget channel that satisfies the hypotheses of Theorem~\ref{thm:weighted}. The weighted
norms are those of~\cite[Appendix~B]{DouglasMghirbiB}. Lemmas~\ref{lem:weighted-extraction} and~\ref{lem:charged-area} extend Lemmas~B.4
and~B.2 there from the Houghton gadget to every such gadget. Lemma~\ref{lem:rounding} is Lemma~B.3 there. Theorem~\ref{thm:sector} is a form
of Theorem~B.1 there, now proved by sequential measurement.

\subsection{Weighted norms}

Throughout, $\|\cdot\|_F$ is the unnormalized Hilbert--Schmidt norm on $M_M$, $\rho$ is a state on $\mathbb C^M$ and $\xi=\sqrt\rho$, so
$\|\xi\|_F=1$. For operators $W$ on $\mathbb C^M$ put
\[
E_\rho(W)=\|(W-I)\xi\|_F,\qquad\chi_\rho(W)=\|W\xi-\xi W\|_F,
\]
the \emph{displacement} and the \emph{centrality defect} of $W$ at $\rho$. We drop the subscript $\rho$. For unitaries $u,v$ and any operator $F$:
\begin{enumerate}[label=(W\arabic*),itemsep=1pt]
\item $E(uv)\le E(u)+E(v)$ and $E(u^*)=E(u)$;
\item $\chi(uv)\le\chi(u)+\chi(v)$ and $\chi(u^*)=\chi(u)$;
\item $\|Fv\xi\|_F\le\|F\xi\|_F+\|F\|\,\chi(v)$;
\item $|E(uWu^*)-E(W)|\le2\chi(u)$ for every $W$ with $\|W\|\le1$;
\item $\|\xi F\|_F=\|F^*\xi\|_F$, and $\|\xi F\|_F=\|F\xi\|_F$ if $F$ is normal.
\end{enumerate}
For (W1) write $(uv-I)\xi=u(v-I)\xi+(u-I)\xi$ and $(u^*-I)\xi=-u^*(u-I)\xi$; for (W2), $uv\xi-\xi uv=u(v\xi-\xi v)+(u\xi-\xi u)v$ and
$u^*\xi-\xi u^*=-u^*(u\xi-\xi u)u^*$. For (W3) write $Fv\xi=F\xi v+F(v\xi-\xi v)$. For (W4), $(uWu^*-I)\xi=u(W-I)u^*\xi$, to which (W3) applies
with $\|W-I\|\le2$, in both directions. For (W5), $\|F^*\xi\|_F^2=\Tr(\xi FF^*\xi)$, which equals $\Tr(\xi F^*F\xi)$ when $F$ is normal. A word
$w$ in $S^{\pm1}$ is evaluated by an assignment $\phi$ of unitaries to $S$ through $\phi(x^{-1})=\phi(x)^*$; by (W2),
$\chi(\phi(w))\le|w|\max_{x\in S}\chi(\phi(x))$.

\subsection{Extraction}

\begin{lemma}[Weighted extraction]\label{lem:weighted-extraction}
Let $G$, $S$ and $R$ be as in Section~\ref{sec:channel-setting}, with $a,b,j\in S$, $j^{-1}ab\in R$ and $j\ne1$ in $G$, and let $m$ be as in
Theorem~\ref{thm:weighted}. Let $\Psi(X)=\mathrm{Tr}_M\,U(X\otimes\rho)U^*$
be a round with $u=\ddia(\Psi,\Phi_G)\le1/8$, leakage $\ell$ relative to $\Phi_G$ and entropy production $\Delta S$, and put $\nu=\ell+\Delta S$,
$\eta=\sqrt{d\nu}$ and $e=256\,m\,\eta$. There are unitaries $\phi(x)$, $x\in S$, on $\mathbb C^M$ such that
\begin{enumerate}[label=(\roman*),itemsep=1pt]
\item $E(\phi(r))\le e$ for every $r\in R$;
\item $\chi(\phi(x))\le e$ for every $x\in S$;
\item $E(\phi(a)\phi(b))\ge1-e$.
\end{enumerate}
\end{lemma}

\begin{proof}
\emph{Normalize the output basis.} Let $(\zeta_i)$ and $(\zeta'_i)$ be orthonormal eigenbases of $\rho$ and $T(\rho)$ with decreasing eigenvalues
$\lambda_i$ and $\mu_i$. Replacing $U$ by $(I\otimes W)U$, where the unitary $W$ maps $\zeta'_i$ to $\zeta_i$, changes neither $\Psi$ nor its
leakage and conjugates $T(\rho)$ by $W$. So we may assume $T(\rho)=\sum_i\mu_i|\zeta_i\rangle\langle\zeta_i|$. Input and output bath are now the
same space, and $\xi=\sum_i\sqrt{\lambda_i}|\zeta_i\rangle\langle\zeta_i|$ acts on both; this does not assert $T(\rho)=\rho$.

\emph{Entropy production controls commutators.} Write $U_{ji}=(I\otimes\langle\zeta_j|)U(I\otimes|\zeta_i\rangle)\in M_d$ and $t_{ji}=\frac1d\|U_{ji}\|_F^2$.
Since the columns and the rows of $U$ are orthonormal, $t$ is doubly stochastic, and the diagonal of $T(\rho)$ gives $\mu_j=\sum_it_{ji}\lambda_i$. The
distributions $p_{ji}=t_{ji}\lambda_i$ and $q_{ji}=t_{ji}\mu_j$ satisfy
\[
D(p\|q)=\sum_{j,i}t_{ji}\lambda_i\log\frac{\lambda_i}{\mu_j}=-S(\rho)+S(T(\rho))=\Delta S .
\]
By Jensen's inequality $-2\ln\sum\sqrt{p_{ji}q_{ji}}\le D(p\|q)\ln2$, and $-2\ln x\ge2(1-x)$, so
$\sum_{j,i}t_{ji}(\sqrt{\lambda_i}-\sqrt{\mu_j})^2\le\Delta S\ln2$. The matrix $t$ is a convex combination of permutation matrices, by Birkhoff's
theorem, and both sequences decrease, so the rearrangement inequality gives
$\sum_{j,i}t_{ji}\sqrt{\lambda_i\mu_j}\le\sum_j\sqrt{\lambda_j\mu_j}$. Expanding the squares and using double stochasticity,
$\sum_j(\sqrt{\lambda_j}-\sqrt{\mu_j})^2\le\sum_{j,i}t_{ji}(\sqrt{\lambda_i}-\sqrt{\mu_j})^2$. The bath block $(j,i)$ of $U(I\otimes\xi)-(I\otimes\xi)U$
is $(\sqrt{\lambda_i}-\sqrt{\lambda_j})U_{ji}$, and $(\sqrt{\lambda_i}-\sqrt{\lambda_j})^2\le2(\sqrt{\lambda_i}-\sqrt{\mu_j})^2+2(\sqrt{\mu_j}-\sqrt{\lambda_j})^2$.
Therefore, with $U=\sum_{o,o'}|o\rangle\langle o'|\otimes A_{oo'}$ in system blocks $A_{oo'}\in M_M$,
\begin{equation}\label{eq:block-commutators}
\sum_{o,o'}\|[A_{oo'},\xi]\|_F^2=\|[U,I\otimes\xi]\|_F^2\le4d\,\Delta S\ln2=:\epsilon_S^2 .
\end{equation}

\emph{Leakage controls residuals.} The Choi state $\mathcal J(\Psi)$ is the marginal on $\mathbb C^d\otimes\mathbb C^d$ of the unit vector
$d^{-1/2}\sum_{o,o'}|oo'\rangle\otimes A_{oo'}\xi$ in $\mathbb C^d\otimes\mathbb C^d\otimes\mathcal H$, where $\mathcal H$ is $M_M$ with the Hilbert--Schmidt
inner product. Indeed $\langle A_{o_2o_2'}\xi,A_{o_1o_1'}\xi\rangle=\Tr(A_{o_1o_1'}\rho A_{o_2o_2'}^*)=\langle o_1|\Psi(|o_1'\rangle\langle o_2'|)|o_2\rangle$. The support
of $\mathcal J(\Phi_G)$ is spanned by the orthogonal vectors $|A_g\rangle\!\rangle$, of squared norms $\mu_g=\|A_g\|_F^2$. Projecting onto it, put
\[
X_g=\mu_g^{-1}\sum_{[w_a]=g}\overline{c_a}A_a,\qquad R_{oo'}=A_{oo'}-\sum_g(A_g)_{oo'}X_g,
\]
where $A_a$ is the block at slot $a$. Then $\sum_{o,o'}\|R_{oo'}\xi\|_F^2=d\ell=:\epsilon_\ell^2$. At a slot $A_a=c_aX_{[w_a]}+R_a$, and at every other position
$A_{oo'}=R_{oo'}$. We have $\epsilon_\ell\le\eta$, $\epsilon_S\le2\sqrt{\ln2}\,\eta$, and, by Cauchy--Schwarz, $\epsilon_\ell+\epsilon_S\le\sqrt{1+4\ln2}\,\eta<2\eta$ and
$2\epsilon_\ell^2+2\epsilon_S^2\le8\ln2\,\eta^2$.

\emph{Anchors are nearly unitary.} Every symbol $w$, and the empty word, occupies a diagonal position $z=z(w)$ whose row and column contain
no other slot; there the coefficient is $1$, so $B_w:=A_{zz}=X_{[w]}+R_{zz}$. Column unitarity gives
$\Tr\rho(I-B_w^*B_w)=\sum_{o\ne z}\|R_{oz}\xi\|_F^2\le\epsilon_\ell^2$. Row unitarity gives
$\Tr\rho(I-B_wB_w^*)=\sum_{o'\ne z}\|\xi A_{zo'}\|_F^2\le\sum_{o'\ne z}(\|R_{zo'}\xi\|_F+\|[A_{zo'},\xi]\|_F)^2\le2\epsilon_\ell^2+2\epsilon_S^2$.
Let $B_w=V_w|B_w|=|B_w^*|V_w$ with $V_w$ unitary. Since $(1-x)^2\le1-x^2$ on $[0,1]$,
\[
\|(V_w-B_w)\xi\|_F\le\epsilon_\ell,\qquad\|\xi(V_w-B_w)\|_F=\|\xi(I-|B_w^*|)\|_F\le\sqrt{8\ln2}\,\eta,
\]
and so $\chi(V_w)\le\|(V_w-B_w)\xi\|_F+\|[B_w,\xi]\|_F+\|\xi(B_w-V_w)\|_F\le\epsilon_\ell+\epsilon_S+\sqrt{8\ln2}\,\eta<4.3\eta$, since
$\epsilon_\ell+\epsilon_S\le\sqrt{1+4\ln2}\,\eta$.

\emph{Triangles enforce the products.} In the block of a triangle $(x,y,z)$, compare each slot $a$ with the anchor of the symbol $w'(a)$ that
carries the same label. These symbols are the empty word at the upper left, $y$ at the upper right, $x$ at the lower left and $z$ at the lower
right. Put
$F_a=A_a-c_aV_{w'(a)}$. Then $F_a=R_a-c_aR_{z'z'}+c_a(B_{w'(a)}-V_{w'(a)})$ with $z'=z(w'(a))$, and since $a$ and $(z',z')$ are different
positions, Cauchy--Schwarz gives
\begin{gather*}
\|F_a\xi\|_F\le\sqrt{3/2}\,\epsilon_\ell+2^{-1/2}\epsilon_\ell<2\eta,\qquad\|F_a\|<2,\\
\|\xi F_a\|_F\le\|F_a\xi\|_F+\|[A_a,\xi]\|_F+2^{-1/2}\chi(V_{w'(a)})<7\eta .
\end{gather*}
For two slots $a_1,a_2$ in the same row of the block, with $A_{a_i}=c_iV_i+F_i$,
\[
A_{a_1}^*A_{a_2}-\overline{c_1}c_2V_1^*V_2=A_{a_1}^*F_2+c_2F_1^*V_2,
\]
whose right-weighted norm is at most $\|F_2\xi\|_F+2^{-1/2}(\|\xi F_1\|_F+\|F_1\|\chi(V_2))<2\eta+2^{-1/2}\cdot17\eta<15\eta$ by (W3) and (W5). The two
columns of the block are orthogonal columns of $U$, so the products for the two rows of the block add up to $-C_1^*C_2$, where $C_1,C_2$ are the
parts of the two columns in the other rows. These are residuals, so $\|C_1^*C_2\xi\|_F\le\|C_2\xi\|_F\le\epsilon_\ell$. The coefficients are
$\overline{c_1}c_2=\frac12$ for the pair (upper left, upper right) and $-\frac12$ for the pair (lower left, lower right), so
\[
\tfrac12\|(V_\varnothing^*V_y-V_x^*V_z)\xi\|_F<15\eta+15\eta+\eta .
\]
Put $W_w=V_\varnothing^*V_w$ for every symbol $w$, and $W_\varnothing=I$. Since $W_xW_y-W_z=V_\varnothing^*V_x(V_\varnothing^*V_y-V_x^*V_z)$,
\begin{equation}\label{eq:triangle}
\|(W_xW_y-W_z)\xi\|_F\le64\eta\quad\text{for every triangle }(x,y,z),\qquad\chi(W_w)<10\eta .
\end{equation}

\emph{The relations nearly hold.} Let $\phi(x)=W_x$ for $x\in S$. For a letter $x\in S$, the triangle $(x,x^{-1},\varnothing)$ gives
$\|(W_{x^{-1}}-\phi(x^{-1}))\xi\|_F=\|W_x^*(W_xW_{x^{-1}}-I)\xi\|_F\le64\eta$. Fix a word $r=\sigma_1\cdots\sigma_\ell$ of $R$ and let
$D_l=\|(\phi(\sigma_1\cdots\sigma_l)-W_{p_l})\xi\|_F$, with $p_l$ its prefix symbols, $p_1=\sigma_1$ and $W_{p_\ell}=I$. Then $D_1\le64\eta$, and
\[
\phi(\sigma_1\cdots\sigma_l)-W_{p_l}=\bigl(\phi(\sigma_1\cdots\sigma_{l-1})-W_{p_{l-1}}\bigr)\phi(\sigma_l)+W_{p_{l-1}}\bigl(\phi(\sigma_l)-W_{\sigma_l}\bigr)+\bigl(W_{p_{l-1}}W_{\sigma_l}-W_{p_l}\bigr).
\]
By (W3), \eqref{eq:triangle} and the triangle $(p_{l-1},\sigma_l,p_l)$, $D_l\le D_{l-1}+2\cdot10\eta+64\eta+64\eta$. So
$E(\phi(r))=D_\ell\le148\,\ell\,\eta\le e$, which is (i); and (ii) holds because $\chi(\phi(x))<10\eta\le e$.

\emph{The copy of $ab$ is displaced.} Let $z=z(j)$ and let $0$ be the position of the empty word. The labels there are $[j]\ne1$ and $1$, so
$\langle z|\Phi_G(|z\rangle\langle0|)|0\rangle=\tau(\lambda([j]))=0$, while $\langle z|\Psi(|z\rangle\langle0|)|0\rangle=\Tr(B_j\rho B_\varnothing^*)$. Feed
$(|zz\rangle+|00\rangle)/\sqrt2$ into $\Psi\otimes\mathrm{id}$ and $\Phi_G\otimes\mathrm{id}$. The difference $\Delta$ of the outputs is Hermitian and
traceless with $\|\Delta\|_1\le2u$, so each of its entries has modulus at most $\|\Delta\|\le\|\Delta\|_1/2\le u$. The entry at $(zz,00)$ is half the
difference of $\langle z|\Psi(|z\rangle\langle0|)|0\rangle$ and $\langle z|\Phi_G(|z\rangle\langle0|)|0\rangle$, so $|\Tr(\rho B_\varnothing^*B_j)|\le2u$. By the
anchor bounds, $|\Tr(\rho V_\varnothing^*V_j)-\Tr(\rho B_\varnothing^*B_j)|\le2\epsilon_\ell$, and so
\[
E(\phi(j))^2=\|(V_j-V_\varnothing)\xi\|_F^2=2-2\,\mathrm{Re}\,\Tr(\rho V_\varnothing^*V_j)\ge2-4u-4\epsilon_\ell\ge\tfrac32-4\eta .
\]
By (i) for the word $j^{-1}ab$, $\|(\phi(a)\phi(b)-\phi(j))\xi\|_F=E(\phi(j)^*\phi(a)\phi(b))\le e$, so $E(\phi(a)\phi(b))\ge\sqrt{3/2-4\eta}-e\ge1-e$ if
$\eta\le1/8$. If $\eta>1/8$ then $e>1$ and (iii) is trivial.
\end{proof}

\subsection{Charged area}

Relator errors add up over a van Kampen diagram, but in the weighted norms every conjugation costs twice the centrality defect of the
conjugator. Writing a null word naively as a product of conjugates of relations would charge each conjugating word once for every relation it
encloses. Peeling the diagram one boundary face at a time charges instead each face a bounded amount and each vertex once.

\begin{lemma}[Charged area]\label{lem:charged-area}
Let $S$ and $R$ be finite, with every word of $R$ of length at most $m$, and let $\phi$ assign unitaries to $S$ with $E(\phi(r))\le e$ for $r\in R$
and $\chi(\phi(x))\le e$ for $x\in S$. Every word $w$ of finite area satisfies
\[
E(\phi(w))\ \le\ \bigl(3m\Area_R(w)+|w|\bigr)\,e\ \le\ \bigl((4m+1)\Area_R(w)+2|w|\bigr)\,e .
\]
\end{lemma}

\begin{proof}
Since $\phi$ is a homomorphism on the free group, $E(\phi(w))$ depends only on the free reduction of $w$, which is not longer and has the same
area; so let $w$ be freely reduced. By van Kampen's lemma~\cite[Chapter~V]{LyndonSchupp1977}, $w$ has a van Kampen diagram with at most
$\Area_R(w)$ faces. Such a diagram is a finite, connected and simply connected planar $2$-complex with edges labelled by letters of $S$. Its
boundary cycle, read from a base vertex, is $w$, and the boundary cycle of every face, read from a suitable vertex in a suitable direction, is
a word of $R$. We show, by induction on the number of edges, that every such diagram $D$ with $V$ vertices and $A$ faces satisfies
\begin{equation}\label{eq:charge}
E(\phi(w_D))\le\bigl((2m-1)A+2(V-1)\bigr)e ,
\end{equation}
where $w_D$ is its boundary word. If $D$ is a single vertex, $w_D$ is empty. Otherwise let $\varepsilon$ be the first edge of the boundary cycle,
from the base vertex $v_0$ to $v_1$, with label $x^{\pm1}$.

If $\varepsilon$ lies on no face, the boundary cycle traverses it twice, and removing it leaves diagrams $D_0\ni v_0$ and $D_1\ni v_1$ with
$w_D=x^{\pm1}w_{D_1}x^{\mp1}w_{D_0}$, the words read from $v_1$ and $v_0$. By (W1) and (W4), $E(\phi(w_D))\le2e+E(\phi(w_{D_1}))+E(\phi(w_{D_0}))$, and
the vertices and faces of $D_0$ and $D_1$ add up to those of $D$.

Otherwise $\varepsilon$ lies on the boundary of a face $f$. Let $c$ be the boundary cycle of $f$ read from $v_0$ and beginning with $\varepsilon$,
so that $c=\varepsilon\pi$ with $\pi$ a path in $\partial f$ from $v_1$ to $v_0$. Removing $f$ and $\varepsilon$ merges $f$ with the outer
region and leaves a van Kampen diagram $D'$ with the same vertices and one face fewer, whose boundary word read from $v_0$ is
$w_{D'}=\pi^{-1}\varrho$, where $w_D=\varepsilon\varrho$. Hence
$w_D=c\,w_{D'}$ freely. The word $c$ is a cyclic permutation of some $r^{\pm1}$, $r\in R$, that is $c=y^{-1}r^{\pm1}y$ freely with $|y|\le m-1$, so
$E(\phi(c))\le e+2(m-1)e$ by (W4) and (W2), and $E(\phi(w_D))\le(2m-1)e+E(\phi(w_{D'}))$ by (W1).

In both cases~\eqref{eq:charge} follows from the induction hypothesis. To finish, $V-1$ is at most the number of edges of $D$, since $D$ is
connected. Counting the two sides of every edge gives twice the number of edges as at most $mA+|w|$, the sum of the lengths of the face
boundaries and of the boundary cycle. So~\eqref{eq:charge} gives $E(\phi(w))\le((2m-1)A+mA+|w|)e$.
\end{proof}

\subsection{Exact \texorpdfstring{$S_3$}{S3} rounding in place}

An approximate copy of $S_3$ is rounded to an exact one by functional calculus on the same space. The state is untouched, so no entropy is
gained or lost. The lemma is~\cite[Lemma~B.3]{DouglasMghirbiB}.

\begin{lemma}[Exact rounding]\label{lem:rounding}
Let $X$ and $Y$ be unitaries on $\mathbb C^M$, and suppose that the displacements of $X^2$, $Y^2$ and $(XY)^3$ and the centrality defects of
$X$ and $Y$ are at most $\delta$. There are unitaries $x,c$ on $\mathbb C^M$ with $x^2=c^3=I$ and $xcx=c^{-1}$, such that
\[
\|(x-X)\xi\|_F\le\delta,\qquad\chi(x)\le3\delta,\qquad\|(c-XY)\xi\|_F\le29\delta,\qquad\chi(c)\le56\delta .
\]
\end{lemma}

\begin{proof}
Let $\mathrm{sgn}\,t=1$ for $t\ge0$ and $-1$ for $t<0$, and let $x=\mathrm{sgn}(\mathrm{Re}\,X)$ and $y=\mathrm{sgn}(\mathrm{Re}\,Y)$ by functional calculus;
they are self-adjoint unitaries. On the unit circle $|z-\mathrm{sgn}\,\mathrm{Re}\,z|\le|z^2-1|$: for $z=e^{i\vartheta}$ with $|\vartheta|\le\pi/2$ this reads
$2|\sin(\vartheta/2)|\le2|\sin\vartheta|$, which holds as $\cos(\vartheta/2)\ge\frac12$, and the case $\mathrm{Re}\,z<0$ follows by symmetry. Hence
$|x-X|^2\le|X^2-I|^2$ as functions of the normal operator $X$, and $\|(x-X)\xi\|_F\le E(X^2)\le\delta$; by (W5) also $\|\xi(x-X)\|_F\le\delta$, so
$\chi(x)\le\chi(X)+2\delta\le3\delta$. The same holds for $y$.

Put $Z=XY$ and $w=xy$, so that $xwx=yx=w^*$. Then $\chi(Z)\le2\delta$, $\chi(w)\le6\delta$, and, by (W3),
$\|(w-Z)\xi\|_F\le\|x(y-Y)\xi\|_F+\|(x-X)Y\xi\|_F\le\delta+3\delta=4\delta$. From $w^3-Z^3=w^2(w-Z)+w(w-Z)Z+(w-Z)Z^2$ and (W3),
\[
E(w^3)\le E(Z^3)+4\delta+(4\delta+2\cdot2\delta)+(4\delta+2\cdot4\delta)\le25\delta .
\]
Let $f$ send each point of the unit circle to a nearest cube root of unity, with the ties $e^{\pm i\pi/3}$ and $-1$ sent to $1$. Then
$f(\bar z)=\overline{f(z)}$, and $|z-f(z)|\le|z^3-1|$: for $z=-1$ both sides equal $2$, and otherwise $z=f(z)e^{i\vartheta}$ with
$|\vartheta|\le\pi/3$, and the inequality reads $2|\sin(\vartheta/2)|\le2|\sin(3\vartheta/2)|$. Let $c=f(w)$. Then $c^3=I$, and $xcx=f(xwx)=f(w^*)=c^*=c^{-1}$. As before
$\|(c-w)\xi\|_F\le E(w^3)\le25\delta$ and $\|\xi(c-w)\|_F\le25\delta$. Hence $\|(c-Z)\xi\|_F\le29\delta$ and
$\chi(c)\le\chi(w)+\|(c-w)\xi\|_F+\|\xi(c-w)\|_F\le56\delta$.
\end{proof}

\subsection{A weighted sector bound}

Rounded copies of $S_3$ that nearly commute in the weighted sense force entropy. We work in the Hilbert space $\mathcal H$ of operators on
$\mathbb C^M$ with the Hilbert--Schmidt inner product $\langle Y,Z\rangle=\Tr(Y^*Z)$, identified with $\mathbb C^M\otimes\mathbb C^M$. Left multiplications
act on the first factor, which we call the system, and right multiplications on the second, the reference; they commute. The unit vector $\xi$
purifies $\rho$ on the system, and its reference marginal has the same nonzero spectrum, so it has entropy $S(\rho)$. The first lemma
is~\cite[Lemma~4.1]{DouglasMghirbiA}, with its short proof.

\begin{lemma}[Projection products]\label{lem:projection-product}
Let $P_1,\dots,P_k$ be orthogonal projections on a Hilbert space. Every vector $y$ satisfies
$\|y\|^2-\mathrm{Re}\langle y,P_k\cdots P_1y\rangle\le2\sum_l\|(I-P_l)y\|^2$.
\end{lemma}

\begin{proof}
By homogeneity let $\|y\|=1$. Put $y_0=y$, $y_l=P_ly_{l-1}$ and $s_l=y_{l-1}-y_l$, which lies in the range of $I-P_l$. Then
$\sum_l\|s_l\|^2=1-\|y_k\|^2=:\beta$, and $\alpha:=1-\mathrm{Re}\langle y,y_k\rangle=\mathrm{Re}\sum_l\langle(I-P_l)y,s_l\rangle\le\sqrt{E\beta}$ with
$E=\sum_l\|(I-P_l)y\|^2$. Since $0\le\|y-y_k\|^2=2\alpha-\beta$, we get $\alpha\le\sqrt{2E\alpha}$, that is $\alpha\le2E$.
\end{proof}

Let $c_1,\dots,c_N$ be unitaries on $\mathbb C^M$ with $c_i^3=I$, and write $c_i=\sum_{k\in\Z/3}\omega^kP_i^k$ with $\omega=e^{2\pi i/3}$ and spectral
projections $P_i^k$. For every $Y\in\mathcal H$,
\begin{equation}\label{eq:spectral-identities}
\|(c_i-I)Y\|_F^2=3\sum_{k\ne0}\|P_i^kY\|_F^2,\qquad\|[c_i,Y]\|_F^2=3\sum_{k\ne l}\|P_i^kYP_i^l\|_F^2,
\end{equation}
because $|\omega^k-\omega^l|^2=3$ for $k\ne l$ and the pieces $P_i^kYP_i^l$ are orthogonal. For $\mathbf a=(a_1,\dots,a_{i-1})\in(\Z/3)^{i-1}$ put
$Q_{\mathbf a}=P_{i-1}^{a_{i-1}}\cdots P_1^{a_1}$. Summing over the coordinates one at a time gives
$\sum_{\mathbf a}Q_{\mathbf a}^*Q_{\mathbf a}=\sum_{\mathbf a}Q_{\mathbf a}Q_{\mathbf a}^*=I$.

\begin{lemma}[Transfer to the reference]\label{lem:transfer}
For every $Y\in\mathcal H$ and $1\le i\le N$,
\[
\sum_{\mathbf a}\|Q_{\mathbf a}Y-YQ_{\mathbf a}\|_F^2\ \le\ \frac43\sum_{k<i}\|[c_k,Y]\|_F^2 .
\]
\end{lemma}

\begin{proof}
The pinchings $\mathcal E_k(Y)=\sum_lP_k^lYP_k^l$ are orthogonal projections on $\mathcal H$, and by~\eqref{eq:spectral-identities}
$\|(I-\mathcal E_k)Y\|_F^2=\sum_{l\ne l'}\|P_k^lYP_k^{l'}\|_F^2=\frac13\|[c_k,Y]\|_F^2$. Both $\sum_{\mathbf a}\|Q_{\mathbf a}Y\|_F^2$ and
$\sum_{\mathbf a}\|YQ_{\mathbf a}\|_F^2$ equal $\|Y\|_F^2$, and summing over the coordinates one at a time,
$\sum_{\mathbf a}\langle YQ_{\mathbf a},Q_{\mathbf a}Y\rangle=\sum_{\mathbf a}\Tr(Q_{\mathbf a}^*Y^*Q_{\mathbf a}Y)=\langle\mathcal E_1\cdots\mathcal E_{i-1}Y,Y\rangle$.
So the left side equals $2\|Y\|_F^2-2\,\mathrm{Re}\langle Y,\mathcal E_1\cdots\mathcal E_{i-1}Y\rangle$, and Lemma~\ref{lem:projection-product} bounds it by
$4\sum_{k<i}\|(I-\mathcal E_k)Y\|_F^2$.
\end{proof}

A three-valued distribution that is nearly symmetric under $k\mapsto-k$ carries about one bit for each unit of weight off $0$. Let $h_2$ be the
binary entropy. With $t=(1+y)/2$,
\[
\ln2\,\bigl(1-h_2(t)\bigr)=\sum_{k\ge1}\frac{y^{2k}}{2k(2k-1)}\ \le\ \sum_{k\ge1}\binom{2k}k\frac{y^{2k}}{(2k-1)4^k}=1-\sqrt{1-y^2}=1-2\sqrt{t(1-t)},
\]
by comparing coefficients: $4^k\le2k\binom{2k}k$ holds for $k=1$, and the right side grows by the factor $2(2k+1)/k\ge4$ from $k$ to $k+1$. For a
distribution $(r_0,r_1,r_2)$ with $s=r_1+r_2>0$ this gives, with $t=r_1/s$,
\begin{equation}\label{eq:binary-hellinger}
H(r_0,r_1,r_2)\ \ge\ s\,h_2(r_1/s)\ \ge\ s-\frac{(\sqrt{r_1}-\sqrt{r_2})^2}{\ln2},
\end{equation}
which also holds when $s=0$.

\begin{theorem}[Weighted sector bound]\label{thm:sector}
Let $c_i,x_i$, $1\le i\le N$, be unitaries on $\mathbb C^M$ with $c_i^3=x_i^2=I$ and $x_ic_ix_i=c_i^{-1}$, with no condition relating different
indices. For a state $\rho$ with $\xi=\sqrt\rho$ put $\Delta_i=E(c_i)$, $\theta_i=\chi(c_i)$, $\theta'_i=\chi(x_i)$, and
\[
T_i=\Bigl(\sum_{k<i}\theta_k^2\Bigr)^{1/2},\qquad K_i=\Bigl(\sum_{k<i}\|[c_i,c_k]\xi\|_F^2\Bigr)^{1/2},\qquad K'_i=\Bigl(\sum_{k<i}\|[x_i,c_k]\xi\|_F^2\Bigr)^{1/2},
\]
\[
w_i=\frac13\Bigl(\Delta_i-\frac4{\sqrt3}T_i\Bigr)_+^2,\qquad\beta_i=\theta'_i+\frac2{\sqrt3}(2T_i+K'_i),\qquad
\pi_i=\frac13\Bigl(\theta_i+\frac2{\sqrt3}(2T_i+K_i)\Bigr)^2 .
\]
If $\pi_i\le2/3$ for all $i$, then
\[
S(\rho)\ \ge\ \sum_{i=1}^N\Bigl[w_i-\frac{\beta_i^2}{2\ln2}-h_2(\pi_i)-\pi_i\Bigr].
\]
\end{theorem}

\begin{proof}
\emph{Measure the copies in turn.} Measure $c_1,\dots,c_N$ in turn on the system, with the projections $P_i^k$, and let $A_i\in\Z/3$ be the outcomes.
Then $\Pr(A_{<i}=\mathbf a,A_i=k)=\|P_i^kQ_{\mathbf a}\xi\|_F^2=:p(\mathbf a,k)$. The measurements do not touch the reference, whose entropy is
$S(\rho)$, so
\[
S(\rho)\ \ge\ I(A_1\cdots A_N:\mathrm{Ref})=\sum_{i=1}^NI(A_i:\mathrm{Ref}\mid A_{<i}),
\]
since the conditional entropy of the reference given the classical outcomes is nonnegative. Fix $i$.

\emph{The copy stays displaced.} By~\eqref{eq:spectral-identities}, $\Pr(A_i\ne0)=\frac13\sum_{\mathbf a}\|(c_i-I)Q_{\mathbf a}\xi\|_F^2$. By
Lemma~\ref{lem:transfer} with $Y=\xi$, the vectors $Q_{\mathbf a}\xi$ and $\xi Q_{\mathbf a}$ differ by at most $\frac2{\sqrt3}T_i$ in $\ell^2$ over $\mathbf a$,
and $\sum_{\mathbf a}\|(c_i-I)\xi Q_{\mathbf a}\|_F^2=\Delta_i^2$ because $\sum_{\mathbf a}Q_{\mathbf a}Q_{\mathbf a}^*=I$. Since $\|c_i-I\|\le2$, the triangle
inequality in $\ell^2$ gives $\Pr(A_i\ne0)\ge w_i$.

\emph{The sign is a fair coin.} Since $x_iP_i^kx_i=P_i^{-k}$ and right multiplication by $x_i$ is unitary,
$\sqrt{p(\mathbf a,-k)}=\|P_i^kx_iQ_{\mathbf a}\xi\|_F$ and $\sqrt{p(\mathbf a,k)}=\|P_i^kQ_{\mathbf a}\xi x_i\|_F$, so
$\sum_{\mathbf a,k}(\sqrt{p(\mathbf a,k)}-\sqrt{p(\mathbf a,-k)})^2\le\sum_{\mathbf a}\|x_iQ_{\mathbf a}\xi-Q_{\mathbf a}\xi x_i\|_F^2$. Write
\[
x_iQ_{\mathbf a}\xi-Q_{\mathbf a}\xi x_i=x_i(Q_{\mathbf a}\xi-\xi Q_{\mathbf a})+\bigl((x_i\xi)Q_{\mathbf a}-Q_{\mathbf a}(x_i\xi)\bigr)+Q_{\mathbf a}(x_i\xi-\xi x_i).
\]
In $\ell^2$ over $\mathbf a$ the first term is at most $\frac2{\sqrt3}T_i$ and the third is $\theta'_i$. For the second apply Lemma~\ref{lem:transfer} with
$Y=x_i\xi$: since $[c_k,x_i\xi]=[c_k,x_i]\xi+x_i[c_k,\xi]$, it is at most $\frac2{\sqrt3}(K'_i+T_i)$. So the sum is at most $\beta_i^2$, and since the
terms $k=1$ and $k=2$ contribute equally, \eqref{eq:binary-hellinger} gives
\[
H(A_i\mid A_{<i})\ \ge\ \Pr(A_i\ne0)-\frac1{\ln2}\sum_{\mathbf a}\bigl(\sqrt{p(\mathbf a,1)}-\sqrt{p(\mathbf a,2)}\bigr)^2\ \ge\ w_i-\frac{\beta_i^2}{2\ln2}.
\]

\emph{The reference knows the outcome.} Let the reference measure $c_i$ from the right, with the projections $Y\mapsto YP_i^l$, and call the outcome
$B_i$. Left and right multiplications commute, so $\Pr(A_{<i}=\mathbf a,A_i=k,B_i=l)=\|P_i^kQ_{\mathbf a}\xi P_i^l\|_F^2$, and
by~\eqref{eq:spectral-identities} $\Pr(A_i\ne B_i)=\frac13\sum_{\mathbf a}\|[c_i,Q_{\mathbf a}\xi]\|_F^2$. The same three-term splitting gives
\[
[c_i,Q_{\mathbf a}\xi]=c_i(Q_{\mathbf a}\xi-\xi Q_{\mathbf a})+\bigl((c_i\xi)Q_{\mathbf a}-Q_{\mathbf a}(c_i\xi)\bigr)+Q_{\mathbf a}(c_i\xi-\xi c_i),
\]
and Lemma~\ref{lem:transfer} with $Y=\xi$ and $Y=c_i\xi$ gives $\Pr(A_i\ne B_i)\le\pi_i$. By data processing and Fano's inequality for three values,
\[
I(A_i:\mathrm{Ref}\mid A_{<i})\ \ge\ I(A_i:B_i\mid A_{<i})\ \ge\ H(A_i\mid A_{<i})-h_2(\pi_i)-\pi_i ,
\]
where we used that $h_2(\pi)+\pi$ increases on $[0,2/3]$. Summing over $i$ proves the theorem.
\end{proof}

\begin{corollary}\label{cor:sector}
In Theorem~\ref{thm:sector}, suppose $\Delta_i\ge1/2$ for all $i$, and that $\theta_i$, $\theta'_i$, $\|[c_i,c_k]\xi\|_F$ and $\|[x_i,c_k]\xi\|_F$ are at
most $b$ for all $i\ne k$, where $\sqrt N\,b\le10^{-3}$. Then $S(\rho)\ge N/16$.
\end{corollary}

\begin{proof}
Now $T_i,K_i,K'_i\le\sqrt N\,b\le10^{-3}$ and $\theta_i,\theta'_i\le10^{-3}$. So $w_i\ge\frac13(\frac12-\frac4{\sqrt3}10^{-3})^2>0.0825$, and $\beta_i$ and
$\sqrt{3\pi_i}$ are at most $10^{-3}+2\sqrt3\cdot10^{-3}<4.5\cdot10^{-3}$. Hence $\beta_i^2/(2\ln2)<1.5\cdot10^{-5}$ and $\pi_i<6.8\cdot10^{-6}$, and with
$h_2(\pi)\le\pi\log(e/\pi)$ we get $h_2(\pi_i)<1.3\cdot10^{-4}$. Each summand exceeds $0.0823>1/16$.
\end{proof}

\subsection{Completion}

\begin{proof}[Proof of Theorem~\ref{thm:weighted}]
Let $\phi$ be given by Lemma~\ref{lem:weighted-extraction}, with $e=256\,m\sqrt{d\nu}$, and put
\[
b_*=100\,m\,(A+L+1)\,e=25600\,m^2(A+L+1)\sqrt{d\nu},\qquad\delta_L=(2L+1)e .
\]
Condition~\eqref{eq:weighted-condition} gives $\sqrt N\,b_*\le25600\cdot2^{-25}<10^{-3}$.

\emph{Round the copies.} Let $\psi_i=\phi(u_i)$, so $\chi(\psi_i)\le Le$, and put $X_i=\psi_i\phi(a)\psi_i^*$, $Y_i=\psi_i\phi(b)\psi_i^*$ and $Z_i=X_iY_i$. Since
$a^2$, $b^2$ and $(ab)^3$ lie in $R$, (W4) and Lemma~\ref{lem:weighted-extraction} give $E(X_i^2),E(Y_i^2),E(Z_i^3)\le e+2Le=\delta_L$, and (W2) gives
$\chi(X_i),\chi(Y_i)\le\delta_L$. Lemma~\ref{lem:rounding} gives unitaries $x_i,c_i$ with $x_i^2=c_i^3=I$, $x_ic_ix_i=c_i^{-1}$ and
\[
\|(x_i-X_i)\xi\|_F\le\delta_L,\qquad\chi(x_i)\le3\delta_L,\qquad\|(c_i-Z_i)\xi\|_F\le29\delta_L,\qquad\chi(c_i)\le56\delta_L .
\]
By (W4) and Lemma~\ref{lem:weighted-extraction}(iii), $E(Z_i)\ge E(\phi(a)\phi(b))-2Le\ge1-\delta_L$, so $\Delta_i=E(c_i)\ge1-30\delta_L$.

\emph{The raw copies nearly commute.} Let $i\ne k$ and $x,y\in\{a,b\}$, and let $P=\psi_i\phi(x)\psi_i^*$ and $P'=\psi_k\phi(y)\psi_k^*$. The group commutator
$g=[u_ixu_i^{-1},u_kyu_k^{-1}]$ is a word of length at most $8L+4$ and area at most $A$, so Lemma~\ref{lem:charged-area} gives
$E(\phi(g))\le(3mA+8L+4)e$. Since $PP'-P'P=(\phi(g)-I)P'P$ and $\chi(P'P)\le2\delta_L$, (W3) gives
\[
\|(PP'-P'P)\xi\|_F\le\gamma:=(3mA+8L+4)e+4\delta_L=(3mA+16L+8)e .
\]
Expanding $[X_i,X_kY_k]=[X_i,X_k]Y_k+X_k[X_i,Y_k]$ and $[X_iY_i,Z_k]=X_i[Y_i,Z_k]+[X_i,Z_k]Y_i$ and using (W3),
$\|[X_i,Z_k]\xi\|_F,\|[Y_i,Z_k]\xi\|_F\le2\gamma+2\delta_L$ and $\|[Z_i,Z_k]\xi\|_F\le4\gamma+6\delta_L$.

\emph{The rounded copies nearly commute too.} For unitaries $U,V,U_0,V_0$,
\[
[U,V]-[U_0,V_0]=U(V-V_0)+(U-U_0)V_0-V(U-U_0)-(V-V_0)U_0,
\]
so by (W3) the weighted norms of the two commutators differ by at most
$2\|(U-U_0)\xi\|_F+2\|(V-V_0)\xi\|_F+2\chi(U_0)+2\chi(V_0)$. With $(U,U_0,V,V_0)=(x_i,X_i,c_k,Z_k)$ and $(c_i,Z_i,c_k,Z_k)$ this gives
\begin{gather*}
\|[x_i,c_k]\xi\|_F\le2\gamma+68\delta_L=(6mA+168L+84)e,\\
\|[c_i,c_k]\xi\|_F\le4\gamma+130\delta_L=(12mA+324L+162)e .
\end{gather*}
Since $m\ge6$, these and $\chi(c_i)\le(112L+56)e$ and $\chi(x_i)\le(6L+3)e$ are at most $b_*$, and $\Delta_i\ge1-(60L+30)e\ge1-b_*>1/2$.
Corollary~\ref{cor:sector} gives $S(\rho)\ge N/16$.
\end{proof}

The constant $2^{25}$ is not optimized; the proof uses only $25600\cdot2^{-25}<10^{-3}$, and Corollary~\ref{cor:sector} has room to spare.

\begin{remark}[The constant term]\label{rem:plus-one}
The term $+1$ in~\eqref{eq:weighted-condition} cannot be dropped. Take $N=1$, the empty transport word, so that $L=0$, and $A=0$, which is
allowed since there are no pairs. The gadget channel has a Stinespring dilation with a pure environment of dimension $r_*$, so there is a round
with $M=r_*$, $\rho$ pure and $\Psi=\Phi_G$ exactly. Then $\ddia(\Psi,\Phi_G)=0$ and the condition with $A+L$ in place of $A+L+1$ holds, since its left
side vanishes, but $S(\rho)=0<1/16$. With the $+1$, the condition forces $\nu\le2^{-50}m^{-4}d^{-1}$, while this round has $\ell=0$ and
$\Delta S=S(T(\rho))$, the entropy of the complementary output at the maximally mixed input. That output is diagonal with entries
$\|A_g\|_F^2/d\ge1/d$, and there are at least two labels, since $ab\ne1$; so $\Delta S\ge h_2(1/d)$, far above the threshold.
\end{remark}
\section{Proof of the streaming theorem}\label{app:streaming}

This appendix proves Theorem~\ref{thm:streaming}, the streaming service at the rank rate. Theorem~B uses it through
Corollary~\ref{cor:streaming-lamps}, which gives the services at the rank rate in Section~\ref{sec:main-channel-proof}.

The device keeps one rectangular dilation $W$ for all rounds and lets the part of the memory in use, the \emph{active space}, grow by the
factor $k'/k$ in every round; that growth is the whole purity bill. Three things are then controlled. The Kraus histories must stay nearly
orthogonal, and a Haar-random unitary of the active space before each round ensures this, through a concentration bound on the Gram matrix of
the histories. The concentration bound needs the state of everything outside the memory to stay spread out, and exporting about
$\frac12\log r$ qubits per round ensures that. And the error follows from the Gram bound alone, because the final vector of the service differs
from the ideal one by the polar part of a frame. The proof fixes the schedule, chooses the unitaries round by round, and then bounds the error
and the memory.

\begin{proof}[Proof of Theorem~\ref{thm:streaming}]
We use the frames of~\cite[Appendix~A]{DouglasMghirbiB}. A \emph{frame} is a linear map $F:\mathbb C^M\to\mathbb C^D\otimes\mathbb C^S$; its columns
are the vectors, in the memory $\mathbb C^D$ and an inaccessible exterior $\mathbb C^S$, attached to the Kraus histories. The exterior holds the
purifier of the initial memory, the exported qubits and the purifiers of the imported ones. Put $g(F)=\|F^*F\|$ and
$\alpha(F)=\|\mathrm{Tr}_{\mathbb C^D}FF^*\|$. Two lemmas from there carry the argument.
\begin{itemize}[leftmargin=1.2em,itemsep=1pt]
\item \emph{Concentration}~\cite[Lemma~A.1]{DouglasMghirbiB}. For Hermitian $X$ on $\mathbb C^D$ and Haar-random $V\in U(D)$, the matrix
$Z=F^*((V^*XV)\otimes I)F-\tau_D(X)F^*F$ satisfies $\Pr\{\|Z\|\ge s\}\le2M\exp[-Ds^2/(4\alpha(F)g(F)\|X\|^2)]$.
\item \emph{Export}~\cite[Lemma~A.4]{DouglasMghirbiB}. If $\mathbb C^D=\mathbb C^{D/e}\otimes\mathbb C^e$, there is a unitary of $\mathbb C^D$ after
which exporting the factor $\mathbb C^e$ and importing a maximally mixed $e$-dimensional register give a frame $F'$ with $F'^*F'=F^*F$ and, for every
$\eta>0$, $\alpha(F')\le(1+\eta)\alpha(F)/e^2+4e^2g(F)\ln(8Se^2)/(\eta D)$.
\end{itemize}

\emph{Fix the schedule and the dimensions.} Put $h=\log r$, $E_p=2^{\lceil h/2\rceil}$, $J_t=\lfloor th/2\rfloor$, $j_t=J_t-J_{t-1}\le\lceil h/2\rceil$ and
$e_t=2^{j_t}$, a divisor of $E_p$. Choose $D_0$, a multiple of $kE_p$, with $D_0\ge nkE_p$ and
\begin{equation}\label{eq:D0}
D_0^2\ \ge\ \frac{32\,A_*\,r^4n^2\ln(8r^{n+2})}{\epsilon^2},\qquad\text{where}\quad A_*=12\bigl(1+8E_p^2n^2L_*\bigr),\quad L_*=\ln(8D_0r^nE_p^2).
\end{equation}
Put $q_0=k'/k$. For $t\ge1$ let $\widetilde D_t=q_0D_{t-1}$, a multiple of $k'E_p$, and let $D_t$ be the least multiple of $kE_p$ with
$D_t\ge\widetilde D_t$. Then $D_t\le q_0D_{t-1}+kE_p$, so $D_n\le q_0^n(D_0+nkE_p)\le2q_0^nD_0$. The memory has $Q=\lceil\log D_n\rceil$ qubits, of
which the last $\lceil h/2\rceil$ form a \emph{port}. A space whose dimension $D'$ is a multiple of $E_p$ is realized as a subspace of the first
$Q-\lceil h/2\rceil$ qubits of dimension $D'/E_p$, tensored with the port. The active space of round $t$ is such a space of dimension $D_{t-1}$,
identified with $\mathbb C^k\otimes\mathbb C^{D_{t-1}/k}$.

\emph{The device never leaves the active spaces.} The memory starts maximally mixed on the active space of round $1$, of dimension $D_0$. In
round $t$ the device applies a unitary $V_t$ of the active space and then $W\otimes I_{D_{t-1}/k}$, which maps the input and the active space
isometrically onto the output and a space of dimension $\widetilde D_t$. It releases the output. It then applies a unitary $V'_t$ of that space;
since $V'_t$ acts on the memory alone, it could equally act before the release. If $t<n$, the device exports $j_t$ port qubits and imports
$j_t$ qubits in the maximally mixed state in their place. The resulting space of dimension $\widetilde D_t$ lies in the active space of round
$t+1$. These maps are completed arbitrarily to unitaries of the input and the whole memory; the completion never acts, since the memory never
leaves the active spaces. So $C=0$, $J=J_{n-1}$, and $P=Q-\log D_0\le n\log q_0+2$.

\emph{Frames track the histories.} Purify the initial memory: its frame $F_0$ is one unit vector, with $\alpha(F_0)=1/D_0$. Let $F$ be the frame
before round $t$, with one column $F_y$ for each history $y\in[r]^{t-1}$. After the call the frame has the columns $(B_i\otimes I)(V_t\otimes I)F_y$,
indexed by the histories $(y,i)$. Its Gram matrix has the blocks $F^*\bigl((V_t^*(B_i^*B_j\otimes I)V_t)\otimes I\bigr)F$, whose Haar means are
$\tau(B_i^*B_j\otimes I)F^*F=\delta_{ij}F^*F$ by the first identity in~\eqref{eq:rectangular}. By the second identity and cyclicity of the partial
trace over the memory, the exterior marginal after the call is exactly $r$ times the one before, whatever $V_t$ is; so the call multiplies
$\alpha$ by $r$.

\emph{Choose the unitaries round by round.} We choose $V_1,V'_1,V_2,\dots$ in turn so that the frame $F_t$ after round $t$ satisfies
\begin{equation}\label{eq:streaming-induction}
\|F_t^*F_t-I\|\le t\epsilon/n,\qquad \alpha(F_t)\le A_*/D_0 .
\end{equation}
This holds for $t=0$. Suppose it holds for $t-1$; then $g(F_{t-1})\le2$. Since $B_i^*B_i\le rI$, the Hermitian and anti-Hermitian parts of the $B_i^*B_j$
have norm at most $r$. Apply the concentration lemma to these $2r^2$ operators, tensored with $I_{D_{t-1}/k}$, with $D=D_{t-1}\ge D_0$, at most $r^n$
columns and level $s=\epsilon/(2rn)$: by~\eqref{eq:D0} the failure probabilities add up to at most
$4r^{n+2}\exp[-D_0^2\epsilon^2/(32A_*r^4n^2)]\le\frac12$. So some $V_t$ makes every block of the new Gram matrix deviate from its mean by at most
$\epsilon/(rn)$, hence the $r\times r$ block matrix by at most $\epsilon/n$, and the Gram matrix after the call is within $t\epsilon/n$ of $I$. Exports and
imports preserve it. Choose $V'_t$ by the export lemma with $e=e_t$, $\eta=1/n$ and $D=\widetilde D_t\ge D_0$. The exterior has dimension
$S\le D_0\,2^{2J_{t-1}}\le D_0r^n$, so $\ln(8Se_t^2)\le L_*$, and
\[
\alpha(F_t)\ \le\ \Bigl(1+\frac1n\Bigr)\frac{r\,\alpha(F_{t-1})}{e_t^2}+\frac{8E_p^2nL_*}{D_0}.
\]
Put $c_t=2^{2J_t}r^{-t}$, so that $c_0=1$, $\frac14<c_t\le1$ and $c_tr/e_t^2=c_{t-1}$. Multiplying by $c_t$ and unrolling gives
$c_t\alpha(F_t)\le(1+\frac1n)^t(1+8E_p^2n^2L_*)/D_0$, hence $\alpha(F_t)\le4e(1+8E_p^2n^2L_*)/D_0\le A_*/D_0$. In round $n$ there is no export, and only the
Gram bound is needed.

\emph{The error is at most $\epsilon$.} Purify the observer. In the ideal experiment, with the Stinespring isometry
$\xi\mapsto\sum_iA_i\xi\otimes|i\rangle$ in each round, its final vector is $\sum_y\psi_y\otimes|y\rangle$. In the service it is
$\sum_y\psi_y\otimes F_n|y\rangle$ with the same vectors $\psi_y$, because the device acts on the input and the memory through
$W=\sum_iA_i\otimes B_i$ and unitaries of the memory, and the observer never touches the memory or the exterior. The polar part
$T=F_n(F_n^*F_n)^{-1/2}$ is an isometry with $\|F_n-T\|=\|(F_n^*F_n)^{1/2}-I\|\le\epsilon$. So the actual final vector is within $\epsilon$ of the
ideal one mapped by $I\otimes T$, and after tracing out the memory and the exterior the observer's final states are within $\epsilon$ in trace
distance.

\emph{The memory is as claimed.} We have $Q\le\log D_n+1\le\log D_0+n\log q_0+2$. The conditions on $D_0$ hold for
$D_0=kE_p\lceil K_r(n/\epsilon)^4(1+\ln k)\rceil$ with $K_r$ large enough in terms of $r$, and then
$\log D_0=\log k+O_r(\log(n/\epsilon)+\log(2+\log k))$.
\end{proof}
\section{The lazy doubling in detail}\label{app:lazy}

This appendix proves parts (b) and (c) of Proposition~\ref{prop:lazy}, used in Section~\ref{sec:physics-lazy}, and
Proposition~\ref{prop:compositions}, used in Appendix~\ref{sec:physics-finer}; it also records two remarks on $\Theta_+$. Part (b) needs the
one-round theorem in a commutator form (Lemma~\ref{lem:commutator-form}), a weighted intertwining lemma (Lemma~\ref{lem:twisted}) and a block
extraction with a smooth cutoff (Lemma~\ref{lem:block-extraction}).

Throughout, $\Phi$ is a general unital channel on $M_d$, with doubling
$\Theta(\begin{smallmatrix}x&y\\u&w\end{smallmatrix})=\mathrm{diag}(\Phi(w),\Phi^*(x))$ and $\Theta_+=\frac12(\mathrm{id}+\Theta)$ on $M_{2d}$. Let
$\mathcal S_\Phi\subseteq M_d$ be the span of the Kraus operators of $\Phi$, so that the Choi support of $\Phi$ is
$\{|A\rangle\!\rangle:A\in\mathcal S_\Phi\}$, and let $\Pi_\Phi$ also denote the Hilbert--Schmidt projection of $M_d$ onto $\mathcal S_\Phi$. Since
$\Theta(X)=\sum_g(e_{01}\otimes A_g)X(e_{01}\otimes A_g)^*+(e_{10}\otimes A_g^*)X(e_{10}\otimes A_g^*)^*$ for Kraus operators $A_g$ of $\Phi$, the Choi
support of $\Theta_+$ corresponds to the span of the three pairwise orthogonal subspaces
\begin{equation}\label{eq:lazy-support}
\mathbb C\,1_{2d},\qquad e_{01}\otimes\mathcal S_\Phi,\qquad e_{10}\otimes\mathcal S_\Phi^*
\end{equation}
of $M_{2d}=M_2\otimes M_d$. As in Appendix~\ref{sec:weighted-proof}, $\|\cdot\|_F$ is the Hilbert--Schmidt norm and $\|\cdot\|$ the operator
norm. An operator $X$ on the bath $\mathbb C^M$ is identified with $1\otimes X$ when it acts on $\mathbb C^d\otimes\mathbb C^M$, so that
$\|1\otimes X\|_F=\sqrt d\,\|X\|_F$.

\begin{lemma}[Commutator form of the one-round theorem]\label{lem:commutator-form}
Lemma~\ref{lem:weighted-extraction}, Theorem~\ref{thm:weighted} and Corollary~\ref{cor:amplification} remain true if, in the definition
$\nu=\ell+\Delta S$, the entropy production $\Delta S$ of the round $\Psi(X)=\mathrm{Tr}_M\,U(X\otimes\rho)U^*$ is replaced by any $\tilde s\ge0$ with
$\|[U,1\otimes\sqrt\rho\,]\|_F^2\le4d\,\tilde s\ln2$, and if the hypothesis $\ddia(\Psi,\Phi_G)\le1/8$ is replaced by the weaker hypothesis
\begin{equation}\label{eq:coherence}
|\langle z|\Psi(|z\rangle\langle0|)|0\rangle|\le\tfrac34,
\end{equation}
where $0$ and $z$ are the anchor positions of the empty word and of $j$; in Lemma~\ref{lem:weighted-extraction}(iii) the bound then reads
$E(\phi(a)\phi(b))\ge0.66-e$. Conversely, every round can be rewritten, without changing its channel or its leakage, as a round
$(U',\rho)$ with $U'=(1\otimes W)U$ for which $\tilde s=\Delta S$ is admissible.
\end{lemma}

\begin{proof}
The rewriting is the normalization step of the proof of Lemma~\ref{lem:weighted-extraction}, after which~\eqref{eq:block-commutators} is
the inequality $\|[U',1\otimes\sqrt\rho\,]\|_F^2\le4d\,\Delta S\ln2$; its proof uses only that $U'$ is a unitary on $\mathbb C^d\otimes\mathbb C^M$. In
the proof of Lemma~\ref{lem:weighted-extraction}, $\Delta S$ enters only through~\eqref{eq:block-commutators}. Afterwards only the bounds
$\epsilon_S\le2\sqrt{\ln2}\,\eta$, $\epsilon_\ell+\epsilon_S\le\sqrt{1+4\ln2}\,\eta$ and $2\epsilon_\ell^2+2\epsilon_S^2\le8\ln2\,\eta^2$ with $\eta=\sqrt{d\nu}$ are used,
together with the unitarity of $U$, its leakage and $\ddia(\Psi,\Phi_G)$. These bounds hold with $\epsilon_S^2\le4d\tilde s\ln2$ and $\nu=\ell+\tilde s$.
The diamond hypothesis enters only in the displacement step, through
$|\Tr(\rho B_\varnothing^*B_j)|=|\langle z|\Psi(|z\rangle\langle0|)|0\rangle|\le2u\le\frac14$. Under~\eqref{eq:coherence} instead,
$E(\phi(j))^2\ge2-\frac32-4\epsilon_\ell\ge\frac12-4\eta$, so $E(\phi(a)\phi(b))\ge\sqrt{1/2-4\eta}-e\ge0.66-e$ if $\eta\le\frac1{64}$, while $e>1$ otherwise. In
the proof of Theorem~\ref{thm:weighted} this gives $E(Z_i)\ge0.66-\delta_L$ and $\Delta_i\ge0.66-30\delta_L\ge0.66-b_*>\frac12$ in place of
$\Delta_i\ge1-b_*$, and $\Delta_i\ge\frac12$ is all that Corollary~\ref{cor:sector} needs. Theorem~\ref{thm:weighted} and Corollary~\ref{cor:amplification} use
the round only through Lemma~\ref{lem:weighted-extraction}.
\end{proof}

\begin{lemma}[Weighted intertwining]\label{lem:twisted}
Let $B$ be an operator on $\mathbb C^d\otimes\mathbb C^M$ with $\|B\|\le1$, let $R,R'$ be self-adjoint operators on $\mathbb C^M$, $\xi$ an operator on $\mathbb C^M$, and
$\Sigma=BR-R'B$.
\begin{enumerate}[label=(\roman*),itemsep=1pt]
\item If $h$ is Lipschitz with constant $L_h$ on an interval containing the spectrum of $R$, then $\|[h(R),\xi]\|_F\le L_h\,\|[R,\xi]\|_F$.
\item If $g(x)=\int\hat g(t)e^{itx}dt$ with $M_1=\int|t\,\hat g(t)|\,dt$ and $M_2=\int t^2|\hat g(t)|\,dt$ finite, then
\[
\bigl\|\bigl(Bg(R)-g(R')B\bigr)\xi\bigr\|_F\ \le\ M_1\,\|\Sigma\xi\|_F+\tfrac12M_2\,\|\Sigma\|\,\sqrt d\,\|[R,\xi]\|_F .
\]
\end{enumerate}
\end{lemma}

\begin{proof}
(i) In an eigenbasis of $R$ with eigenvalues $r_k$, the entries of $[h(R),\xi]$ and $[R,\xi]$ are $(h(r_k)-h(r_l))\xi_{kl}$ and $(r_k-r_l)\xi_{kl}$.
(ii) Differentiating $\theta\mapsto e^{i\theta tR'}Be^{i(1-\theta)tR}$ gives
$Be^{itR}-e^{itR'}B=it\int_0^1e^{i\theta tR'}\,\Sigma\,e^{i(1-\theta)tR}\,d\theta$. Write $e^{i(1-\theta)tR}\xi=\xi e^{i(1-\theta)tR}+[e^{i(1-\theta)tR},\xi]$ and use
(i) with $h(x)=e^{i(1-\theta)tx}$; unitaries do not change Hilbert--Schmidt norms, so
$\|(Be^{itR}-e^{itR'}B)\xi\|_F\le|t|\,\|\Sigma\xi\|_F+\frac{t^2}2\|\Sigma\|\sqrt d\,\|[R,\xi]\|_F$. Integrate against $|\hat g(t)|\,dt$.
\end{proof}

We fix the cutoff once and for all. Let $S:\mathbb R\to[0,1]$ be smooth with $S=0$ on $(-\infty,0]$ and $S=1$ on $[1,\infty)$, let $0<s_0\le\frac14$, and put
\[
g(x)=\sin\bigl(\tfrac\pi2S(2x/s_0-1)\bigr)\,S(3-2x),\qquad g_2(x)=\cos\bigl(\tfrac\pi2S(2x/s_0-1)\bigr),\qquad f(x)=\sqrt x\,g(x).
\]
Then $g$ is smooth with compact support, so $M_g:=\int(|t|+t^2)|\hat g(t)|\,dt<\infty$; on $[0,1]$ we have $0\le g\le1$, $g=0$ on $[0,s_0/2]$,
$g=1$ and $g_2=0$ on $[s_0,1]$, and $g^2+g_2^2=1$. Let $L_g$ be a common Lipschitz constant of $g$ and $g_2$ on $[0,1]$, put $\sigma=\sqrt{2/s_0}$ and
\[
K_a=11M_g+1+1.5L_g+2\sigma(4+1.5L_g).
\]
These depend only on $s_0$.

\begin{lemma}[Block extraction]\label{lem:block-extraction}
Let $\Phi$ be a unital channel on $M_d$, and let $\Psi(X)=\mathrm{Tr}_M\,U(X\otimes\rho)U^*$ be a round of $\Theta_+$, with $U$ unitary on
$\mathbb C^2\otimes\mathbb C^d\otimes\mathbb C^M$, $u=\ddia(\Psi,\Theta_+)$ and leakage $\ell$ relative to $\Theta_+$. Put
$\epsilon_\ell=\sqrt{2d\ell}$, $\epsilon_S=\|[U,1\otimes\sqrt\rho\,]\|_F$, $\gamma=\epsilon_\ell+\epsilon_S$ and $\mathcal E=2u+s_0+4K_a\sqrt d\,\gamma$. If $\mathcal E<\frac12$,
there are a unitary $U'$ on $\mathbb C^d\otimes\mathbb C^M$ and a state $\rho'$ on $\mathbb C^M$ such that the round
$\Psi'(X)=\mathrm{Tr}_M\,U'(X\otimes\rho')U'^*$ of $\Phi$, with leakage $\ell_\Phi(\Psi')$ relative to $\Phi$, satisfies
\[
\ddia(\Psi',\Phi)\le\frac{\mathcal E}{\frac12-\mathcal E},\qquad \ell_\Phi(\Psi')\le\frac{4K_a^2\gamma^2}{d(\frac12-\mathcal E)},\qquad
\bigl\|[U',1\otimes\sqrt{\rho'}\,]\bigr\|_F^2\le\frac{8K_a^2\gamma^2}{\frac12-\mathcal E},
\]
and $|\langle x|\Psi'(|x\rangle\langle y|)|y\rangle|\le \mathcal E/(\frac12-\mathcal E)$ for all unit vectors $x,y$ with $\langle x|\Phi(|x\rangle\langle y|)|y\rangle=0$.
\end{lemma}

\begin{proof}
Put $\xi=\sqrt\rho$ and write $U=\sum_{a,b\in\{0,1\}}e_{ab}\otimes U_{ab}$ with $D_0=U_{00}$, $B=U_{01}$, $D_1=U_{11}$, operators on $\mathbb C^d\otimes\mathbb C^M$.

\emph{Leakage bounds the residuals.} As in the proof of Lemma~\ref{lem:weighted-extraction}, $\mathcal J(\Psi)$ is the marginal of the unit vector
$(2d)^{-1/2}U\xi\in M_{2d}\otimes M_M$, so $\ell=(2d)^{-1}\|(\Pi^\perp\otimes\mathrm{id})(U\xi)\|_F^2$, with $\Pi$ the projection of $M_{2d}$ onto the span
of~\eqref{eq:lazy-support}. Put $\alpha=(2d)^{-1}(\Tr\otimes\mathrm{id})(U)$, the average of the $2d$ diagonal entries of $U$ in $M_M$, and
$E_a=D_a-\alpha$, $F=(\Pi_\Phi^\perp\otimes\mathrm{id})(B)$. By~\eqref{eq:lazy-support}, $(\Pi^\perp\otimes\mathrm{id})(U)$ has the diagonal blocks $E_0,E_1$ and
the upper right block $F$; since $\Pi^\perp\otimes\mathrm{id}$ commutes with right multiplication by $\xi$, the numbers $e_a=\|E_a\xi\|_F$ and
$e_F=\|F\xi\|_F$ are at most $\epsilon_\ell$. Also $\|\alpha\|\le1$ and $\|E_a\|\le2$.

\emph{The blocks nearly commute with $\xi$.} Put $k_X=\|[X,\xi]\|_F$ for $X\in\{D_0,B,D_1\}$; these are at most $\epsilon_S$. The diagonal entries of
$U$ are entries of $D_0$ and $D_1$, so by Cauchy--Schwarz
$k_\alpha:=\sqrt d\,\|[\alpha,\xi]\|_F\le(k_{D_0}^2+k_{D_1}^2)^{1/2}/\sqrt2\le\epsilon_S/\sqrt2$. With $R=1-\alpha^*\alpha$ and
$R'=1-\alpha\alpha^*$, both between $0$ and $1$, $k_R:=\sqrt d\,\|[R,\xi]\|_F\le2k_\alpha\le\sqrt2\,\epsilon_S$, and likewise $k_{R'}\le\sqrt2\,\epsilon_S$.

\emph{The unitarity defects are small.} Unitarity gives $B^*B=1-D_1^*D_1$ and $BB^*=1-D_0D_0^*$, hence
\begin{gather*}
\Delta:=R-B^*B=\alpha^*E_1+E_1^*D_1,\qquad\Delta':=R'-BB^*=\alpha E_0^*+E_0D_0^*,\\
\Sigma:=BR-R'B=B\Delta-\Delta'B,
\end{gather*}
with $\|\Delta\|,\|\Delta'\|\le1$ and $\|\Sigma\|\le2$. From $\xi E_1=E_1\xi-[D_1,\xi]+[\alpha,\xi]$ and $D_1\xi=\xi D_1+[D_1,\xi]$,
$\|\Delta\xi\|_F\le e_1+\|\xi E_1\|_F+2k_{D_1}\le2e_1+3k_{D_1}+k_\alpha\le4\gamma$. Symmetrically $\|\Delta'\xi\|_F\le2e_0+3k_{D_0}+k_\alpha\le4\gamma$,
and $\Sigma\xi=B\Delta\xi-\Delta'\xi B-\Delta'[B,\xi]$ gives $\|\Sigma\xi\|_F\le9\gamma$.

\emph{$Z_0$ factors in two ways.} Let $P$ and $P'$ be the spectral projections of $R$ and $R'$ for $(s_0/2,\infty)$. Since $g$ vanishes on $[0,s_0/2]$,
$g(R)=R^{-1/2}Pf(R)$ with $\|R^{-1/2}P\|\le\sigma$, and likewise for $R'$. The operators $B_1=BR^{-1/2}P$ and $B_1'=R'^{-1/2}P'B$ satisfy
\[
B_1^*B_1=P-\tilde\Delta,\quad \tilde\Delta=PR^{-1/2}\Delta R^{-1/2}P,\qquad B_1'B_1'^*=P'-\tilde\Delta',\quad\tilde\Delta'=P'R'^{-1/2}\Delta'R'^{-1/2}P'.
\]
Take polar decompositions $B_1=V_1|B_1|$ and $B_1'=|B_1'^*|V_1'$ with $V_1,V_1'$ unitary. On the range of $1\otimes P$ we have
$|B_1|=1+\varphi(\tilde\Delta)$ with $\varphi(x)=\sqrt{1-x}-1$ and $\tilde\Delta\le1$, and $\varphi(x)^2\le x^2$; so $\|(|B_1|-P)Y\|_F\le\|\tilde\Delta Y\|_F$
for every $Y$ with range in that subspace, and similarly for $B_1'$. Put $Z_0=Bg(R)\xi=B_1f(R)\xi$. Then
$Z_0-V_1f(R)\xi=V_1(|B_1|-P)f(R)\xi$, and $\tilde\Delta f(R)\xi=PR^{-1/2}\Delta g(R)\xi$, while
$\|\Delta g(R)\xi\|_F\le\|\Delta\xi\|_F+\sqrt d\,\|[g(R),\xi]\|_F\le\|\Delta\xi\|_F+L_g\,k_R$ by Lemma~\ref{lem:twisted}(i). Hence
\[
a_1:=\|Z_0-V_1f(R)\xi\|_F\le\sigma(4+1.5L_g)\gamma .
\]
On the other side,
$Z_0=(Bg(R)-g(R')B)\xi+g(R')[B,\xi]+[g(R'),\xi]B+\xi f(R')B_1'$ and $\xi f(R')B_1'=\xi f(R')V_1'+\xi f(R')(|B_1'^*|-P')V_1'$, where
$\|\xi f(R')(|B_1'^*|-P')\|_F\le\|\tilde\Delta'f(R')\xi\|_F\le\sigma(\|\Delta'\xi\|_F+L_g\,k_{R'})$. By Lemma~\ref{lem:twisted}(ii),
$\|(Bg(R)-g(R')B)\xi\|_F\le M_g(\|\Sigma\xi\|_F+k_R)\le11M_g\gamma$. Hence
\[
a_2:=\|Z_0-\xi f(R')V_1'\|_F\le11M_g\gamma+(1+1.5L_g)\gamma+\sigma(4+1.5L_g)\gamma,\qquad a_1+a_2\le K_a\gamma .
\]

\emph{The polar part of $Z_0$ nearly intertwines $\zeta_1$ and $\zeta_2$.} Put $\zeta_1=|f(R)\xi|=(\xi f(R)^2\xi)^{1/2}$ and
$\zeta_2=|(\xi f(R'))^*|=(\xi f(R')^2\xi)^{1/2}$. Then $|V_1f(R)\xi|=\zeta_1$ and $|(\xi f(R')V_1')^*|=\zeta_2$, as operators $1\otimes\zeta_i$. By the
Araki--Yamagami inequality $\||X|-|Y|\|_F\le\sqrt2\,\|X-Y\|_F$~\cite{ArakiYamagami1981}, applied to $Z_0$ and to $Z_0^*$,
$\||Z_0|-\zeta_1\|_F\le\sqrt2\,a_1$ and $\||Z_0^*|-\zeta_2\|_F\le\sqrt2\,a_2$. Let $Z_0=O|Z_0|=|Z_0^*|O$ with $O$ unitary. Then
\[
\|O\zeta_1-\zeta_2O\|_F\le\|O(\zeta_1-|Z_0|)\|_F+\|(|Z_0^*|-\zeta_2)O\|_F\le\sqrt2(a_1+a_2)=:\delta .
\]

\emph{Build the new round.} The Hermitian operators $O(1\otimes\zeta_1)O^*$ and $1\otimes\zeta_2$ are within $\delta$ in $\|\cdot\|_F$, and their
spectra are those of $\zeta_1$ and $\zeta_2$, each repeated $d$ times. By the Hoffman--Wielandt inequality~\cite{HoffmanWielandt1953}, the decreasing
eigenvalue lists $(x_k)$ of $\zeta_1$ and $(y_k)$ of $\zeta_2$ satisfy $d\sum_k(x_k-y_k)^2\le\delta^2$; so the unitary $W_h$ of $\mathbb C^M$ that maps
an eigenbasis of $\zeta_1$ to one of $\zeta_2$ in the same order has $\|1\otimes(W_h\zeta_1W_h^*-\zeta_2)\|_F\le\delta$. Put $U'=(1\otimes W_h^*)O$,
$c_1=\Tr\zeta_1^2=\Tr\rho f(R)^2$ and $\rho'=\zeta_1^2/c_1$. Then
$\|[U',1\otimes\zeta_1]\|_F\le\|O\zeta_1-\zeta_2O\|_F+\|(\zeta_2-W_h\zeta_1W_h^*)O\|_F\le2\delta$, which is the third bound once we know
$c_1\ge\frac12-\mathcal E$. The channel of $(U',\rho')$ is $\Psi'=\Phi_{O\zeta_1}/c_1$, where $\Phi_X(w)=\mathrm{Tr}_M[X(w\otimes1)X^*]$. Its leakage is
$(dc_1)^{-1}\|(\Pi_\Phi^\perp\otimes\mathrm{id})(O\zeta_1)\|_F^2$, and
\[
\|(\Pi_\Phi^\perp\otimes\mathrm{id})(O\zeta_1)\|_F\le\|(\Pi_\Phi^\perp\otimes\mathrm{id})(Z_0)\|_F+\sqrt2\,a_1\le e_F+L_g\,k_R+\sqrt2\,a_1\le2K_a\gamma,
\]
since $(\Pi_\Phi^\perp\otimes\mathrm{id})(Bg(R)\xi)=F\xi\,g(R)+(\Pi_\Phi^\perp\otimes\mathrm{id})(B[g(R),\xi])$. This is the second bound.

\emph{The new round is close to $\Phi$.} For the completely positive maps $\Phi_X$ just defined, $\|\Phi_X\|_\diamond=\|\mathrm{Tr}_M(X^*X)\|$, and
$\|\Phi_X-\Phi_Y\|_\diamond\le\|X-Y\|_F\bigl(\|\mathrm{Tr}_M X^*X\|^{1/2}+\|\mathrm{Tr}_M Y^*Y\|^{1/2}\bigr)$, where $\mathrm{Tr}_M$ traces out the input bath.
Let $Y=B\xi$. The upper left block of $\Psi(e_{11}\otimes w)$ is $\Phi_Y(w)$ and that of $\Theta_+(e_{11}\otimes w)$ is $\frac12\Phi(w)$; compressing to
the block, $\|\Phi_Y-\frac12\Phi\|_\diamond\le2u$. Since $g^2+g_2^2=1$ on $[0,1]$, $\Phi_Y-\Phi_{Z_0}-\Phi_{Bg_2(R)\xi}$ is
$w\mapsto\mathrm{Tr}_M[B(w\otimes X)B^*]$ with $X=\sum_{h\in\{g,g_2\}}h(R)[h(R),\rho]$, and $\|X\|_1\le4L_g\|[R,\xi]\|_F$, so its diamond norm is at most
$4L_g\,k_R/\sqrt d$, by splitting $X$ into positive and negative parts and using $B^*B\le1$. Next,
$\|\Phi_{Bg_2(R)\xi}\|_\diamond\le\Tr\rho\,g_2(R)Rg_2(R)+\|\xi g_2(R)\Delta g_2(R)\xi\|_1\le s_0+\sqrt d(\|\Delta\xi\|_F+L_g\,k_R)$, because $xg_2(x)^2\le s_0$ on
$[0,1]$. And $\mathrm{Tr}_M\zeta_1^2=c_1\le1$ and $\mathrm{Tr}_M Z_0^*Z_0\le1$ give $\|\Phi_{O\zeta_1}-\Phi_{Z_0}\|_\diamond\le2\sqrt2\,a_1$. Adding up,
\[
\bigl\|\Phi_{O\zeta_1}-\tfrac12\Phi\bigr\|_\diamond\le2u+s_0+2\sqrt2\,a_1+\sqrt d(4+1.5L_g)\gamma+6L_g\gamma/\sqrt d\le \mathcal E .
\]
Evaluating at $\tau_d$ gives $|c_1-\frac12|\le \mathcal E$, so $c_1\ge\frac12-\mathcal E$, and $\|\Psi'-\Phi\|_\diamond\le \mathcal E/c_1+|1/(2c_1)-1|\le2\mathcal E/c_1$. If
$\langle x|\Phi(|x\rangle\langle y|)|y\rangle=0$, then $|\langle x|\Phi_{O\zeta_1}(|x\rangle\langle y|)|y\rangle|\le\|\Phi_{O\zeta_1}-\frac12\Phi\|_\diamond\le \mathcal E$, and dividing by
$c_1$ gives the last bound.
\end{proof}

\begin{proof}[Proof of Proposition~\ref{prop:lazy}(b)]
Let a service of $\Theta_+$ for $n$ rounds have error $\epsilon\le\frac1{16}$ and $P+C\le B$ with $\log n\le B$. Theorem~\ref{thm:extraction},
applied to $\Theta_+$ on $M_{1746}$, gives a round with bath $\mathbb C^M$, $M=2^Q$, $u=\ddia(\Psi_t,\Theta_+)\le\frac1{15}$ and
$\nu_t=\ell_t+\Delta S_t\le t_n\le\frac{16}{15}\cdot\frac{B+1+\log n}n\le\frac{3.2\,B}n$, as in the proof of Theorem~\ref{thm:main-channel}. The normalization
step of the proof of Lemma~\ref{lem:weighted-extraction}, which uses only that $U_t$ is a unitary on $\mathbb C^{2d}\otimes\mathbb C^M$, replaces $U_t$ by
$(1\otimes W)U_t$ without changing $\Psi_t$ or its leakage, after which $\epsilon_S^2=\|[U_t,1\otimes\sqrt{\rho_t}\,]\|_F^2\le4(2d)\Delta S_t\ln2$; so
$\gamma^2\le2(2d\ell_t+8d\ln2\,\Delta S_t)\le12d\,\nu_t$. Apply Lemma~\ref{lem:block-extraction} with $s_0=\frac1{100}$ and $d=873$:
$\mathcal E\le\frac2{15}+\frac1{100}+14K_ad\sqrt{\nu_t}\le\frac3{14}$ as soon as $14K_ad\sqrt{3.2B/n}\le0.07$. Then $c_1\ge\frac12-\mathcal E\ge\frac27$, and the anchors $0$
and $z=z(j)$ satisfy $\langle z|\Phi(|z\rangle\langle0|)|0\rangle=\tau(\lambda([j]))=0$, so the round $\Psi'$ satisfies~\eqref{eq:coherence}. By
Lemma~\ref{lem:commutator-form} it may be used in Corollary~\ref{cor:amplification} with
\[
\nu'=\ell_\Phi(\Psi')+\frac{\|[U',1\otimes\sqrt{\rho'}]\|_F^2}{4d\ln2}\le\Bigl(4+\frac2{\ln2}\Bigr)\frac{K_a^2\gamma^2}{d\,c_1}\le290K_a^2\,\nu_t\le928K_a^2\,\frac Bn .
\]
Let $C_R$ and $y_0$ be the constants of Corollary~\ref{cor:amplification}(a) for $R_\Phi$ and the transport of Section~\ref{sec:lamp-channel}, and put
$y_n=n/(928K_a^2C_RB)$, so that $C_R\nu'y_n\le1$ and $y_n\le n$. With
$c_0=\min\{1/(928K_a^2C_Ry_0),\ (0.07/(14K_ad))^2/3.2\}$ and $B\le c_0n$, the corollary gives
\[
Q=\log M\ \ge\ S(\rho')\ \ge\ \frac{y_n}{32(\log y_n)^{2\kappa}}\ \ge\ \frac1{29696\,K_a^2C_R}\cdot\frac n{B(\log n)^{2\kappa}} .
\]
For the second statement take $B=\max(1,A_0)\log n$.
\end{proof}

\begin{proof}[Proof of Proposition~\ref{prop:lazy}(c)]
Let $E$, $K_G$ and a $(1-r^{-1})$-expansive $r^{-1}$-morphism $f:E\to\Sym(\Omega)$ be as in Proposition~\ref{prop:sofic-upper}, with its unitary
$U_f=\sum_ac_aE_a\otimes P_{\pi_a}$ and its set $\Omega_0$ of good points, at which $\mathrm{Tr}_\Omega U_f(w\otimes|\omega\rangle\langle\omega|)U_f^*=\Phi(w)$. Let
$\Omega_1$ be the set of $\omega$ with $(f([x])f([y]))^{-1}(\omega)=f([z])^{-1}(\omega)$ for every triangle $(x,y,z)$ and $f(g)^{-1}(\omega)\ne f(h)^{-1}(\omega)$ for
distinct labels $g,h$. As $U_f^*(\psi\otimes|\omega\rangle)=\sum_a\overline{c_a}E_a^*\psi\otimes|\pi_a^{-1}(\omega)\rangle$, the computation of
Proposition~\ref{prop:sofic-upper} gives $\mathrm{Tr}_\Omega U_f^*(x\otimes|\omega\rangle\langle\omega|)U_f=\Phi^*(x)$ for $\omega\in\Omega_1$. A point outside $\Omega_1$
is the image under the bijection $f([z])$ of a point where $f([x])f([y])$ and $f([z])$ differ, or under $f(g)$ of a point where $f(g)$ and $f(h)$
agree; so $|\Omega\setminus\Omega_1|\le K_GM/r$, and $\Omega_*=\Omega_0\cap\Omega_1$ has $|\Omega_*|\ge(1-2K_G/r)M$.

Let the bath be a coin qubit, a phase qubit carrying $z=\mathrm{diag}(1,i)$, and $\mathbb C^\Omega$, and put
$W_f=|0\rangle\langle0|\otimes1+|1\rangle\langle1|\otimes V_f$ with $V_f=\bigl(\begin{smallmatrix}0&U_f\otimes z\\U_f^*\otimes z^*&0\end{smallmatrix}\bigr)$. For every
state $\beta$ that is a mixture of the $|\omega\rangle\langle\omega|$, $\omega\in\Omega_*$, the bath state $\tau_2\otimes\tau_2\otimes\beta$ makes $W_f$ apply
exactly $\Theta_+$. The coin splits the channel into $\frac12\mathrm{id}$ and $\frac12$ the channel of $V_f$, the phase qubit removes the off-diagonal
blocks since $\tau(z^2)=0$, and the diagonal blocks are $\Phi(w)$ and $\Phi^*(x)$ at the points of $\Omega_*$. Now repeat the exact service of
Proposition~\ref{prop:sofic-upper}(c) with $K_G$ replaced by $2K_G$ and two more qubits, which are imported maximally mixed. With
$\eta=B/(4n)$ and $r=2K_G/\eta$ the imported register is maximally mixed on $\mathrm{span}\{|\omega\rangle:\omega\in\Omega_*\}\otimes\mathbb C^{k_0}$ and on the
two qubits, its deficit is below $B/n$, every round applies exactly $\Theta_+$, and $Q=\lceil\log\sigma_E(8K_Gn/B)+\log(4n/B)\rceil+2$. The bound
$\log\sigma_E(R)\le C_ER\log R$ of Theorem~\ref{thm:main-group}(b) gives $Q\le K(n/B)\log(2n/B)$. With $\beta$ maximally mixed on all of $\Omega$,
the same unitary gives a finite tracial factorization of $(1-\vartheta)\Theta_++\vartheta\,\Upsilon$ for some channel $\Upsilon$ and some $\vartheta\le2K_G/r$, for every $r$; hence
$\Theta_+\in\cK_{1746}$.
\end{proof}

\begin{remark}[Exact finite factorizations of $\Theta_+$]\label{rem:lazy-exact}
$\Theta_+$ is an exact autonomous value but not an exact finite factor. Indeed, if $\Theta_+(X)=(\mathrm{id}\otimes\tau_m)[W(X\otimes1)W^*]$ for a unitary
$W$ on $\mathbb C^{2d}\otimes\mathbb C^m$, the proof of Lemma~\ref{lem:block-extraction} with $\rho=\tau_m$ has $\ell=0$ and $\epsilon_S=0$: the diagonal blocks of
$W$ are both $1\otimes\alpha$, its upper right block $B$ lies in $\mathcal S_\Phi\otimes M_m$, $B^*B=1\otimes R$ and $BB^*=1\otimes R'$. Then
$B(1\otimes R)=(1\otimes R')B$, so for every eigenvalue $\mu>0$ of $R$ the operator $\mu^{-1/2}B$ maps $\mathbb C^d\otimes\ker(R-\mu)$ unitarily onto
$\mathbb C^d\otimes\ker(R'-\mu)$, and the upper left block of $\Theta_+(e_{11}\otimes w)$ writes $\Phi$ as the convex combination
$\sum_\mu\frac{2\mu m_\mu}m\,\Phi_\mu$ of exact factors with baths of dimension $m_\mu=\dim\ker(R-\mu)\le m$. For the channel of
Theorem~\ref{thm:main-channel} this is impossible: such a combination is a single round with $\ell=\Delta S=0$, a direct sum of flat baths, and
Theorem~\ref{thm:weighted} with $\nu=0$ would give it entropy at least $N/16$ for every $N$.
\end{remark}

\begin{remark}[A lower exchange]\label{rem:lazy-exchange}
At an anchor $z$ of the gadget, $\Phi(|z\rangle\langle z|)=|z\rangle\langle z|$, so $\Theta_+$ maps $|1\rangle\langle1|\otimes|z\rangle\langle z|$ to the flat state
$\frac12(|0\rangle\langle0|+|1\rangle\langle1|)\otimes|z\rangle\langle z|$ of rank two. Proposition~\ref{prop:purity-budget} with $r=2$ therefore gives $P+C<3Q+6$ for
services of $\Theta_+$ with error at most $1/16$ and exchange $J\le n/2$. Proposition~\ref{prop:lazy}(b), used as in the proof of
Theorem~\ref{thm:main-channel}(b) with $B=\max(3Q+6,\log n)$, then gives $Q\ge c\sqrt n/(\log n)^\kappa$ for such services, whatever their purity. If
$B>c_0n$, then already $Q\ge(c_0n-6)/3$.
\end{remark}

\begin{proof}[Proof of Proposition~\ref{prop:compositions}]
Let $\mathcal F$ be the set of labels, and write $A_g=\sum_{[w_a]=g}c_aE_a$, so that $\Phi_G\Phi_G^*(X)=\sum_{g,h}A_gA_h^*XA_hA_g^*$. A product $A_gA_h^*$
is a sum of terms $c_a\overline{c_b}E_aE_b^*$ over slots $a,b$ in a common column, so it vanishes unless $g$ and $h$ label two slots of one column. In a
triangle block with labels $1,y$ in the first row and $x,z$ in the second, where $z=xy$, the columns carry the labels $\{1,x\}$ and $\{y,z\}$; so
$A_gA_h^*\ne0$ forces $hg^{-1}\in K=\{1\}\cup\{x^{\pm1}\}$, the union over the triangles. For $k\in K$ let $\pi_k$ be a permutation of $\mathcal F$ that
extends the partial map $g\mapsto kg$, defined when $g,kg\in\mathcal F$, with $\pi_1=\mathrm{id}$ and $\pi_{k^{-1}}=\pi_k^{-1}$. For $k\ne k^{-1}$ extend
one of the two maps and invert it; for an involution $k$ the partial map is an involution on its domain, which we extend by the identity.

The bath is $r_*+|K|-1\le3r_*-2$ qubits. On the first $r_*$, indexed by $\mathcal F$, let $V_g$ be the Pauli matrix $Z$ on the qubit $g$, so that
$\tau(V_gV_h)=\delta_{gh}$, and let $P_k$ permute these qubits by $\pi_k$, so that $P_kV_gP_k^*=V_{kg}$ whenever $g,kg\in\mathcal F$. On the other
qubits choose diagonal unitaries $\chi_k$ with $\chi_1=I$, $\chi_{k^{-1}}=\chi_k^*$ and $\tau(\chi_k^*\chi_l)=\delta_{kl}$: for an involution, $Z$ on one
qubit, and for a pair $k\ne k^{-1}$, $\mathrm{diag}(1,i,-1,-i)$ and its adjoint on two qubits, different pairs on different qubits. Put
$R_k=P_k\otimes\chi_k$, so that $R_{k^{-1}}=R_k^*$ and $R_kV_gR_k^*=V_{kg}$ for $g,kg\in\mathcal F$, and put $C_{g,h}=V_gR_{hg^{-1}}^*$. These are
orthonormal for the normalized trace: if $hg^{-1}\ne h'g'^{-1}$ the factor $\chi$ has trace zero, and otherwise
$\tau(C_{g,h}^*C_{g',h'})=\tau(R_kV_gV_{g'}R_k^*)=\delta_{gg'}$, which then forces $h=h'$. In the same way
$\tau(C_{g,h}C_{g',h'}^*)=\tau(V_gV_{g'})=\delta_{gg'}$ when $hg^{-1}=h'g'^{-1}$, and it vanishes otherwise.

Let $W=\sum A_gA_h^*\otimes C_{g,h}$, summed over the pairs with $hg^{-1}\in K$; the other products $A_gA_h^*$ vanish. On the block of a symbol
with label $g$, including the first block, where $g=1$, it is $V_g$. On a triangle block, with $H$ the Hadamard matrix and using
$R_xV_1R_x^*=V_x$ and $R_xV_yR_x^*=V_z$,
\[
W=\frac12\begin{pmatrix}V_1+V_y&(V_1-V_y)R_x^*\\ R_x(V_1-V_y)&V_x+V_z\end{pmatrix}
=\begin{pmatrix}I&0\\0&R_x\end{pmatrix}(H\otimes I)\begin{pmatrix}V_1&0\\0&V_y\end{pmatrix}(H\otimes I)\begin{pmatrix}I&0\\0&R_x^*\end{pmatrix},
\]
a self-adjoint unitary; triangles with coinciding labels are covered by the same computation. So $W$ is unitary, and by orthonormality of the
$C_{g,h}$, $(\mathrm{id}\otimes\tau)[W(X\otimes I)W^*]=\sum_{g,h}A_gA_h^*XA_hA_g^*=\Phi_G\Phi_G^*(X)$. For $\Phi_G^*\Phi_G$ use right translations
$g\mapsto gk$ with $k\in\{1\}\cup\{y^{\pm1}\}$, the operators $C_{g,h}=V_gR_{g^{-1}h}^*$, and rows instead of columns; on a triangle block the rows carry
$\{1,y\}$ and $\{x,z\}$, and the same computation applies with $x$ and $y$ exchanged. For $\Theta^2$, one further maximally mixed qubit dephases the
two blocks, and controlled on the block the two factorizations act on a common bath. Refreshing the whole bath in every round gives the services.
\end{proof}
\section{Finer constraints on autonomous limits}\label{sec:physics-finer}

This appendix complements Section~\ref{sec:physics-autonomous}: a return inequality that autonomous limits obey, an exact finite
factor that it excludes, and the cheapness of adjoint compositions of gadget channels.

The limit $\overline\Phi$ of Proposition~\ref{prop:autonomous} is more than self-adjoint and positive. It is a compression of a conditional
expectation, and this ties the decay of a coherence to the return probability of a state.

\begin{proposition}[A return inequality]\label{prop:return}
Let $T=E\circ E_{\mathcal N}\circ\iota$ be as in Proposition~\ref{prop:autonomous}, on $M_D$. If $T(e_{00})=e_{00}$ and $T(e_{0j})=a\,e_{0j}$ for some
$j\ne0$ and $0\le a\le1$, then
\[
\langle j|T(e_{jj})|j\rangle\ \ge\ a^2+\frac{(1-a)^2}{D-1}.
\]
\end{proposition}

\begin{proof}
Write $P=E_{\mathcal N}$, the orthogonal projection onto $L^2(\mathcal N)$, and identify $x\in M_D$ with $\iota(x)$; then
$\mathrm{tr}_D(x^*T(y))=\langle x,Py\rangle$ with $\langle x,y\rangle=\tau(x^*y)$. Put $p_0=e_{00}$, $p_j=e_{jj}$ and $x=e_{0j}$. Since
$\|Pp_0\|_2^2=\langle p_0,T(p_0)\rangle=\|p_0\|_2^2$, we have $Pp_0=p_0$, so $p_0\in\mathcal N$. Let $y=P(x)$, $X=P(p_j)$ and $B=y^*y$. As $P$ is
$\mathcal N$-bimodular, $p_0y=P(p_0x)=y$, $yp_0=P(xp_0)=0$ and $p_0X=P(p_0p_j)=0$; so $yy^*$ is supported in $p_0$, and $B$ and $X$ in $1-p_0$. The
Schwarz inequality for conditional expectations gives $B\le P(x^*x)=X$. Moreover $\tau(X)=\tau(p_j)=1/D$ and $\tau(B)=\|Px\|_2^2=\langle x,T(x)\rangle=a/D$,
so $R=X-B\ge0$ has $\tau(R)=(1-a)/D$. By the Cauchy--Schwarz inequality, $\tau(B^2)=\tau((yy^*)^2)\ge\tau(yy^*)^2/\tau(p_0)=a^2/D$ and
$\tau(R^2)\ge\tau(R)^2/\tau(1-p_0)=(1-a)^2/(D(D-1))$, and $\tau(BR)\ge0$. Hence
\[
\frac1D\,\langle j|T(e_{jj})|j\rangle=\langle p_j,Pp_j\rangle=\tau(X^2)=\tau(B^2)+2\tau(BR)+\tau(R^2)\ \ge\ \frac{a^2}D+\frac{(1-a)^2}{D(D-1)} .\qedhere
\]
\end{proof}

\begin{remark}[Self-adjoint and positive is not enough]\label{rem:qutrit}
Haagerup and Musat give self-adjoint positive maps that are not such compressions because they are not
factorizable~\cite[Theorem~5.4]{HaagerupMusat2011}. One exact finite factor is excluded all the same. Let $a=\sqrt3/2$ and $b=\frac12$, and let
$\Lambda$ on $M_3$ have the Kraus operators $\mathrm{diag}(1,a,a)$, $b\,e_{12}$ and $-b\,e_{21}$. With
$\omega=e^{2\pi i/3}$ and $W_k=1\oplus\bigl(\begin{smallmatrix}a&b\omega^k\\-b\omega^{-k}&a\end{smallmatrix}\bigr)$ one has
$\Lambda=\frac13\sum_{k=0}^2\mathrm{Ad}\,W_k$, so $\Lambda$ has a finite tracial factorization with bath dimension three. It is self-adjoint, and positive
for the Hilbert--Schmidt inner product, with eigenvalues between $\frac12$ and $1$. But $\Lambda(e_{00})=e_{00}$, $\Lambda(e_{01})=a\,e_{01}$ and
$\langle1|\Lambda(e_{11})|1\rangle=a^2$, which violates Proposition~\ref{prop:return}. Hence $\Lambda$ is not the long-time limit of any autonomous
dynamics with a tracial bath.
\end{remark}

Composing a gadget channel with its adjoint erases the multiplication of the group. In a triangle block the unitarity of the gadget forces the
coefficient of the product label to be the product of the coefficients, and this is what the one-round theorem reads. In $\Phi\Phi^*$ a triangle only
requires a permutation $\pi_x$ of the labels that carries $1$ to $x$ and $y$ to $xy$, and no relation between the permutations $\pi_x$ for
different $x$ is needed.

\begin{proposition}[Adjoint compositions of gadgets are exact finite factors]\label{prop:compositions}
Let $\Phi_G$ be a gadget channel with $r_*$ labels. Then $\Phi_G\Phi_G^*$ and $\Phi_G^*\Phi_G$ have finite tracial factorizations with baths of at
most $3r_*-2$ qubits. Consequently, for every $n$, they have exact causal services with $Q\le3r_*-2$, $P=C=0$ and $J=(n-1)Q$, and so does
$\Theta^2$ in Remark~\ref{rem:doubling}, with one more qubit.
\end{proposition}

The proof is in Appendix~\ref{app:lazy}. The exchange $(n-1)Q$ of these services is far above the rank rate, as for the services of
Proposition~\ref{prop:sofic-upper}. At its own rank rate $\Phi\Phi^*$ is cheap too: as an exact finite factor with a bath of at most $3r_*-2$
qubits, it has services with $P=C=0$ and memory $O(r_*+\log(n/\epsilon))$ by~\cite[Theorem~8.6]{DouglasMghirbiB}. An autonomous dynamics with a tracial bath therefore cannot be expensive through $\Phi\Phi^*$. It can be through a lazy version
of the doubling (Section~\ref{sec:physics-lazy}).
\section{Hyperlinearity transfers from $F$ to $\mathcal L$}\label{app:hyperlinear}

Proposition~\ref{prop:thompson-hyperlinear} states that $\mathcal L=\bigl(\bigoplus_{\mathbb D}S_3\bigr)\rtimes F$ is hyperlinear exactly when
$F$ is. This appendix proves it by embedding $L(\mathcal L)$ into an ultrapower $L(F)^\omega$: a copy of $S_3$ is built from a localized
sequence of $\mathrm{II}_1$ subfactors that $F_{1/2}$ asymptotically fixes. Nothing else in the paper depends on this appendix.

\begin{proof}[Proof of Proposition~\ref{prop:thompson-hyperlinear}]
Subgroups of hyperlinear groups are hyperlinear, so suppose that $F$ is hyperlinear. A countable group $\Gamma$ is hyperlinear if and only if
its von Neumann algebra $L(\Gamma)$ has a trace-preserving embedding into an ultrapower $\mathcal R^\omega$ of the hyperfinite
$\mathrm{II}_1$ factor~\cite{CapraroLupini2015}. Then every separable von Neumann subalgebra of $L(\Gamma)^\omega$ embeds in
$\mathcal R^\omega$ as well, by a diagonal argument. It therefore suffices to embed $L(\mathcal L)$ into $L(F)^\omega$ preserving the trace,
for a free ultrafilter $\omega$ on $\mathbb N$. We write $\tau$ for the canonical traces and $F_J$ for the subgroup of elements supported in an
interval $J$. If $J,J'$ are disjoint, $L(F_J)$ and $L(F_{J'})$ commute and $\tau(xy)=\tau(x)\tau(y)$ for $x\in L(F_J)$ and $y\in L(F_{J'})$,
since $F_J$ and $F_{J'}$ generate their direct product. If $f,f'\in F$ agree on $J$, then $\lambda(f)x\lambda(f)^{-1}=\lambda(f')x\lambda(f')^{-1}$
for $x\in L(F_J)$, because $f'^{-1}f$ fixes $J$ pointwise and therefore commutes with $F_J$.

\emph{A localized copy of $S_3$.} Choose $\sigma\in F$ supported in $[\frac12,\frac34]$, with $\sigma(x)<x$ on $(\frac12,\frac34)$ and
$\sigma(x)=\frac12+\frac12(x-\frac12)$ near $\frac12$, and a dyadic $x_0\in(\frac12,\frac34)$. The intervals
$J_k=\sigma^k((\sigma(x_0),x_0))$, $k\in\mathbb Z$, are disjoint and accumulate at $\frac12$ as $k\to\infty$. For a nonempty open interval $J$ with
dyadic endpoints, $F_J\cong F$ has infinite conjugacy classes, since a nontrivial element has conjugates with infinitely many different
supports, so $L(F_J)$ is a $\mathrm{II}_1$ factor. In $L(F_{J_0})$ choose a unital copy of $M_6$ and, in its relative commutant, which is again a
$\mathrm{II}_1$ factor, a self-adjoint $Z_0$ with uniform distribution on $[0,1]$. Let $v_0:S_3\to U(M_6)$ be the regular representation, so
that $\tau(v_0(g))=0$ for $g\ne1$, and put $v_k=\lambda(\sigma)^kv_0\lambda(\sigma)^{-k}$ and $Z_k=\lambda(\sigma)^kZ_0\lambda(\sigma)^{-k}$,
which lie in $L(F_{J_k})$. By the factorization of traces, the $Z_k$ are independent and uniform. For a finite $K\subset\mathbb Z$ and $k\in K$,
let $p^K_k$ be the spectral projection of the event that $Z_k$ is the largest of the $Z_j$, $j\in K$. Ties have trace zero, so these projections
are orthogonal and sum to $1$, and they commute with every $v_j(g)$. Hence
\[
V_K(g)=\sum_{k\in K}p^K_kv_k(g)
\]
is a representation of $S_3$. Since $v_k(g)$ lies in a copy of $M_6$ whose relative commutant contains $Z_k$, and the other $Z_j$ lie in
algebras with factorized trace, $\tau(p^K_kv_k(g))=\tau(p^K_k)\tau(v_0(g))$; so $\tau(V_K(g))=0$ for $g\ne1$. If $K$ is an interval of $n$ integers and $|q|\le n$,
the two selected indices agree unless the largest $Z_j$ over $K\cup(K+q)$ has its index outside $K\cap(K+q)$, an event of trace
$2|q|/(n+|q|)$. On its complement $V_K(g)=V_{K+q}(g)$. Hence $\|V_{K+q}(g)-V_K(g)\|_2^2\le8|q|/(n+|q|)$.

\emph{Invariance.} Put $V_n=V_{\{n,\dots,2n-1\}}$, and let $h\in F_{1/2}$. Near $\frac12$ on the right, $h(x)=\frac12+2^m(x-\frac12)$ for an
integer $m$. Since $\sigma$ has slope $\frac12$ on the right of $\frac12$, $h$ agrees with $\sigma^{-m}$ on $J_k$ for all large $k$. Conjugation by $\lambda(h)$ therefore carries $v_k$ and $Z_k$ to
$v_{k-m}$ and $Z_{k-m}$ for large $k$, and $V_n$ to $V_{\{n,\dots,2n-1\}-m}$ for large $n$, so $\|\lambda(h)V_n(g)\lambda(h)^{-1}-V_n(g)\|_2\to0$.
Thus $V(g)=(V_n(g))_\omega$ is a representation of $S_3$ in $L(F)^\omega$ with $\tau(V(g))=0$ for $g\ne1$, and it commutes with $\lambda(F_{1/2})$.

\emph{The embedding.} For $x\in\mathbb D$ choose $f_x\in F$ with $f_x(\frac12)=x$ and put $V_x=\lambda(f_x)V\lambda(f_x)^{-1}$. By the
invariance of $V$ under $\lambda(F_{1/2})$, $V_x$ does not depend on the choice of $f_x$, and $\lambda(f)V_x\lambda(f)^{-1}=V_{f(x)}$ for $f\in F$.
The representatives $V_n(g)$ lie in $L(F_{I_n})$, where $I_n=\bigcup_{n\le k<2n}J_k$ shrinks to $\frac12$ from the right, so those of $V_x$ lie
in $L(F_{f_x(I_n)})$. For $x\ne y$ the intervals $f_x(I_n)$ and $f_y(I_n)$ are disjoint for large $n$. So copies at distinct points commute, and
the $V_x$ together with $\lambda$ define a homomorphism $\pi:\mathcal L\to U(L(F)^\omega)$. Let $\gamma=\bigl(\prod_{x\in X}\ell_x\bigr)f$, with
$X$ finite, $\ell_x$ in the copy of $S_3$ at $x$ and $f\in F$. If $f=1$ and some $\ell_x\ne1$, the factors lie in commuting algebras with
factorized trace for large $n$, so $\tau(\pi(\gamma))=\prod_x\tau(V(\ell_x))=0$. If $f\ne1$, choose a point moved by $f$ outside $X$. For large
$n$ it lies outside the closure of every $f_x(I_n)$, so $f$ is not in the group $K_n$ generated by the $F_{f_x(I_n)}$. Then $\tau(a\lambda(f))=0$
for every $a\in L(K_n)$, in particular for the representatives of $\prod_x\pi(\ell_x)$. Hence $\tau(\pi(\gamma))=0$ whenever $\gamma\ne1$, and
$\pi$ extends to a trace-preserving embedding of $L(\mathcal L)$ into $L(F)^\omega$.
\end{proof}
\section{The grid host}\label{sec:grid}

This appendix is a sketch: its conclusions are conditional on proofs that we only outline, and they are stated in the indicative with that
understanding. Nothing elsewhere in the paper depends on them. Placing the lamps on configurations of the cells of a grid, rather than of the
points of $Y$, gives a stronger lower bound with the same upper bound. It also gives a nearly linear exponent $r^{1-o(1)}$ from Lee's
exponential Dehn bound alone.

\begin{remark}[The grid host]\label{rem:grid}
Let $I=Y\times Y$, let $T=H\times H$ act on $I$ coordinatewise, and let
\[
\Gamma_2=\Bigl(\bigoplus_{v\in V_2}S_3\Bigr)\rtimes\bigl(V_2\rtimes(\GL_{\rm fin}(I)\rtimes T)\bigr),\qquad V_2=\F_2^{(I)} .
\]
It is locally finite-by-$\Z^4$ and generated by seven elements: the generators $p_1,q_1,p_2,q_2$ of the two copies of $H$, the translation $t$
by the cell $o=((1,1),(1,1))$, the element $c=d\,ab$, where $d$ is the transvection that adds $e_o$ to $e_{((1,2),(1,2))}$, and $a$.
Transvections within a row or a column are not needed as generators, since commutators $[E_{x\leftarrow y},E_{y\leftarrow z}]=E_{x\leftarrow z}$
of transvections of the third kind produce them. A certificate of $77$ relations, of total length $5445$ and maximum length $224$, is listed in
\texttt{anc/relators.txt} and checked by \texttt{anc/relators\_check.py}. It has the families of Table~\ref{tab:certificate},
with two copies of the Houghton relators and their commutation, and with stabilizer and action relations for three kinds of transvection: within
a row, within a column, and neither.

Transport now happens at scale $\sqrt N$. $N$ cells fit in a square of side $\lceil\sqrt N\,\rceil$, so the transporters have length $O(\sqrt N)$,
and the local fillings cost $C(\delta_H(C\sqrt N)+N)$; the term $N$ counts the swaps of letters from the two copies of $H$ in a change of
transporter. The clearing argument of Theorem~\ref{thm:pair-area} then bounds the pair areas by $CN^2(\delta_H(C\sqrt N)+N)$, which is at most
$CN^{2+p_H/2}$ by Proposition~\ref{prop:dehn}(b) and at most $e^{C\sqrt N}$ by Proposition~\ref{prop:dehn}(a). As in the proof of
Theorem~\ref{thm:main-group}, the chunk $E_2$ of this certificate therefore has
\[
\log\sigma_{E_2}(r)\ \ge\ c\,\frac{r}{(\log r)^{2+p_H/2}},\qquad 2+\tfrac{p_H}2<7.38,
\]
and, using only Lee's bound, $\log\sigma_{E_2}(r)\ge c\,r\,e^{-C\sqrt{\log r}}$; unitary models need $\log D\ge c_\Delta e^{-2}/(\log(1/e))^{4+p_H}$.
In the other direction, hashing the configurations directly, as in Proposition~\ref{prop:upper}, needs a configuration space of dimension of
order $r^2$. Hashing the cells first avoids this. A cell $(y_1,y_2)$ gets the label $a_1(f_1^{-1}y_1)+a_2(f_2^{-1}y_2)\in\F_q$. Here the $f_j$ are
positions in finite F{\o}lner charts of $H$ of size $e^{O(r\log r)}$, inverses of right F{\o}lner sets because the charts move by left
multiplication, and the $a_j$ are random tables on $O(r)$ points. The label is covariant under the charts, and it separates the finitely many
cells of a chunk with probability $1-O(1/q)$. The configurations are then compressed linearly to $\F_2^{\F_q}$ and hashed affinely as in
Proposition~\ref{prop:upper}. With $q$ of order $\ell^2r$, for a chunk of words of length at most $\ell$, and the other parameters of order $r$,
this gives $\log\sigma_{E'}(r)\le C_{E'}r\log(2+r)$ for every chunk $E'$ of $\Gamma_2$. The same holds for the $d$-dimensional grid host, built
in the same way on the cells of $Y^d$ with $H^d$ acting, for every fixed $d$. Hence $\sigma_{E_2}(r)=\exp(r^{1+o(1)})$ as well, and no
fixed-dimensional grid host exceeds $\exp(r^{1+o(1)})$. The proofs follow those of Sections~\ref{sec:fillings}--\ref{sec:profile} with the
changes indicated.
\end{remark}

\begin{remark}[The gadget channel of the grid host]\label{rem:grid-channel}
Adding $b$ and $j$ and the four words of Section~\ref{sec:lamp-channel} to the certificate of $\Gamma_2$ gives $s=9$, $r_0=81$ and $\Lambda=5462$,
so the gadget channel $\Phi_2$ of $\Gamma_2$ acts on $M_d$ with
\[
d=1+4\cdot9+3\cdot5462-4\cdot81=16099,
\]
and has $m=224$. Of its $81$ words, $78$ have length at least three, and the next-to-last prefix of each of them represents
the inverse of its last letter; $a=a^{-1}$, $b=b^{-1}$ and $j=ab$ give three more coincidences. Its Choi rank, the number of
distinct labels (Section~\ref{sec:lamp-channel}), is therefore at most $1+18+5300-78-3=5238$. With the transport of Remark~\ref{rem:grid}, words of length $O(N^{3/2})$ and pair areas
$O(N^2(\delta_H(C\sqrt N)+N))$, the hypothesis of Corollary~\ref{cor:amplification}(a) holds with the exponent $2+p_H/2$, and part~(b) of that
corollary holds with Lee's exponential Dehn bound~\cite{Lee2012} alone. The proof of Theorem~\ref{thm:main-channel} then gives three bounds for
$\Phi_2$. At purity $O(\log n)$ the memory is at least $c\,n/(\log n)^{p_H+5}$, where $p_H+5<15.75$, and at least $n\,2^{-C\sqrt{\log n}}$ using
only Lee's bound. At the rank rate it is at least $c\sqrt n/(\log n)^{2+p_H/2}$, whatever the purity. By the two-stage hashing of
Remark~\ref{rem:grid}, $\log\sigma(r)=O(r\log r)$ for $\Gamma_2$ as well, which gives an upper bound of order $n\log n$. For $\Gamma$ itself,
Lee's bound alone gives pair areas $2^{O(N)}$, and the argument gives only memory $n^{c}$ for some $c>0$.
\end{remark}
\clearpage
\section{Results from the companion papers, and machine checks}\label{app:provenance}

Table~\ref{tab:dependencies} lists every result of the companion papers~\cite{DouglasMghirbiA,DouglasMghirbiB} on which an argument here
depends, and where it is used. Results of the companions mentioned only for comparison are cited in the text.

\begin{table}[ht]
\centering\small
\begin{tabular}{@{}p{0.47\textwidth}p{0.2\textwidth}p{0.25\textwidth}@{}}
\toprule
result & source & used in\\
\midrule
counting criterion (proof in App.~\ref{app:imported}) & \cite[Thm.~8.1]{DouglasMghirbiA} & Thm.~\ref{thm:criterion}, every lower bound\\
transport by the powers of one element & \cite[Cor.~8.2]{DouglasMghirbiA} & Section~\ref{sec:criterion}\\
Johnson's presentation of $H_3$ and its conventions & \cite[Section~2]{DouglasMghirbiA} & Appendix~\ref{sec:houghton}\\
far-commutator area and polynomial Dehn bound for $H_3$ & \cite[Thms.~3.2, 3.3]{DouglasMghirbiA} & Prop.~\ref{prop:dehn}\\
quadratic lower bound along a ray & \cite[Prop.~3.4]{DouglasMghirbiA} & Sections~\ref{sec:intro}, \ref{sec:criterion}, \ref{sec:ladder}\\
bounded or at least $2r$ (sharpening Cornulier's Fact~3.5); embeddings in finite groups & \cite[Lemma~2.1]{DouglasMghirbiA} & Section~\ref{sec:criterion}\\
superpolynomial profile of $H_3$ & \cite[Thm.~1.1]{DouglasMghirbiA} & Section~\ref{sec:thompson}\\
projection products, exact $S_3$ rounding, sector bounds & \cite[Section~4]{DouglasMghirbiA} & Lemma~\ref{lem:projection-product}, App.~\ref{app:imported}\\
\midrule
causal services, exchange, purity and error & \cite[Def.~2.3, Section~3]{DouglasMghirbiB} & Section~\ref{sec:channel}\\
gadget channels: Choi rank $r_*$, flat probe, entropy $\log r_*$ & \cite[Prop.~9.1]{DouglasMghirbiB} & Section~\ref{sec:channel-setting}\\
the rank rate is the optimal exchange rate of gadget channels in $\cK_d$ & \cite[Props.~3.1, 9.1]{DouglasMghirbiB} & Thm.~B(b)\\
sequential extraction (proof in App.~\ref{app:imported}) & \cite[Thm.~5.2]{DouglasMghirbiB} & Thm.~\ref{thm:extraction}\\
purity budget $14Q+3$ at the rank rate (sharpened here) & \cite[Thm.~5.8]{DouglasMghirbiB} & Prop.~\ref{prop:purity-budget}\\
entropy variance & \cite[Lemma~5.9]{DouglasMghirbiB} & Lemma~\ref{lem:surprisal}\\
rolling-reservoir compiler & \cite[Appendix~A, Lemmas~A.1, A.4]{DouglasMghirbiB} & Thm.~\ref{thm:streaming}\\
channel profiles are group profiles; hyperlinear groups give channels in $\cK_d$ & \cite[Thm.~9.2]{DouglasMghirbiB} & Section~\ref{sec:channel}\\
distances between channels and the Choi profile & \cite[Def.~2.1, Section~2]{DouglasMghirbiB} & Cor.~\ref{cor:profile}\\
exact factors are cheap at their rank rate & \cite[Thm.~8.6]{DouglasMghirbiB} & Appendix~\ref{sec:physics-finer}\\
one-round theorem for the Houghton gadget (generalised here) & \cite[Thm.~5.5, Appendix~B]{DouglasMghirbiB} & Thm.~\ref{thm:weighted}, Appendix~\ref{sec:weighted-proof}\\
weighted norms; charged transport and weighted extraction for the Houghton gadget (generalised here) & \cite[Appendix~B, Lemmas~B.2, B.4]{DouglasMghirbiB} & Lemmas~\ref{lem:weighted-extraction}, \ref{lem:charged-area}\\
weighted $S_3$ rounding & \cite[Lemma~B.3]{DouglasMghirbiB} & Lemma~\ref{lem:rounding}\\
weighted sector bound (proved here by sequential measurement) & \cite[Thm.~B.1]{DouglasMghirbiB} & Thm.~\ref{thm:sector}\\
\bottomrule
\end{tabular}
\caption{Results of the companion papers used in this paper (numbering of their version 1.1.2).}
\label{tab:dependencies}
\end{table}

\subsection{Machine checks}\label{app:machine-checks}

The $25$ freely reduced words are listed in the file \texttt{anc/relators1d.txt}. The script \texttt{anc/relators1d\_check.py}
rebuilds each word from the abbreviations and checks it, independently of Lemma~\ref{lem:validity}, in the faithful action of $\Gamma$ on
$V\times S_3$. For a word $w$ of length $l$ it checks that the exponent sums of $p$ and $q$ vanish; that $w$ fixes the empty configuration and
every configuration $\{y\}$ with $y$ of height at most $l+2$ or far out on a ray; and that it fixes $(v,s)$ for every $s\in S_3$ at every
configuration $v$ at which one of its lamp letters acts. These checks are exhaustive: while $w$ is applied, a point of height greater than $l+2$
never reaches the points $1$ and $2$ and is moved only by translations along its ray, so the checks determine the affine part of $w$, and then
the lamps of $w$ sit at the configurations just listed. The script also checks random configurations, that the abbreviations act as stated, and,
as a negative control, that deleting any single letter from a sample of eleven relations produces a word that acts nontrivially.

The script \texttt{anc/embedding3v\_check.py} composes the prefix tables exactly and checks the identities of the proof of
Proposition~\ref{prop:3V} and
that the images of $p,q,t,c,a$ satisfy the $25$ relations of Table~\ref{tab:certificate}.

The script \texttt{anc/choi\_rank\_check.py} computes the Choi rank of the lamp channel. It sorts the identity and the $276$ symbols of
the gadget of $R_\Phi$ into classes of equal elements of $\Gamma$ twice: once with the faithful action on configurations and the triviality
test just described, and once with the exact images in $3V$. The two partitions agree and have $130$ classes, so $r_*=130$.

The fillings in the proof of Theorem~\ref{thm:dyadic} have been generated and checked by free reduction for $k\le64$; the largest ratio of cell count to $k^2$ among
them is $15/4$, at $k=2$.



\end{document}